\documentclass[a4paper,USenglish,numberwithinsect]{lipics-v2021}
\newcommand{\dom}{\mathrm{dom}}

\newcommand{\NN}{\mathbb{N}}
\newcommand{\UR}{\mathrm{UR}}
\newcommand{\LL}{\mathrm{L}}
\newcommand{\RR}{\mathrm{R}}
\newcommand{\PQE}{\mathrm{PQE}}
\usepackage{colonequals}
\usepackage[bibliography=common]{apxproof}
\usepackage{bm}
\usepackage{tikz}
\nolinenumbers

\newtheoremrep{theorem}{Theorem}
\newtheoremrep{proposition}[theorem]{Proposition}
\newtheoremrep{lemma}[theorem]{Lemma}
\newtheoremrep{claim}[theorem]{Claim}
\newtheoremrep{observation}[theorem]{Observation}
\counterwithin{theorem}{section}


\Copyright{Antoine Amarilli}

\title{Uniform Reliability for Unbounded Homomorphism-Closed Graph Queries}

\author{Antoine Amarilli}{LTCI, Télécom Paris, Institut Polytechnique de Paris,
France}{}{https://orcid.org/0000-0002-7977-4441}{Partially supported by the ANR
project ANR-18-CE23-0003-02 (``CQFD'').}
\authorrunning{Antoine Amarilli}

\acknowledgements{I am grateful to Mikaël Monet, Charles Paperman, and Martin
Retaux for helpful discussions about this research. Thanks to the reviewers for
their helpful feedback.}

\EventEditors{Floris Geerts and Brecht Vandevoort}
\EventNoEds{2}
\EventLongTitle{26th International Conference on Database Theory (ICDT 2023)}
\EventShortTitle{ICDT 2023}
\EventAcronym{ICDT}
\EventYear{2023}
\EventDate{March 28--31, 2023}
\EventLocation{Ioannina, Greece}
\EventLogo{}
\SeriesVolume{255}
\ArticleNo{12}

\relatedversiondetails[cite={amarilli2023uniform}]{Full version}{https://arxiv.org/abs/2209.11177}

\hypersetup{
    colorlinks,
    linkcolor={red!50!black},
    citecolor={blue!50!black},
    urlcolor={blue!30!black}
}
\newcommand\restr[2]{{%
  \kern-\nulldelimiterspace %
  #1 %
  _{|#2} %
  }}

\begin{document}

\maketitle

\begin{abstract}
  We study the uniform query reliability problem, which asks, for a fixed
  Boolean query~$Q$, given an instance~$I$, how many subinstances of~$I$
  satisfy~$Q$. Equivalently, this is a restricted case of Boolean query
  evaluation on tuple-independent probabilistic databases where
  all facts must have probability~$1/2$. We focus on graph
  signatures, and on queries closed under homomorphisms. We show that for any
  such query that is \emph{unbounded}, i.e., not equivalent to a union of
  conjunctive queries, the uniform reliability problem is \#P-hard. This
  recaptures the hardness, e.g., of s-t connectedness, which counts how many
  subgraphs of an input graph have a path between a source and a sink.

  This new hardness result on uniform reliability strengthens our earlier hardness result
  on probabilistic query evaluation for unbounded
  homomorphism-closed queries~\cite{amarilli2021dichotomy}. Indeed,
  our earlier proof crucially used facts with probability~$1$, so it did
  not apply to the unweighted case.
  The new proof presented in this paper avoids this; it uses our recent hardness
  result on uniform reliability for non-hierarchical conjunctive queries without
  self-joins~\cite{amarilli2022uniform}, along with new techniques.
\end{abstract}

\section{Introduction}
\label{sec:intro}
A long line of research~\cite{suciu2011probabilistic} has investigated how
to extend relational databases with probability values. The most
common probabilistic model, called \emph{tuple-independent databases} (TID), annotates each fact
of the input database with an independent probability of existence. The
\emph{probabilistic query evaluation} (PQE) problem then asks for
the probability that a fixed Boolean query is true
in the resulting product distribution on possible worlds. 
The PQE problem has been
historically studied for conjunctive queries (CQs) and unions of conjunctive
queries (UCQs). This study led to the dichotomy result
of Dalvi and Suciu~\cite{dalvi2012dichotomy}, which identifies a class of \emph{safe UCQs} for which
the problem can be solved in PTIME:

\begin{theorem}[\cite{dalvi2012dichotomy}]
  Let $Q$ be a UCQ. Consider the PQE problem for~$Q$ which asks, given a
  TID~$I$, to compute
  the probability that $Q$ holds on~$I$. This problem is in PTIME if~$Q$ is
  safe, and \#P-hard otherwise.
\end{theorem}

This result has been extended in several ways, to apply to some queries
featuring negation~\cite{FiOl16}, disequality~($\neq$) joins~\cite{OlHu08}, or
inequality~($<$) joins~\cite{OlHu09}. More recently, two new directions have
been explored. First, our work
with 
Ceylan~\cite{amarilli2021dichotomy} extended
the study from UCQs to the broader class of \emph{homomorphism-closed} queries.
This class captures recursive queries such
as regular path queries (RPQs) or Datalog (without inequalities or negation).
In~\cite{amarilli2021dichotomy}, we focused on homomorphism-closed queries that
were \emph{unbounded}, i.e., not equivalent to a UCQ. We showed that
PQE is \#P-hard for \emph{any} such query,
though for technical reasons the result only applies to graphs, i.e.,
arity-two signatures.
This extended the above dichotomy to the full class of homomorphism-closed queries
(on arity-two signatures).

Second, the dichotomy has been extended from PQE to
restricted problems which do not allow arbitrary probabilities on the TID.
Kenig and Suciu~\cite{kenig2020dichotomy} have shown that the dichotomy
of~\cite{dalvi2012dichotomy} still held for the so-called \emph{generalized model
counting} problem, where the allowed probabilities on tuples are only $0$ (the
tuple is missing), $1/2$, or~$1$; this is in contrast with the original proof of the
dichotomy, which uses arbitrary probabilities. Our result
in~\cite{amarilli2021dichotomy} already held for the generalized model counting
problem.
What is more, for a subclass of the unsafe
queries, they showed that hardness still held for the \emph{model counting
problem}, where the probabilities are either~$0$ or~$1/2$. Independently, with
Kimelfeld~\cite{amarilli2022uniform}, we have shown hardness of the
same problem for the incomparable class of non-hierarchical CQs without
self-joins. Rather than model counting, we called this the \emph{uniform
reliability} (UR) problem, following the terminology in the work of Grädel, Gurevich, and
Hirsch~\cite{graedel1998complexity}.

In our opinion, this uniform reliability problem is interesting even
outside of the context of probabilistic databases: we simply ask,
for a fixed query $Q$, given a database instance $I$,
\emph{how many} subinstances of~$I$ satisfy~$Q$. The UR problem
also relates to computing the \emph{causal effect} and \emph{Shapley
values} in
databases~\cite{salimi2016quantifying,livshits2020shapley,amarilli2022uniform}.
What is more, UR for homomorphism-closed queries captures existing counting
problems on graphs, such as
\emph{st-connectedness}~\cite{valiant1979complexity} which asks how many
subgraphs of an input graph contain a path between a source and a
sink.

The ultimate goal of these two lines of work would be to classify the
complexity of \emph{uniform reliability}, across \emph{all
homomorphism-closed queries}. Specifically, one can conjecture:

\begin{conjecture}
  \label{con:goal}
  Let $Q$ be a homomorphism-closed query on an arbitrary signature. The
  uniform reliability problem for~$Q$ is in PTIME if $Q$ is a safe UCQ, and
  \#P-hard otherwise.
\end{conjecture}

To establish this, there are three obstacles to overcome. First, in the case where~$Q$ is a
UCQ, one would need to establish the hardness of UR for all unsafe UCQs,
extending the work of Kenig and
Suciu~\cite{kenig2020dichotomy}. Second, when $Q$ is unbounded,
one would need to adapt the methods
of~\cite{amarilli2021dichotomy} to apply to UR rather than
PQE. Third, the methods of~\cite{amarilli2021dichotomy} would need to be extended from graph signatures to
arbitrary arity signatures.

\subparagraph*{Result statement.}
In this paper, we address the second difficulty and show the
following, which extends the main result of~\cite{amarilli2021dichotomy} from
PQE to UR, and brings us closer to Conjecture~\ref{con:goal}:

\begin{theorem}[Main result]
  \label{thm:main}
  Let $Q$ be an unbounded homomorphism-closed query on an arity-two signature. The
  uniform reliability problem for~$Q$ is \#P-hard.
\end{theorem}

The proof of this result has the same high-level
structure as in~\cite{amarilli2021dichotomy},
but there are significant new technical challenges to overcome. In particular,
we now reduce from different problems, whose hardness rely (among other things)
on the hardness of uniform reliability for
the
query $R(x), S(x, y), T(y)$,
shown in~\cite{amarilli2022uniform}.
The impossibility to assign a probability of~$1$ to facts also makes reductions
much more challenging: intuitively, as all facts can now be missing, there is no
longer a clear connection between the possible worlds of the source problem and
the possible worlds of the database built in the reduction.
We use multiple tools to work around this, for instance
a \emph{saturation} technique that creates a large but polynomial number of
copies of some facts and argues
that their absence is sufficiently unlikely to be negligible.
As saturation cannot apply to unary facts, we also need to identify so-called
\emph{critical models}, a more elaborate variant of a notion
in~\cite{amarilli2021dichotomy}, minimizing carefully-chosen weight criteria.

We give a high-level structure of the proof below as it is presented in the rest of
the paper, and comment in more detail on how the techniques relate to our 
earlier work~\cite{amarilli2021dichotomy}.

\subparagraph*{Paper structure.}
We give preliminaries and the formal definition of UR in
Section~\ref{sec:prelim}, along with the two problems from which
we reduce: one problem on bipartite graphs from~\cite{amarilli2022uniform},
and one variant of a connectivity problem of~\cite{valiant1979complexity}. We
show that they are \#P-hard (Appendix~\ref{apx:prelim}).

We then review notions from~\cite{amarilli2021dichotomy} in
Section~\ref{sec:basic}: the \emph{dissociation} operation on instances, and the
notion of a \emph{tight edge}, which makes the query false when we apply
dissociation to it. We invoke a result from~\cite{amarilli2021dichotomy} showing
that tight edges always exist for unbounded queries. This is the only place
where we use the unboundedness of the query, and is unfortunately the only
result from~\cite{amarilli2021dichotomy} that can be used as-is. Some other
notions are reused and extended from~\cite{amarilli2021dichotomy} but they are
always re-defined and re-proved in a self-contained way in the present paper.

We then present in Section~\ref{sec:minimal} the notion of a \emph{critical
model}, as a model of the query which is \emph{subinstance-minimal} and
features a tight edge which is minimal by optimizing three successive
quantities: \emph{weight}, \emph{extra weight}, and \emph{lexicographic weight}.
The notion of \emph{weight} is from~\cite{amarilli2021dichotomy}, the two other
notions relate to \emph{side weight}
from~\cite{amarilli2021dichotomy} but significantly extend it. We show in
this section that a query
having a model with a tight edge also has a critical model.

We then move on to the hardness proof. As in~\cite{amarilli2021dichotomy}, there
are two cases: a \emph{non-iterable} case where we reduce from the problem on
bipartite graphs, and an \emph{iterable} case where we reduce from the
connectivity problem. In Section~\ref{sec:noniter}, we formally define the
notion of iteration (essentially identical to the notion
in~\cite{amarilli2021dichotomy}) and show hardness when there is a non-iterable
critical model. The coding used
in the reduction extends that of~\cite{amarilli2021dichotomy} with the
\emph{saturation} technique of creating a large number of copies of some
elements. There are many new technical challenges, e.g., proving that a polynomial number
of copies suffices to make the absence of the facts sufficiently unlikely, and
justifying that all the other facts are ``necessary'' for a query
match, using in particular subinstance-minimality and the notion of extra
weight.

Last, in Section~\ref{sec:iter}, we show hardness in the case where all critical
models are iterable. We first show that such models can be repeatedly iterated,
and that the measure of \emph{extra weight}
must be zero in this case, allowing us to focus on the more precise criterion of
lexicographic weight. Then we define the coding, which is similar
to~\cite{amarilli2021dichotomy} up to technical modifications. The reduction
does not use saturation but argues that all facts
are ``necessary'' using the notion of lexicographic weight and a new
\emph{explosion} structure.

We then conclude in Section~\ref{sec:conc}. To save space, most
proofs are deferred to the appendix.

\section{Preliminaries and Problem Statement}
\label{sec:prelim}
\begin{toappendix}
  \label{apx:prelim}
\end{toappendix}

\subparagraph*{Instances.}
We consider an \emph{arity-two relational signature}~$\sigma$ consisting of
\emph{relations} with
an associated \emph{arity}, where the maximal arity of the signature is assumed to be~$2$.
A \emph{$\sigma$-instance} (or just \emph{instance}) is a set of \emph{facts}, i.e., expressions of
the form $R(a, b)$ where $a$ and~$b$ are constants and~$R \in \sigma$.
We assume without loss of
generality that all relations in~$\sigma$ are binary, i.e., have arity two.
Indeed, if there are unary
relations $U$, we can simply code them with a binary relation~$U'$,
replacing facts $U(a)$ by $U'(a, a)$ in instances, and modifying the query to
interpret $U'(a, a)$ as~$U(a)$ and to ignore
facts $U'(a, b)$ with~$a \neq b$: this is similar to
Theorem~8.4 of~\cite{amarilli2021dichotomy}.
Accordingly, we call a fact $R(a, b)$ \emph{unary} if~$a=b$, otherwise it
is \emph{binary}.

The \emph{domain}~$\dom(I)$ of an instance~$I$ is the set of constants occurring in~$I$. A
\emph{homomorphism} from~$I$ to an instance~$I'$ is
a function $h \colon \dom(I) \to \dom(I')$ such that, for each fact $R(a, b)$
of~$I$, the fact $R(h(a), h(b))$ is in~$I'$. We say that $I'$
is a \emph{subinstance} of~$I$, 
written $I' \subseteq I$, 
if $I'$ is a subset of the facts of~$I$; we then have $\dom(I') \subseteq
\dom(I)$.

\subparagraph*{Queries.}
A \emph{query} $Q$ over~$\sigma$ is a Boolean function over $\sigma$-instances
which we always assume to be \emph{homomorphism-closed},
i.e., if $Q$ returns true on $I$ and $I$ has a homomorphism to an
instance~$I'$ then $Q$ also returns true on~$I'$.
When $Q$ returns true on~$I$ we call $I$ a
\emph{model} of~$Q$,
or say that $I$ \emph{satisfies} $Q$ (written $I \models Q$); otherwise 
$I$ \emph{violates}~$Q$.
Any homomorphism-closed query~$Q$ is
\emph{monotone}, i.e., if $I$ satisfies~$Q$
and~$I \subseteq I'$ then~$I'$ satisfies~$Q$.
A \emph{subinstance-minimal model} of~$Q$ is a model $I$ of~$Q$ such that no strict
subinstance of~$I$ satisfies~$Q$.

We focus on \emph{unbounded queries}, i.e., queries having an
infinite number of subinstance-minimal models.
Examples of well-studied homomorphism-closed query languages 
include \emph{conjunctive queries} (CQs), \emph{unions of CQs} (UCQs),
\emph{regular path
queries} (RPQs), and \emph{Datalog} without inequalities or negations.
The queries defined by
Datalog or RPQs are unbounded unless they are equivalent to a UCQ (i.e., 
non-recursive Datalog); more generally a query is either unbounded or equivalent
to a UCQ.

\subparagraph*{UR and PQE problems.}
In this paper, we study \emph{uniform reliability} (UR). The problem
$\UR(Q)$ for a fixed query~$Q$ is the following: we are given as input
an instance~$I$, and we must return how many subinstances of~$I$ satisfy~$Q$,
i.e., the number $|\{I' \subseteq I \mid I' \models Q\}|$. Note that we have no
general upper bound on the complexity of this problem, as we allow queries to be arbitrarily complex or even
undecidable to evaluate, e.g., ``there is a path $R(x_1), S(x_1, x_2), \ldots,
S(x_{n-1}, x_n), T(x_n)$ where~$n$ is the index of a Turing machine
that halts''.

We will sometimes consider the generalization of UR called \emph{probabilistic
query evaluation} (PQE). The $\PQE(Q)$ problem for a fixed query~$Q$ asks,
given an instance $I$ and a probability distribution $\pi\colon I \to [0,1]$
mapping each fact of~$I$ to a rational in~$[0,1]$, to determine the total
probability of the subinstances of~$I$ satisfying~$Q$, when each fact $F \in I$ is
drawn independently from the others with the probability $\pi(F)$. Formally,
we must compute:
$
\sum_{I' \subseteq I \text{~s.t.~} I' \models Q} \prod_{F \in I'} \pi(F)
\times \prod_{F' \in I \setminus I'} (1 - \pi(F)).
$

The UR problem is a special case of PQE where the function $\pi$ maps all facts
to~$1/2$, up to renormalization, i.e., multiplying by~$2^{|I|}$. We will sometimes abusively talk about
UR as the problem of computing that probability,
because this probabilistic phrasing makes it more convenient, e.g., to reason about
conditional probabilities, or about negligible probabilities.

\subparagraph*{Hard problems.}
The goal of this paper is to show Theorem~\ref{thm:main}.
We will establish \mbox{\#P-hardness} using \emph{polynomial-time Turing
reductions}~\cite{Cook71} (see~\cite{amarilli2021dichotomy} for details).
Specifically, we reduce from one of two \#P-hard problems, depending on the query.
In~\cite{amarilli2021dichotomy}, we reduce from the problems \#PP2DNF
and U-ST-CON (undirected source-to-target connectivity), which are
shown to be \#P-hard in~\cite{provan1983complexity}. In this paper, given our focus on UR, we
reduce from variants of these problems: the
\emph{$\lambda,\mu,\nu$-variable-clause-variable probabilistic 
\#PP2DNF problem}
and the \emph{$\phi,\eta$-vertex-edge probabilistic U-ST-CON problem}. We first define the first problem:

\begin{definition}
  \label{def:pp2dnf}
  Let $0 < \lambda,\nu < 1$ and $0 < \mu \leq 1$ be fixed probabilities. The
  \emph{$\lambda,\mu,\nu$-variable-clause-variable probabilistic \#PP2DNF problem}
  (or for brevity \emph{$\lambda,\mu,\nu$-\#PP2DNF}) is the following: given a bipartite graph $(U
  \cup V, E)$ with $E \subseteq U \times V$, we ask for the probability that we
  keep an edge and its two incident vertices, where vertices of $U$ have
  probability~$\lambda$ to be kept, 
  edges of~$E$ have
  probability~$\mu$ to be kept, and
  vertices of~$V$ have probability~$\nu$ to
  be kept, all these choices being independent. Formally, we must compute:\\
  $
    \sum_{\substack{(U', E', V') \subseteq U \times E \times V \\ E' \cap (U' \times V')
    \neq \emptyset}}\quad
    \lambda^{|U|'} \times (1-\lambda)^{|U|-|U|'} \times
    \mu^{|E|'} \times (1-\mu)^{|E|-|E|'} \times
    \nu^{|V|'} \times (1-\nu)^{|V|-|V|'}
  $
\end{definition}

The name \#PP2DNF is because of the link to positive partitioned 2-DNF
formulas, which we do not need here.
We can show that $\lambda,\mu,\nu$-\#PP2DNF is \#P-hard, by adapting the
proof in~\cite{amarilli2022uniform} which shows the hardness of
uniform reliability for the query $R(x), S(x, y), T(y)$:

\begin{toappendix}
  In this appendix, we give the formal details to show that the two problems
  from which we reduce are \#P-hard.

  We first explain more carefully why the first problem is intractable, as for
  inessential reasons this actually requires some inspection of the proof
  of~\cite{amarilli2022uniform}:
\end{toappendix}

\begin{propositionrep}[\cite{amarilli2022uniform}]
  For any fixed $0 < \lambda,\nu < 1$ and $0 < \mu \leq 1$, the
  problem $\lambda,\mu,\nu$-\#PP2DNF 
  is \#P-hard.
\end{propositionrep}

\begin{proof}
  We can see the input bipartite graph as an instance over the signature with
  unary relations $R$ and $T$ and binary relation~$S$, but the instances thus
  obtained must obey the following
  restrictions (*): the domains of~$R$
  and~$T$ are disjoint, and the domain of~$S$ is a subset of the products of
  that of~$R$ and~$T$. The probability of the $R$-facts is then~$\lambda$, that of the
  $T$-facts is then $\nu$, and that of the $S$-facts is then $\mu$, and we ask for
  the probability of drawing a possible world where the query $Q_0: R(x), S(x, y),
  T(y)$ holds.

  This is an instance of probabilistic query evaluation for the query $Q_0$
  with fixed probabilities $\lambda,\mu,\nu$ for the facts of the relations $R,S,T$
  respectively. This PQE problem is shown \#P-hard on arbitrary instances by Theorem~4.3
  of~\cite{amarilli2022uniform}. By inspection of the proof (see Section~5,
  ``Defining the gadgets'' and ``Defining the reduction''), one can check the
  instances used in the hardness proof in fact obey the restrictions~(*), so that the
  problem is still \#P-hard even when restricted to input instances obeying~(*).
  Thus, the $\lambda,\mu,\nu$-\#PP2DNF problem is also \#P-hard, which concludes the
  proof.
\end{proof}

We now define the second problem:

\begin{definition}
  \label{def:ustcon}
  Let $0 < \phi \leq 1$ and $0 < \eta < 1$ be fixed probabilities.
  The \emph{$\phi,\eta$-vertex-edge-probabilistic U-ST-CON problem}
  (or for brevity \emph{$\phi,\eta$-U-ST-CON})
  is the following: given an undirected
  graph~$G = (V, E)$ and source and sink vertices $r,s\in V$ with $r\neq s$, we ask for the
  probability that we keep a subset of edges and vertices containing a path that
  connects~$r$
  and~$s$ (in particular keeping~$r$ and~$s$),
  where vertices have probability~$\phi$ to be kept and edges
  have probability~$\eta$ to be kept, all these choices being independent. Formally, we must
  compute:
  \[
    \sum_{\substack{V' \subseteq V, E' \subseteq E\\
    r\text{~and~}s\text{~connected in~}(V',~ \restr{E'}{V'})}}
    \phi^{|V|'} \times (1-\phi)^{|V|-|V|'} \times
    \eta^{|E|'} \times (1-\eta)^{|E|-|E|'}
  \]
\end{definition}

This intuitively combines features of the undirected source-to-target
edge-connectedness and node-connectedness problems
of~\cite{valiant1979complexity}.
With standard techniques and some effort, we can show that
\mbox{$\phi,\eta$-U-ST-CON} is \#P-hard (see Appendix~\ref{apx:ustcon}):

\begin{toappendix}
We now show hardness of the second problem in the rest of this appendix:
\end{toappendix}

\begin{propositionrep}
  \label{prp:stconhard}
  For any fixed $0 < \phi \leq 1$ and~$0 < \eta < 1$, 
  the problem $\phi,\eta$-U-ST-CON is \#P-hard.
\end{propositionrep}

\begin{toappendix}
\label{apx:ustcon}
In the process of proving Proposition~\ref{prp:stconhard}, we will need to
consider the $\phi,\eta$-vertex-edge-probabilistic U-ST-CON problem in cases where
$\eta=1$, even though the statement of the proposition does not mention it.

So we actually show the following stronger statement:

\begin{proposition}
  \label{prp:stconhard2}
  For any fixed $0 < \phi \leq 1$ and~$0 < \eta \leq 1$ such that $\phi<1$
  or~$\eta<1$,
  the problem $\phi,\eta$-U-ST-CON is \#P-hard.
\end{proposition}

To do so, we first distinguish the cases where one of the probabilities is~$1$.
When the vertex probability is~$\phi = 1$, the problem is already known to be
hard:

\begin{proposition}[\cite{provan1983complexity}]
  \label{prp:hard1q}
  For any fixed probability~$0 < \eta < 1$, the problem $1,\eta$-U-ST-CON is \#P-hard.
\end{proposition}

\begin{proof}
  This is simply the standard U-ST-CON problem with probability~$\eta$ on the edges, which is \#P-hard~\cite{provan1983complexity}.
\end{proof}

When the edge probability is~$\eta = 1$, we can show hardness by reducing from the
undirected S--T NODE CONNECTEDNESS problem. That problem is shown to be \#P-hard
in~\cite{valiant1979complexity} but unfortunately only for the vertex
probability~$\phi = 1/2$. We can show that hardness also applies
to any arbitrary choice of~$\phi$ using a standard interpolation argument:

\begin{proposition}
  \label{prp:hardp1}
  For any fixed probability~$0 < \phi < 1$, the problem $\phi,1$-U-ST-CON is \#P-hard.
\end{proposition}

\begin{proof}
  We reduce from the undirected S--T NODE CONNECTEDNESS problem, which is shown
  to be \#P-hard in~\cite{valiant1979complexity}. This is the $\frac
  12,1$-U-ST-CON problem, i.e., where we count the
  number of vertex subsets such that there is a path connecting the source~$r$
  and sink~$s$.

  First note
  that, given that there is no path in the cases where~$r$ and~$s$ are not kept,
  up to multiplying by a factor of $1/4$ we can assume that the vertices~$r$ and~$s$ are
  always present. The same holds for our problem $\phi,1$-U-ST-CON, multiplying by a factor of~$\phi^2$. So we will work with these slightly modified
  problems where we only consider subsets of vertices where the vertices~$r$ and~$s$ are kept. In this case the
  answer to the problem is clearly~$1$ if~$r$ and~$s$ are adjacent,
  so we further assume that~$r$ and~$s$ are not adjacent.

  Given an instance of the S--T NODE CONNECTEDNESS PROBLEM, namely, an undirected graph $G = (V, E)$ and
  source and sink vertices $r$ and~$s$, our goal is to compute the number of
  \emph{good} vertex subsets, where we call a subset $V' \subseteq
  V \setminus \{s,t\}$ \emph{good} if there is an undirected path between~$r$
  and~$s$ in
  the induced subgraph on~$V'$ in~$G$.
  We define quantities $X_0, \ldots, X_{|V|-2}$
  where $X_i$ for $0 \leq i \leq |V|-2$ denotes the number of good subsets of
  cardinality~$i$.
  Note that, for instance, $X_0 = 0$ as $r$ and~$s$ are not adjacent. The
  quantity that we must compute is $\sum_i X_i$.
  We will explain how the computation of the vector $\vec{X}$ of these
  quantities reduces in polynomial time to the our 
  problem $\phi,1$-U-ST-CON problem, which suffices to conclude.

  For any positive integer $1 \leq q \leq |V|-1$, let $G_q$ denote the undirected graph
  obtained from~$G$ where each vertex~$v$ of~$V \setminus \{s,t\}$
  is replaced by copies~$v_1, \ldots, v_q$, where every edge $\{u, v\}$ with $u,
  v \notin \{s, t\}$ is
  replaced by the $q^2$ edges $\{u_i, v_j\}$ for~$1 \leq i, j \leq q$,
  and where every edge $\{u, s\}$ or $\{u, t\}$ is replaced by the $q$ edges
  $\{u_i, s\}$ or $\{u_i, t\}$ for $1 \leq i \leq q$ respectively. (Remember
  that there are no edges between~$r$ and~$s$.)
  Remark that after this transformation
  two copies $u_i$ and~$u_j$ of a vertex~$u$ are \emph{not} adjacent. The
  construction of~$G_q$ from~$G$ is 
  in polynomial time because the value~$q$ is polynomial in the size of~$q$.

  We now invoke our oracle for the problem
  $\phi,1$-U-ST-CON on the graph $G_q$ for each $1 \leq q \leq |V|-1$, returning
  probabilities $P_1, \ldots, P_{|V|-1}$. Let us examine these probabilities in
  more detail and show that they are connected to the quantities $X_0, \ldots,
  X_{|V|-2}$ that we must compute, by showing the following equation:
  \begin{equation}
    \label{eq:pq}
    P_q = (1-\phi^q)^{|V|-2} \sum_{0 \leq i \leq |V|-2} X_i \times 
    \left(\frac{1}{1-\phi^q} - 1\right)^i
  \end{equation}
  To show Equation~\ref{eq:pq}, notice that 
  we can choose a subset of vertices of~$G_q$ by selecting
  vertices in~$V \setminus \{s, t\}$
  and choosing whether they are \emph{kept} (i.e., one of their copies in~$G_q$
  is kept) or \emph{not
  kept} (i.e., none of the copies are kept). Note that the probability that a
  vertex is not kept is $(1-\phi)^q$ and the probability that it is kept is $1 -
  (1-\phi)^q$. Now, note that a subset of vertices of~$G_q$ has a path
  connecting~$r$ and~$s$ iff the kept vertices in~$G$ have a path
  connecting~$r$ and~$s$, i.e., the kept vertices in~$G$ are a good subset.
  Indeed, in one direction, if there is a path of kept vertices
  connecting~$r$ and~$s$ in~$G_q$, modifying the path by replacing each vertex
  copy~$v_i$ by the vertex~$v$, we obtain a path connecting~$r$ and~$s$ in~$G$,
  each vertex of which was kept. Conversely,
  if there is a path of kept vertices connecting~$r$ and~$s$ in~$G$, we
  obtain a path connecting~$r$ and~$s$ in~$G_q$ by replacing each kept vertex by
  some copy witnessing that it is kept.

  For this reason, the value~$P_q$ returned by the oracle on~$G_q$ can be
  expressed as:
  \[
    P_q = \sum_{\substack{V' \subseteq V \setminus\{s, t\}\\V' \text{~good~in~}
    G}}
    (1-\phi^q)^{|V|-2-|V'|} \times (1-(1 - \phi^q))^{|V'|}
  \]
  As the summand only depends on the cardinality of~$V'$,
  we can split the sum according to this cardinality
  and regroup the factors, making the
  quantities~$X_i$ appear, to obtain:
  \[
    P_q = \sum_{0 \leq i \leq |V|-2} X_i \times 
    (1-\phi^q)^{|V|-2-i} \times (1-(1 - \phi^q))^{i}
  \]
  Taking out the common factors, and rewriting, we get:
  \[
    P_q = (1-\phi^q)^{|V|-2} \sum_{0 \leq i \leq |V|-2} X_i \times 
    \left(\frac{1-(1 - \phi^q)}{1-\phi^q}\right)^i
  \]
  Rewriting the contents of the parenthesis, we have
  obtained Equation~\ref{eq:pq}.

  Now, Equation~\ref{eq:pq} means that, from the oracle answers $P_1, \ldots,
  P_{|V|-1}$,
  by dividing by $(1-\phi^q)^{|V|-2}$, we
  can recover the vector $\vec Q$ of the following quantities, for $1 \leq q
  \leq |V|-1$:
  \begin{equation}
    \label{eq:qq}
    Q_q = \sum_{0 \leq i \leq |V|-2} X_i \times \alpha_\phi(q)^i
  \end{equation}
  where $\alpha_\phi(q) = \frac{1}{1-\phi^q} - 1$, which is incidentally
  well-defined because $0 < \phi < 1$ and $q > 0$. Further observe that 
  the function~$\alpha_\phi$ is in fact bijective (strictly decreasing) over the
  positive reals.

  Now, Equation~\ref{eq:qq} can be seen as a matrix equation
  \[
    \vec Q =  M \vec X
  \]
  where the matrix~$M$ is a square matrix whose $q$-th row and $i$-th column
  contains $\alpha_\phi(q)^i$. We recognize a Vandermonde matrix. Further, as
  $\alpha_\phi$ is bijective, the values $\alpha_\phi(q)$ are pairwise distinct,
  so~$M$ is invertible.

  This concludes the presentation of our reduction: given the input graph~$G$,
  we build the graph $G_q$ for each $1 \leq q \leq |V|-1$ in polynomial time,
  we obtain the values~$P_q$ from the oracle, deduce the
  vector~$\vec Q$ by dividing by $(1 - \phi^q)^{|V|-2}$,
  build the matrix~$M$, invert it in polynomial time, compute $\vec X = M^{-1}
  \vec Q$, and compute the sum $\sum_i X_i$ which is the answer to the
  undirected S-T NODE CONNECTEDNESS problem instance
  from which we were reducing. This concludes the proof.
\end{proof}

We will conclude the proof of 
Proposition~\ref{prp:stconhard2}, hence of
Proposition~\ref{prp:stconhard}, using
  Propositions~\ref{prp:hard1q} and~\ref{prp:hardp1}, by a saturation and
  rounding argument. To do this, we need a general-purpose lemma about
  probabilistic computation, which we will actually reuse when proving
  Lemma~\ref{lem:negligible} in Appendix~\ref{apx:noniter}:

\begin{lemma}
  \label{lem:proba}
  Fix some probability $0 < \zeta < 1$.
  We are given $0 < \epsilon < 1$ be a probability, and let $\chi > 0$ be an
  integer. Then for any integer $q > \frac{\ln(\chi) -
  \ln(\epsilon)}{-\ln(1-\zeta)}$, we have:
  \[
  1 - (1 - (1-\zeta)^q)^\chi < \epsilon
  \]
\end{lemma}
Note that $\ln(\epsilon) < 0$, $\ln(\chi) \geq 0$, and $\ln(1-\zeta)<0$.

The intuition of the lemma is
the following. We are considering a process where we want to have all of~$\chi$ draws
succeed. Each draw consists of performing some number $q$ of attempts, and each attempt
succeeds with probability~$\zeta$ to succeed one attempt.
The probability~$\zeta$ is fixed, but the number of draws~$\chi$ is part of the
input, and we must choose the value of~$q$.
The lemma tells us how large we need to make~$q$ so as to
guarantee that the probability that some attempt succeeds in each of the $\chi$
draws is sufficiently high, i.e, the probability that there is a draw where we fail
all the attempts is at most some target probability~$\epsilon$.

\begin{proof}[Proof of Lemma~\ref{lem:proba}]
  The equation to show can be equivalently rephrased to:
\[
  (1 - (1-\zeta)^q)^\chi > 1- \epsilon
\]
Or, as $\chi>0$:
\[
  1 - (1-\zeta)^q > (1- \epsilon)^{1/\chi}
\]
Equivalently:
  \begin{equation}
    \label{eqn:target}
  (1-\zeta)^q < 1 - (1- \epsilon)^{1/\chi}
  \end{equation}
Thus, establishing Equation~\ref{eqn:target} suffices to conclude.

  Let us first show that we have:
  \begin{equation}
    \label{eqn:bound}
    (1-\zeta)^q < \frac{\epsilon}{\chi}
  \end{equation}
  To see why this holds, recall that we have defined:
  \[
    q > \frac{\ln(\chi) - \ln(\epsilon)}{-\ln(1-\zeta)}
  \]
  As $0<\zeta<1$, we have $\ln(1-\zeta) < 0$, so we deduce:
  \[
    q \ln(1-\zeta) < \ln(\epsilon) - \ln(\chi)
  \]
  Exponentiating, we get Equation~\ref{eqn:bound}.

  Now let us conclude. 
  We use Bernoulli's inequality, which states that for every real number $0
  \leq r \leq 1$ and $x \geq -1$ we have:
  \[
    (1+x)^r \leq 1 + rx
  \]
  In particular, for $0 \leq x \leq 1$ and $\chi>0$, we have:
  \[
    (1-x)^{1/\chi} \leq 1 - \frac{x}{\chi}
  \]
  Equivalently:
  \begin{equation}
    \label{eqn:bernoulli}
    \frac{x}{\chi} \leq 1 - (1-x)^{1/\chi}
  \end{equation}
  Thus from Equation~\ref{eqn:bound} and Equation~\ref{eqn:bernoulli}, taking
  $x = \epsilon$, we get
  Equation~\ref{eqn:target}, which concludes the proof.
\end{proof}

We can now show Proposition~\ref{prp:stconhard2}:

\begin{proof}[Proof of Proposition~\ref{prp:stconhard2}]
  Let $\zeta$ and $\eta$ be the fixed vertex and edge probabilities. If the
  vertex probability~$\zeta$ is~$1$, we
  conclude by Proposition~\ref{prp:hard1q}.
  If the edge probability~$\eta$ is~$1$, we conclude by 
  Proposition~\ref{prp:hardp1}.
  Otherwise, we show the hardness of $\zeta,\eta$-U-ST-CON in the rest of the
  proof, by reducing from the case $\zeta,1$-U-ST-CON where edges have probability~$1$. 
  We use a saturation and rounding argument where we intuitively replace each
  edge by a large number of parallel paths.

  Specifically, let us consider an instance
  to the problem $\phi,1$-U-ST-CON, consisting of an
  undirected graph $G$ with distinguished source and sink vertices $r$
  and~$s$.
  Let $n$ be the number of vertices of~$G$, and let $m$ be its number of edges.
  We assume without loss of generality that $n, m > 0$ as the answer is trivial
  otherwise.
  For any integer $q > 0$, we denote by~$G_q = (V_q, E_q)$ the
  graph obtained from~$G$ by keeping the same source and sink vertex and replacing each edge by $q$ parallel paths of
  length~$2$, i.e., every edge $\{u, v\}$ is replaced by $2q$ edges $\{u, w_{u,v,1}\},
  \ldots, \{u, w_{u,v,q}\}$ and $\{w_{u,v,1}, v\}, \ldots, \{w_{u,v,q}, v\}$ where
  the intermediate vertices $w_{u,v,1},
  \ldots, w_{u,v,q}$ are fresh.

  A \emph{possible world} of~$G_q$ is a subset $(V', E')$ of $V_q \times E_q$;
  remember that we have an oracle that can tell us the total probability of the
  possible worlds of~$G_q$ that are \emph{good}, i.e., contain a path connecting $r$ and~$s$
  where all vertices and edges are kept. By contrast, what we want to compute is
  the probability of vertex subsets $V'$ of~$V$ that are \emph{good}, i.e., the
  induced subgraph of~$G$ on~$V'$ contains a path connecting~$r$ and~$s$.

  We will argue that for sufficiently large~$q$ it is very likely that each edge
  of~$G$ is reflected by one of the paths of length~$2$ that codes it. Formally,
  we say that a possible world 
  is \emph{valid} if,
  for each of the $m$ edges $\{u, v\}$ of~$G$,
  there is some choice of~$1 \leq p \leq q$ such that 
  the vertex $w_{u,v,p}$ and the edges $\{u, w_{u,v,p}\}$ and $\{v, w_{u,v,p}\}$
  were kept; otherwise the possible world is \emph{invalid}.
  We claim that, if we consider the possible worlds of~$G_q$ 
  where we condition the choices of the vertices $w_{u,v,i}$ for $\{u,v\} \in E$
  and $1 \leq i \leq q$ and of all edges to ensure that~$V'$ is valid,
  then there is a path connecting~$r$ and $s$
  (under draws of the other vertices
  i.e., those of $G$) with same probability 
  as in~$G$. Indeed, there is an immediate
  probability-preserving bijection between the subsets, and there is a
  length-2 path connecting two vertices $u$ and~$v$ in~$G_q$ precisely for the
  vertex pairs $u$ and $v$ connected by an edge in~$G$.

  So let us investigate the probability $\epsilon'$ of getting an invalid
  possible world of~$G_q$, and show
  that it can be made sufficiently small with a polynomial value of~$q$, so that
  the answer of our oracle on~$G_q$ is sufficiently close to the probability of
  the subsets of~$V$ containing a path connecting~$r$ and~$s$. 
  Let us write the rational~$\phi$ as~$a/b$
  where $a$ and $b$ are integers, so that $1-\phi = (b-a)/b$.
  Define $\epsilon \colonequals \frac{(1/b)^n}{2}$: this ensures that the
  probability of any subset of vertices of~$G$ is at most $2\epsilon$. 
  We want to show that the probability of invalid subsets is
  negligible, i.e., that we have $\epsilon' < \epsilon$. Now, the probability of an
  invalid subset is by definition the following:
  \[
    \epsilon' = 1 - (1 - (1 - \phi\eta^2)^q)^m
  \]
  because a subset is invalid if there is an edge for which for all $q$ paths of
  length~$2$ coding it we did not keep the two edges and the intermediate vertex
  in the path of length~$2$ that codes it.
  For an explanation of this expression, see the details below the statement of
  Lemma~\ref{lem:proba}.
  By Lemma~\ref{lem:proba}, taking $\zeta \colonequals \phi\eta^2$ and $\chi
  \colonequals m$, 
  we have $\epsilon' < \epsilon$ if we take $q \colonequals 1+\left\lfloor \frac{\ln(n) -
  \ln(\epsilon)}{-\ln(1-\zeta)}\right\rfloor$.
  Remembering that $\ln(1-\zeta)$ is a constant, and
  noting that by definition of~$\epsilon$ we have $\ln(\epsilon) = n \ln(1/b)
  - \ln 2$ where~$b$ is a constant, this value~$q$ is polynomial in~$n$, hence in the size of the input
  graph~$G$.

  This allows us to conclude the reduction. Given the graph~$G$ and vertices $r$
  and~$s$, we build~$G_q$ for this value of~$q$, which is in polynomial time
  in~$G$.
  We then call our oracle for
  the $\phi,\eta$-vertex-edge-probabilistic problem on~$G_q$. It returns the value:
  \[
    O = \epsilon' X + (1-\epsilon') Y
  \]
  Where $\epsilon' < \epsilon$ is the probability of getting an invalid
  possible world, where $X$ is the probability that an
  invalid possible world is good (which is unknown); and $Y$ is the probability 
  that a valid possible world is good (which as we argued is the answer to the  
  $\phi,1$-U-ST-CON instance~$G$, i.e., what we need to compute).
  Equivalently, we obtain:
  \[
    O = Y + \epsilon'(X-Y) + Y
  \]
  Recall that $\phi = a/b$, so the
  answer that we wish to compute is of the form $Y = Z/b^n$ for some
  integer~$Z$, and it is sufficient to determine the integer~$Z$.
  We can multiply the
  oracle result~$O$ by $b^n$ and obtain:
  \[
    O b^n = Z + b^n \epsilon' (X-Y)
  \]
  Now, let us argue that we can recover~$Z$ from~$O b^n$, hence recover $Y$
  from~$O$ and conclude.
  As $X$ and $Y$ are probabilities we
  have $-1 \leq X-Y \leq 1$, and as $0 \geq \epsilon' < \epsilon = \frac{b^{-n}}{2}$, we conclude
  that $-1/2 < b^n \epsilon' (X-Y) < 1/2$, so we recover~$Z$ from the oracle
  answer by rounding the rational $O b^n$ to the nearest integer, concluding the reduction.
\end{proof}

\end{toappendix}

\section{Basic Techniques: Dissociation, Tight Edges}
\label{sec:basic}
Having presented the hard problems, we now recall the notion
of \emph{edges} and how we \emph{copy} them, and the \emph{dissociation}
operation introduced in~\cite{amarilli2021dichotomy}. We also present 
\emph{tight edges} and re-state the result of~\cite{amarilli2021dichotomy}
showing that unbounded queries have models with tight edges.

\subparagraph*{Edges and copies.}
An \emph{edge} $e$ in an instance $I$ is an ordered pair $(u,v)$ of distinct elements of $\dom(I)$ such
that there is at least one fact of~$I$ using both~$u$ and~$v$, i.e., of the form
$R(u, v)$ or~$R(v, u)$, hence non-unary.
The \emph{covering facts} of~$e$ in~$I$ is the non-empty set of these facts.
Note that $(u,v)$ is an edge iff $(v, u)$ is, and they have the same covering
facts.

We call~$e = (u,v)$ a \emph{non-leaf edge}
if~$I$ contains facts using~$u$ but not $v$ (called \emph{left-incident facts})
and facts using~$v$ but not~$u$ (called \emph{right-incident facts}).
An example is shown in Figure~\ref{fig:edge} (with no unary facts).
The left-incident and right-incident facts are called together the \emph{incident
facts}; note that they may include unary facts.
In this paper we will often modify instances $I$ by \emph{copying} an edge $e = (u,
v)$ of~$I$ to some other ordered pair $(u', v')$ of elements.
This means that we modify~$I$ to add,
for each covering fact~$F$ of~$e$, the fact 
obtained by replacing~$u$ by~$u'$ and~$v$ by~$v'$. Note that, if $u'$ and~$v'$
are both fresh, or if $u'=u$ and $v'$ is fresh or $v'=v$ and~$u'$ is fresh,
then the result of this process has a homomorphism back to~$I$.
Clearly, copying~$(u,v)$ on~$(u',v')$ is equivalent to copying~$(v,u)$
on~$(v',u')$ (but different from copying, say, $(u,v)$ on~$(v',u')$).
Note that copying an edge does \emph{not} copy its incident facts, though
our constructions will often separately copy some of them.

\begin{example}
  In the instance $I = \{R(a), S(a, b), S'(b, a), T(b)\}$, copying $(a, b)$ on~$(a, b')$
  for a fresh element~$b'$ means adding the facts $S(a, b'), S'(b', a)$.
\end{example}

\subparagraph*{Dissociation.}
One basic operation on instances is \emph{dissociation}, which replaces one edge
by two copies connected to each endpoint:

\begin{definition}
  Let $I$ be an instance and $e = (u, v)$ be a non-leaf edge of~$I$. The
  \emph{dissociation} of~$e$ in~$I$ is obtained by modifying $I$ to add two
  fresh elements~$u'$ and~$v'$, copying~$e$ to~$(u', v)$ and
  to~$(u, v')$, and then removing the covering facts of~$e$.
\end{definition}

\begin{figure}
  \begin{subfigure}[t]{.25\linewidth}
\begin{tikzpicture}[xscale=.5,yscale=.6,inner sep=1.3]
  \node (u) at (0, 0) {$u$} ;
  \node (v) at (2, 0) {$v$} ;
  \draw[thick,orange,->] (u) -- (v);
  \node (t1) at (-2, -1) {$t_1$};
  \node (t2) at (-2, 1) {$t_2$};
  \node (w1) at (4, -1) {$w_1$};
  \node (w2) at (4, 1) {$w_2$};
  \draw[thick,black,->] (u) -- (t1);
  \draw[thick,blue,->] (u) -- (t2);
  \draw[thick,black,->] (w1) -- (v);
  \draw[thick,orange,->] (w2) -- (v);
\end{tikzpicture}
    \caption{A non-leaf edge~$e$ with incident facts}
    \label{fig:edge}
  \end{subfigure}
\hfill
  \begin{subfigure}[t]{.25\linewidth}
\begin{tikzpicture}[xscale=.5,yscale=.6,inner sep=1.3]
  \node (u) at (0, 0) {$u$} ;
  \node (up) at (0, -.5) {$u'$} ;
  \node (v) at (2, 0) {$v$} ;
  \node (vp) at (2, .5) {$v'$} ;
  \draw[thick,orange,->] (u) -- (vp);
  \draw[thick,orange,->] (up) -- (v);
  \node (t1) at (-2, -1) {$t_1$};
  \node (t2) at (-2, 1) {$t_2$};
  \node (w1) at (4, -1) {$w_1$};
  \node (w2) at (4, 1) {$w_2$};
  \draw[thick,black,->] (u) -- (t1);
  \draw[thick,blue,->] (u) -- (t2);
  \draw[thick,black,->] (w1) -- (v);
  \draw[thick,orange,->] (w2) -- (v);
\end{tikzpicture}
    \caption{The dissociation of~$e$}
    \label{fig:dissoc}
  \end{subfigure}
\hfill
  \begin{subfigure}[t]{.45\linewidth}
    \centering
\begin{tikzpicture}[xscale=.9,yscale=.6,inner sep=1.3]
  \node (u) at (0, 0) {$u$} ;
  \node (v) at (2, 0) {$v$} ;
  \draw[thick,orange,->] (u) -- (v);
  \node (t1) at (-2, -1) {$t_1$};
  \node (t2) at (-2, 1) {$t_2$};
  \node (w1) at (4, -1) {$w_1$};
  \node (w2) at (4, 1) {$w_2$};
  \draw[thick,orange,->] (u) -- (t1);
  \draw[thick,orange,->] (u) -- (t2);
  \draw[thick,orange,->] (w1) -- (v);
  \draw[thick,orange,->] (w2) -- (v);
  \node (g1) at (-.5,-1) {$g_1$};
  \node (g2) at (1,-1) {$g_2$};
  \node (g3) at (2.5,-1) {$g_3$};
  \draw[thick,orange,dashed,->] (u) -- (g1);
  \draw[thick,orange,dashed,->] (u) -- (g2);
  \draw[thick,orange,dashed,->] (g2) -- (v);
  \draw[thick,orange,dashed,->] (g3) -- (v);
  \node (x1) at (-.5,1) {$x_1$};
  \node (x2) at (.5,1) {$x_2$};
  \node (x3) at (1.5,1) {$x_3$};
  \node (x4) at (2.5,1) {$x_4$};
  \draw[thick,orange,dashed,->] (u) -- (x2);
  \draw[thick,orange,->] (u) -- (x2);
  \draw[thick,orange,->] (x2) -- (v);
  \draw[thick,orange,dashed,->] (u) -- (x4);
  \draw[thick,black,->] (u) -- (x1);
  \draw[thick,blue,->] (u) -- (x3);
  \draw[thick,black,->] (x3) -- (v);
  \draw[thick,orange,->] (x4) -- (v);
\end{tikzpicture}
    \caption{An edge and various kinds of incident facts. See
    Example~\ref{exa:minimal}}
    \label{fig:minimal}
  \end{subfigure}
  \caption{Examples of Section~\ref{sec:basic} and~\ref{sec:minimal}}
  \label{fig:basic}
\end{figure}

The process is illustrated in Figures~\ref{fig:edge} and~\ref{fig:dissoc}. Note the following
immediate observation:

\begin{claim}
  The dissociation of an edge in~$I$ has a homomorphism back to~$I$.
\end{claim}

\subparagraph*{Tight edges.}
We can then define a \emph{tight edge} as one whose dissociation breaks the
query:

\begin{definition}
  A non-leaf edge $(u, v)$ in an instance~$I$ is \emph{tight} for the query~$Q$
  if $I$ satisfies~$Q$ but the
  dissociation of $(u,v)$ in~$I$ does not.
\end{definition}

We use a result of~\cite{amarilli2021dichotomy} which shows that
unbounded queries must have a
model with a tight edge. This is the only point where we
use the unboundedness of the query.

\begin{theorem}[Theorem~6.6 in~\cite{amarilli2021dichotomy}]
  \label{thm:tight}
  Any unbounded query has a model with a tight edge.
\end{theorem}

We give a proof sketch for completeness (see~\cite{amarilli2021dichotomy} for
the proof):

\begin{proofsketch}
  As the query~$Q$ is unbounded, it has infinitely many minimal models: let~$I$ be a
  sufficiently large one. Iteratively dissociate the non-leaf edges of~$I$ until
  none remain (this always terminates), and let~$I'$ be the result. If $I'$
  violates~$Q$, then some dissociation broke~$Q$, i.e., was applied to 
  a tight edge in a model of~$Q$. Otherwise, $I'$
  has no non-leaf edges and satisfies~$Q$. We can then show thanks to the simple
  structure of~$I'$ that it has a constant-sized subset that satisfies~$Q$, and
  deduce that~$Q$ already holds on a
  constant-sized subinstance of~$I$. As $I$ is large, this contradicts the
  minimality of~$I$.
\end{proofsketch}

Thus, in the sequel, we fix the query~$Q$ and assume that it has a model with a
tight edge.
Note that some bounded queries may also have a tight edge, e.g., the
prototypical unsafe CQ $R(x), S(x, y), T(y)$; our results in this paper thus
also apply to some bounded queries.

\begin{toappendix}
  \label{apx:basic}
  In the appendix, to simplify some of the technical proofs that use the dissociation process, we
  explain how we can extend the definition of this process so that it also
  applies (in a vacuous way) to non-leaf edges and non-edges.
  This will make it easier to deal with several different
  cases in a unified fashion in proofs:
  \begin{definition}
    Let $I$ be an instance and let $(u,v)$ be a \emph{non-edge}, i.e., a ordered
    pair of distinct elements of
    $\dom(I)$ such that no fact of~$I$ uses both~$u$ and~$v$.
    The \emph{dissociation} of~$(u,v)$ in~$I$ is simply~$I$, i.e., we do not do anything.

    Let $I$ be an instance. A \emph{leaf edge} is an ordered pair $e= (u,v)$ of distinct
    elements of~$\dom(I)$ which is an edge but not a non-leaf edge. Hence, one
    of~$u$ and~$v$ (or both) is such that the only facts of~$I$ where this
    element appears are the covering facts of~$e$. We call such elements the
    \emph{leaf elements} of~$e$. (Specifically, if both $u$ and $v$ are leaf
    elements then the only facts of~$I$ using $u$ or~$v$ are the covering facts
    of~$e$.) Then the \emph{dissociation} of~$e$ in~$I$ is
    obtained by replacing the leaf elements by fresh constants. For instance, if
    $u$ occurs in other facts of~$I$ than the covering facts of~$e$ but~$v$ does
    not, then we rename~$v$ to be a fresh element. Likewise, if both~$u$ and~$v$
    are leaf elements, then we rename them both, i.e., the covering facts of~$e$
    are replaced by an isomorphic copy on two fresh elements.

    Note that, for any instance~$I$, if we take $e= (u,v)$ to be a non-edge or a
    leaf edge of~$I$, then the dissociation of~$e$ in~$I$ is isomorphic to~$I$,
    in particular it satisfies the query iff $I$ does. Specifically,
    dissociating a non-edge or a leaf edge never makes the query false.

    Further note that, for any ordered pair $e = (u,v)$ of distinct elements
    of~$\dom(I)$ (whether non-edge, leaf edge, or non-leaf edge) then the
    dissociation of~$e$ in~$I$ is an instance in which~$e$ is no longer an edge.
  \end{definition}
\end{toappendix}

\section{Minimality and Critical Models}
\label{sec:minimal}
In this section, we refine the notion of a tight edge to impose minimality
criteria and get to the notion of \emph{critical models}.
We define three successive minimality criteria, which we present intuitively
here before formalizing them in the rest of this section.
The first is called \emph{weight} and counts the covering facts;
the \emph{critical weight} $\Theta$ is the minimal weight of a tight edge. Having defined~$\Theta$,
we restrict our attention to \emph{clean} tight edges~$e$,
whose incident facts do not include so-called \emph{garbage facts}, i.e., strict subsets of the covering
facts of~$e$.
The second criterion is \emph{extra weight} and counts the incident facts that
are not isomorphic to the covering facts; the \emph{critical extra
weight} $\Xi$ is the minimal extra weight of a tight edge of weight~$\Theta$.
The third criterion is \emph{lexicographic weight} and counts the other left-incident and
right-incident facts, ordered lexicographically: the \emph{critical
lexicographic weight} $\Lambda$ is the minimal lexicographic weight of a tight
edge of weight~$\Theta$ and extra weight~$\Xi$.

We then define a \emph{critical model} as a \emph{subinstance-minimal} model with a clean tight edge that
optimizes these three weights in order,
and show that such models exist.

\subparagraph*{Weight.}
The \emph{weight} was defined in~\cite{amarilli2021dichotomy},
but unlike in~\cite{amarilli2021dichotomy} we do not count unary facts:

\begin{definition}
  The \emph{weight} of an edge $e = (u, v)$ in an instance~$I$ is the number
  of covering facts of~$e$ (it is necessarily greater than~$0$).
\end{definition}

\begin{example}
  The weight of $(a,b)$ in $I = \{R(b), T(b, c), S(b,a), S'(b,a), U(a,b)\}$
  is~$3$.
\end{example}

The minimal weight of a tight edge across all models is an intrinsic
characteristic of~$Q$, called the \emph{critical weight}:

\begin{definition}
  The \emph{critical weight} of the query~$Q$, written $\Theta \geq 1$, is the
  minimum, across all models $I$ of~$Q$ and tight edges $e$ of~$I$, of the weight
  of~$e$ in~$I$.
\end{definition}

The point of the critical weight is that edges with weight less than~$\Theta$ can never
be tight:

\begin{claimrep}
  \label{clm:dissocweight}
  Let $I$ be a model of~$Q$ and $e = (u, v)$ be a non-leaf edge of~$I$. If the
  weight of~$e$ is less than~$\Theta$,
  then the dissociation of~$e$ in~$I$ is also a model of~$Q$.
\end{claimrep}

\begin{proof}
  If the dissociation did not satisfy~$Q$, then $e$ would be tight, which would
  contradict the minimality of~$\Theta$.
\end{proof}

\begin{example}
  \label{exa:weight}
  The bounded CQ $Q': R(x), S(x, y), S'(x, y), T(y)$ has critical weight~$2$,
  as witnessed by the model $I' = \{R(a), S(a, b), S'(a, b), T(b)\}$ with a
  tight non-leaf edge
  $(a, b)$ of weight 2 and the inexistence of a model with a tight non-leaf edge of weight $1$.

  As $Q'$ has critical weight~$2$, in any model $I$ of~$Q'$, if we have an edge
  $e= (u, v)$ with only one covering fact using both~$u$
  and~$v$, we know that dissociating~$e$ cannot make~$Q'$ false.

\end{example}

Having defined~$\Theta$, to simplify further definitions, we introduce the
notion of a \emph{clean} edge as one that does not have incident facts achieving
strict subsets of its covering facts:
\begin{definition}
  Let $I$ be an instance, let $e=(u,v)$ be an edge of~$I$,
  and let~$C \subseteq I$ be the covering facts of~$e$.
  For any edge $(u,t)$, if its covering facts are isomorphic to a strict subset
  of~$C$ when renaming $t$ to~$v$, then we call these left-incident facts
  \emph{left garbage facts}. Likewise, the \emph{right garbage facts} are the
  right-incident facts that are covering facts of edges $(w,v)$ that are
  isomorphic to a strict subset of~$C$ when renaming~$w$ to~$u$.

  We call~$e$ \emph{clean} if
  it has no left or right garbage facts (called collectively \emph{garbage
  facts}).
\end{definition}

\begin{example}
  In the instance $I = \{S(a, b'), U(a), S(a, b), S'(b, a), T(c, b), S(c, b), S'(d, b), 
  \allowbreak S'(b, e), \allowbreak S(f, b)\}$, the left garbage facts of the
  edge~$(a, b)$ are $\{S(a, b')\}$ on the edge $(a,
  b')$, and the right
  garbage facts are $\{S'(b, e)\}$ on the edge $(e, b)$ and $\{S(f, b)\}$ on the
  edge $(f, b)$. Note that there are no
  garbage facts on the edge $(b, c)$, because the covering facts $\{T(c, b),
  S(c, b)\}$ of this edge are not isomorphic to a strict subset of the covering
  facts of~$(a, b)$. Further note that there are no garbage facts on the edge
  $(d, b)$, because the covering facts $\{S'(d, b)\}$ are not isomorphic to a
  strict subset of the covering facts of $(a,b)$ when renaming $d$ to~$a$.
\end{example}

We will always be able to ensure that tight edges with critical weight are
clean, justifying that we restrict our attention to clean tight edges in the
sequel:

\begin{claimrep}
  \label{clm:clean}
  If $Q$ has a model with a tight edge, then it has a model with a clean tight
  edge of weight~$\Theta$.
\end{claimrep}

\begin{proofsketch}
  We find a model with a tight edge of weight~$\Theta$ by definition
  of~$\Theta$. Then, any edges with garbage facts have weight~$<\Theta$,
  so they can be dissociated using Claim~\ref{clm:dissocweight} and homomorphically
  merged to~$e$. At the end of this process, $e$ is clean and is
  still tight.
\end{proofsketch}

\begin{toappendix}
  \begin{figure}
    \begin{tikzpicture}[xscale=.8,yscale=.4,inner sep=1.5]
    \node (u) at (0, 0) {$u$};
    \node (v) at (1, 0) {$v$};
    \node (t1) at (-1, 1) {$t_1$};
    \node (t2) at (-1, 0) {$t_2$};
    \node (t3) at (-1, -2) {$x_1$};
    \node (w1) at (2, 1) {$w_1$};
    \node (w2) at (2, 0) {$w_2$};
    \node (w3) at (2, -2) {$y_1$};
    \draw[orange,->,thick] (u) -- (v);
    \draw[orange,->,thick] (u) -- (t1);
    \draw[black,->,thick] (u) -- (t2);
    \draw[orange,dashed,->,thick] (u) -- (t3);
    \draw[orange,->,thick] (w1) -- (v);
    \draw[black,->,thick] (w2) -- (v);
    \draw[orange,dashed,->,thick] (w3) -- (v);
    \end{tikzpicture}
    \hfill
    \begin{tikzpicture}[xscale=.8,yscale=.4,inner sep=1.5]
    \node (u) at (0, 0) {$u$};
    \node (up) at (0, -1) {$u'$};
    \node (v) at (1, 0) {$v$};
    \node (vp) at (1, -1) {$v'$};
    \node (t1) at (-1, 1) {$t_1$};
    \node (t2) at (-1, 0) {$t_2$};
    \node (t3) at (-1, -3) {$x_1$};
    \node (t3p) at (-1, -2) {$x_1'$};
    \node (w1) at (2, 1) {$w_1$};
    \node (w2) at (2, 0) {$w_2$};
    \node (w3) at (2, -3) {$y_1$};
    \node (w3p) at (2, -2) {$y_1'$};
    \draw[orange,->,thick] (u) -- (v);
    \draw[orange,->,thick] (u) -- (t1);
    \draw[black,->,thick] (u) -- (t2);
    \draw[orange,dashed,->,thick] (up) -- (t3);
    \draw[orange,dashed,->,thick] (u) -- (t3p);
    \draw[orange,->,thick] (w1) -- (v);
    \draw[black,->,thick] (w2) -- (v);
    \draw[orange,dashed,->,thick] (w3p) -- (v);
    \draw[orange,dashed,->,thick] (w3) -- (vp);
    \end{tikzpicture}
    \hfill
    \begin{tikzpicture}[xscale=.8,yscale=.4,inner sep=1.5]
    \node (u) at (0, 0) {$u$};
    \node (up) at (0, -1) {$u'$};
    \node (v) at (1, 0) {$v$};
    \node (vp) at (1, -1) {$v'$};
    \node (t1) at (-1, 1) {$t_1$};
    \node (t2) at (-1, 0) {$t_2$};
    \node (t3) at (-1, -3) {$x_1$};
    \node (w1) at (2, 1) {$w_1$};
    \node (w2) at (2, 0) {$w_2$};
    \node (w3) at (2, -3) {$y_1$};
    \draw[orange,->,thick] (u) -- (v);
    \draw[orange,->,thick] (u) -- (t1);
    \draw[black,->,thick] (u) -- (t2);
    \draw[orange,dashed,->,thick] (up) -- (t3);
    \draw[orange,->,thick] (w1) -- (v);
    \draw[black,->,thick] (w2) -- (v);
    \draw[orange,dashed,->,thick] (w3) -- (vp);
    \end{tikzpicture}
    \hfill
    \begin{tikzpicture}[xscale=.8,yscale=.4,inner sep=1.5]
    \node (u) at (0, 0) {$u$};
    \node (up) at (0, 1) {$u'$};
    \node (v) at (1, 0) {$v$};
    \node (vp) at (1, -1) {$v'$};
    \node (t1) at (-1, 1) {$t_1$};
    \node (t2) at (-1, 0) {$t_2$};
    \node (t3) at (-1, -2) {$x_1$};
    \node (w1) at (2, 1) {$w_1$};
    \node (w2) at (2, 0) {$w_2$};
    \node (w3) at (2, -2) {$y_1$};
    \draw[orange,->,thick] (u) -- (vp);
    \draw[orange,->,thick] (up) -- (v);
    \draw[orange,->,thick] (u) -- (t1);
    \draw[black,->,thick] (u) -- (t2);
    \draw[orange,dashed,->,thick] (u) -- (t3);
    \draw[orange,->,thick] (w1) -- (v);
    \draw[black,->,thick] (w2) -- (v);
    \draw[orange,dashed,->,thick] (w3) -- (v);
    \end{tikzpicture}
    \caption{Illustration of the proof of Claim~\ref{clm:clean}: an instance $I$
    with a tight edge~$e$ with incident garbage facts (dashed orange edges),
    copy facts (orange edges), and extra facts (black edges); the dissociation
    $I_1$ of the edges containing garbage facts; the result $I_2$ of merging the
    dangling copies on~$(u,v)$; the dissociation of~$e$ in~$I$.}
    \label{fig:clean}
  \end{figure}
\end{toappendix}

\begin{proof}
  In this proof, we use the conventions on dissociation introduced in
  Appendix~\ref{apx:basic}. The process of the proof is illustrated in
  Figure~\ref{fig:clean}.

  By definition of $\Theta$, we know that~$Q$ has a model~$I$ with a tight
  edge~$e$  of weight~$\Theta$. Let us modify~$I$ to obtain a model~$I'$
  where~$e$ is a tight edge of weight~$\Theta$ which is clean. To do so,
  consider all the elements $x_1, \ldots, x_n$ with which $e$ has left garbage
  facts, and the elements~$y_1, \ldots, y_m$ with which $e$ has right garbage
  facts; note that these sets may not be disjoint (though they are on
  Figure~\ref{fig:clean}).

  The edges $(x_1,v), \ldots,
  (x_n,v)$ and $(u,y_1), \ldots, (u,y_m)$ have weight~$<\Theta$, so they can be
  dissociated without breaking the query (using Claim~\ref{clm:dissocweight} if
  they are non-leaf, or trivially if they are leaf using the conventions of
  Appendix~\ref{apx:basic}). Let $I_1$ be the result of performing all these
  dissociations: we know that $I_1$ satisfies the query. In $I_1$, the edge~$e$ still has garbage facts because of the
  edge copies created in the dissociations (called \emph{dangling copies}), but all the elements with which it
  has garbage facts are now leaf elements. Also note that in $I_1$ relative
  to~$I$ we have created other dangling edges (on the elements $x_1, \ldots,
  x_n$ and $y_1, \ldots, y_m$).

  Now, let $I_2$ be the result of homomorphically mapping the dangling copies
  into~$e$, which is possible because their covering facts are (up to renaming)
  a subset of those of~$e$. We know that $I_2$ satisfies the query.
  Further, by construction the edge~$e$ 
  does not have any garbage facts in~$I_2$, and obviously it still has weight~$\Theta$.

  The only remaining point is to show that 
  $e$ is still tight in~$I_2$. To see why, let $I_2'$ be the
  result of dissociating~$e$ in~$I_2$ (i.e., isomorphic to~$I_2$ if~$e$ is a
  leaf edge, though this case can in fact never happen), and let~$I'$ be the result of dissociating~$e$ in~$I$: we know
  that $I'$ violates~$Q$ because $e$ is tight in~$I$. Now, we claim that~$I_2'$
  has a homomorphism to~$I'$, defined by mapping the two copies of the
  edge~$(u,v)$ dissociated in~$I_2$ to the two copies of that edge in~$I'$, and
  mapping the other dangling edges created when going from~$I$ to~$I_1$ to the
  incident edges of~$(u,v)$ that were dissociated to create them. Thus, as $I'$
  does not satisfy the query, neither does~$I_2'$. Thus $e$ is indeed tight
  in~$I_2$, in particular it must be non-leaf. This concludes the proof.
\end{proof}

\subparagraph*{Extra weight.}
We further restrict tight edges~$e$ by limiting their number of incident facts,
similarly to the notion of \emph{side weight} in~\cite{amarilli2021dichotomy}.
However, in this paper, we additionally
partition the incident facts between so-called \emph{extra facts} and \emph{copy
facts}. Intuitively, our reductions will use codings that introduce copies of
the edge~$e$, and the \emph{extra facts} are those that can be
``distinguished'' from incident copies of~$e$ added in codings; by contrast copy
facts are non-unary facts in edges that are isomorphic copies of~$e$ and therefore ``indistinguishable''.

We want to minimize the number of extra facts, to intuitively ensure that they
are all ``necessary'', in the sense that a copy of~$e$ missing an incident extra
fact can be dissociated. Let us formally define the extra facts: among the
non-garbage incident facts, they are those
that are part of a so-called \emph{triangle} (i.e., involve an element occurring
both in a left-incident in a
right-incident fact), those which are unary, or those which are a covering fact of an edge whose covering facts are not
isomorphic to the covering facts of~$e$.

\begin{definition}
    Let $I$ be an instance with an edge $e = (u, v)$, and
    let~$C \subseteq I$ be the covering facts of~$e$. An element~$w
    \in \dom(I)$ forms a \emph{triangle} with~$e$ 
    if both~$(u, w)$ and $(v, w)$ are edges.

    Let $(u', v')$ be
  some edge of~$I$. We call $(u', v')$ a \emph{copy} of~$(u, v)$ if the
  covering facts of~$(u', v')$ are isomorphic to $C$ by the
  isomorphism mapping~$u'$ to~$u$ and~$v'$ to~$v$.

  We partition the non-garbage left-incident facts of~$(u, v)$ between:
  \begin{itemize}
    \item The \emph{left copy facts}, i.e., the binary facts involving~$u$
      and an element~$v'$ such that $(u,v')$ is a copy of~$(u,v)$ and $v'$ does
      \emph{not} form a triangle with~$e$: we call~$v'$ a \emph{left copy
      element} of~$e$.
    \item The \emph{left extra facts}, which comprise all other non-garbage
      left-incident facts, namely:
      \begin{itemize}
        \item The unary facts on~$u$.
        \item The non-garbage binary facts involving~$u$ and some element~$x$
          such that:
          \begin{itemize}
            \item the element~$x$ forms a triangle with~$e$; or
            \item the covering facts of the edge~$(u,x)$ are not isomorphic to~$C$.
              \end{itemize}
      \end{itemize}
  \end{itemize}
  We partition the non-garbage right-incident facts into \emph{right extra facts} and
  \emph{right copy facts} with
  \emph{right copy elements} in a similar way.
  Note that, as we prohibit triangles, the
  left copy elements and right copy elements are disjoint.
  We talk of the \emph{copy elements},
  \emph{copy facts}, \emph{extra facts}
  of~$e$ to denote both the left and right kinds.
\end{definition}

\begin{example}
  \label{exa:minimal}
  Consider the instance of Figure~\ref{fig:minimal} and the edge
  $e = (u,v)$. The covering facts $C$ of~$e$ are represented as an orange edge,
  and the other orange edges represent edges which are copies of~$e$.
  The left and right copy elements are respectively $t_1$ and $t_2$
  and $w_1$ and $w_2$. The dashed orange edges represent edges whose covering
  facts are a strict subset of~$C$,
  i.e., they are garbage facts. The extra facts include
  unary facts (not pictured), facts with $x_1$ (the black edge $(u, x_1)$ represents
  non-garbage facts not isomorphic to~$C$), and facts with $x_2$, $x_3$, and
  $x_4$ (which form triangles).
\end{example}

Note that garbage facts are neither extra facts nor copy facts, and are ignored
in the definition above except in that
they may help form triangles. This does not matter: thanks to Claim~\ref{clm:clean},
garbage facts will only appear in intermediate steps of some proofs.
We can now define the
\emph{critical extra weight} as the minimal extra weight of a tight edge with
weight~$\Theta$:

\begin{definition}
  The \emph{critical extra weight} of~$Q$, written $\Xi \geq 0$, is
  the minimum across all models~$I$ of~$Q$ and tight edges~$e$ of~$I$ of
  weight~$\Theta$, of the number of extra facts of~$e$ in~$I$.
\end{definition}

\begin{example}
  Continuing Example~\ref{exa:weight}, the query~$Q'$ had critical extra
  weight~$2$, as witnessed by~$I'$. The query $Q'': R(x), S(x, y), S(x', y), S(x', y'),
  T(y')$, has critical weight~$1$ and critical extra weight~$0$, as
  witnessed by the model $I'' = \{R(a), S(a, b), S(a', b), S(a', b'), T(b')\}$
  where the edge $(a', b)$ is tight and has weight~$1$ and extra weight~$0$.
\end{example}

Again, the definition of critical extra weight clearly ensures:

\begin{claim}
  \label{clm:dissocextra}
  Let $I$ be a model of~$Q$ and $e = (u, v)$ be a non-leaf edge. If $e$ has
  weight~$\Theta$ and extra weight $< \Xi$, then the dissociation of~$e$ in~$I$
  is also a model of~$Q$.
\end{claim}

\subparagraph*{Lexicographic weight.}
We then impose a third minimality requirement on tight edges~$e$, which is needed in
Section~\ref{sec:iter} (but unused in Section~\ref{sec:noniter}). The intuition is
that we want to limit the number of copy elements.
Specifically, we
minimize first the
number~$\tau$ of left copy elements,
then the number~$\omega$ of right copy elements, hence the name
\emph{lexicographic weight}. This is why, when choosing a tight edge, we also choose
an orientation (i.e., choosing $(u,v)$ as a tight edge is different from choosing $(v,u)$):

\begin{definition}
  Let $I$ be an instance with an edge $e = (u, v)$. Let $\tau$ be the
  number of left copy elements and~$\omega$ be the number of right copy
  elements of~$e$. The \emph{lexicographic weight} of~$e$ is the ordered pair $(\tau,
  \omega)$. We order these ordered pairs lexicographically, i.e., $(\tau, \omega) <
  (\tau', \omega')$ with $\tau, \tau', \omega, \omega' \in \NN$ iff $\tau <
  \tau'$ or $\tau = \tau'$ and $\omega < \omega'$.

  The \emph{critical lexicographic weight} $\Lambda$ of~$Q$ is the minimum, over
  all models~$I$ of~$Q$ and all tight edges of~$e$ with weight~$\Theta$ and
  extra weight~$\Xi$, of the lexicographic weight of~$e$.
\end{definition}

Note that minimizing the lexicographic weight does not always minimize the
total number of copy facts\footnote{Minimizing the total number of copy facts,
or minimizing along the componentwise partial order on $\NN \times \NN$, would
suffice almost everywhere in the proof except in part of Section~\ref{sec:iter}.},
e.g., $(1, 3) < (2, 1)$
but $1+3 > 2+1$. However, it is always the case that removing a
copy fact of an edge~$e$ causes the lexicographic weight of~$e$ to decrease
(and does not cause the extra weight to increase, as the remaining covering
facts of the edge are garbage facts).

Again, we have:

\begin{claim}
  \label{clm:dissoclex}
  Let $I$ be a model of~$Q$ and $e = (u, v)$ be a non-leaf edge with
  weight~$\Theta$, extra weight $\Xi$, and lexicographic weight $<\Lambda$. 
  Then, the dissociation of~$e$ in~$I$ is also a model of~$Q$.
\end{claim}

\subparagraph*{Critical models.}
We now define \emph{critical models} (significantly refining
the so-called \emph{minimal tight patterns} of~\cite{amarilli2021dichotomy}).
A critical model $I$ is intuitively a
model of~$Q$ with a clean tight edge~$e$ that achieves the minimum of our three weight
criteria, and where we additionally impose that~$I$ is subinstance-minimal. For
convenience we also specify a choice of incident facts in the critical model,
but this choice is arbitrary, i.e., we can pick any pair of a left-incident fact and right-incident fact.

\begin{definition}
  A \emph{critical model} $(I, e, F_\LL, F_\RR)$ is a model~$I$ of~$Q$ which is subinstance-minimal,
  a clean tight edge~$e$ of~$I$ having weight~$\Theta$, extra weight
  $\Xi$, and lexicographic weight~$\Lambda$, and a left-incident
  fact~$F_\LL \in I$ and a right-incident fact~$F_\RR \in I$ of~$e$.
\end{definition}

We can now claim that critical models exist:

\begin{propositionrep}
  \label{prp:critical}
  If a query $Q$ has a model with a tight edge, then it has a critical model.
\end{propositionrep}

\begin{proofsketch}
  The existence of models with tight edges achieving the critical weights is by
  definition, cleanliness can be imposed by the process used to prove Claim~\ref{clm:clean}, and
  subinstance-minimality can easily be imposed by picking
  some minimal subset of facts of the model that satisfy the query.
\end{proofsketch}

  \begin{toappendix}
    To prove the proposition, we first prove a lemma:
    \begin{lemma}
      \label{lem:fctremove}
      Let $I$ be a model of~$Q$ and let $e$ be a non-leaf clean edge of~$I$. Let
      $J$ be a subinstance of~$I$. Then:

      \begin{itemize}
        \item If we removed in~$J$ a covering fact of~$e$ in~$I$, then the weight of~$e$
      in~$J$ is less than the weight of~$e$ in~$I$.
        \item Otherwise, the weight of~$e$ in~$J$ is the same as the weight
          of~$e$ in~$I$. Now, if we removed in~$J$ an extra fact of~$e$ in~$I$,
          then the extra weight of~$e$ in~$J$ is less than the extra weight
          of~$e$ in~$I$.
        \item Otherwise, the extra weight of~$e$ in~$J$ is the same as the
          extra weight of~$e$ in~$I$. Now, if we removed in~$J$ a copy fact
          of~$e$ in~$I$, then the lexicographic weight of~$e$ in~$J$ is less
          than the lexicographic weight of~$e$ in~$I$.
        \item Otherwise, the lexicographic weight of~$e$ in~$J$ is the same as
          the lexicographic weight of~$e$ in~$I$, and in fact we did not remove
          any incident fact of~$e$ in~$I$.
      \end{itemize}
    \end{lemma}
    \begin{proof}
      It is clear that removing facts cannot increase the weight, and that the
      weight decreases strictly iff some covering fact is removed, so the first
      bullet point and the first part of the second bullet point are immediate. 
      (Note that removing an extra fact may create a new copy element, and
      increase the lexicographic weight while the extra weight is decreased, but
      this does not contradict the statement.)
      From now on, we assume that no covering fact was removed.

      For the extra weight, we claim that the extra facts of~$e$ in~$J$ are also
      extra facts of~$e$ in~$I$. This implies that removing facts did not cause
      the extra weight to increase and removing an extra fact of~$e$ in~$I$
      must have strictly decreased the extra weight. We consider the possible
      extra facts of~$e$ in~$J$:
      \begin{itemize}
        \item The unary facts on~$u$ or~$v$ in~$J$ are also unary facts on~$u$
          or~$v$ in~$J$.
        \item For the facts involving one of~$u$ or~$v$ in~$J$ along with some
          element~$x$ achieving a triangle with~$e$ in~$J$, the element~$x$ must
          also achieve a triangle with~$e$ in~$J$, so they are also extra facts
          of~$e$ in~$I$. In other words, removing facts can never cause new
          triangles to appear.
        \item For the facts involving one of~$u$ or~$v$ in~$J$ along with an
          element~$x$ such that the set $C'$ of covering facts of the edge are not
          isomorphic to the set~$C$ of covering facts of~$e$ in~$I$, we know
          that in~$I$ the covering facts of the edge are some (not necessarily
          strict) superset of~$C'$. Here it is important that we excluded
          garbage facts in the definition of extra facts: as $C'$ is in fact not
          isomorphic to a subset of~$C$, a superset of~$C'$ cannot be isomorphic
          to~$C$.
      \end{itemize}
      Thus we have established the second part of the second bullet point and
      the first part of the third bullet point. From now on, we assume that no
      extra fact was removed.
      
      For the lexicographic weight, it is now clear that, if no covering fact or
      extra fact of~$e$ is removed, then the lexicographic weight decreases iff
      we remove a copy fact of~$e$ (this may create garbage facts, which are not
      accounted in the extra weight or lexicographic weight). This establishes
      the second part of the third bullet point, and the first part of the
      fourth bullet point.

      The second part of the fourth bullet point is because we assumed that $e$
      was clean, so all its incident facts are extra facts or copy facts.
    \end{proof}

  We are now ready to prove Proposition~\ref{prp:critical}:
  \end{toappendix}

  \begin{proof}[Proof of Proposition~\ref{prp:critical}]
    We know by hypothesis that $Q$ has a
    model with a tight edge $e = (u, v)$, and by definition of~$\Theta$
    and~$\Xi$ and~$\Lambda$ there is such
  a model $I$ with a tight edge~$e$ having weight~$\Theta$ and extra
  weight~$\Xi$ and lexicographic weight~$\Lambda$.  We use in this proof the
    conventions introduced in Appendix~\ref{apx:basic}.

    Up to modifying~$I$ with the process of Claim~\ref{clm:clean}, we can ensure
    that~$e$ is clean in~$I$. The only additional point to verify is that the
    process in the proof does not change the extra weight or lexicographic weight
    of~$e$. Indeed, as~$e$ is still a tight non-leaf edge after the process, we
    know that the extra weight did not decrease (otherwise
    Claim~\ref{clm:dissocextra} would apply, contradicting tightness),
    and dissociating edges involving
    incident garbage facts and merging them in~$e$ cannot have caused the extra weight
    to increase. Thus, the extra weight is still~$\Xi$. Further, the
    lexicographic weight did not decrease (otherwise Claim~\ref{clm:dissoclex}
    would apply, contradicting tightness),
    and again the process cannot have caused the lexicographic
    weight to increase. Hence, we can additionally ensure that $e$ is clean.

    To achieve subinstance-minimality,
    let $J$ be a minimal subset of~$I$ which still satisfies the query. By
    definition, $J$ is subinstance-minimal. Consider the dissociation $J'$ of~$e$
    in~$J$, in which $(u,v)$ is no longer an edge, and which can add dangling
    copies of~$e$ as $(u,v')$ or $(u',v)$ with $u'$ and $v'$ leaf elements.
    (Specifically, unless $(u,v)$ in~$J$ is an edge which is non-leaf, then the
    dissociation does nothing except renaming elements, following our
    conventions). We can map~$J'$ to the dissociation~$I'$ of~$e$ in~$I$, with
    the identity homomorphism extended to map the dangling copies on the
    dangling copies added in~$I'$ relative to~$I$ by the dissociation. Thus,
    $J'$ does not satisfy the query, whereas $J$ does. Hence, $(u,v)$ is still a
    leaf in~$I$, and it is a non-leaf edge which is tight.

    We now claim that the covering facts of~$e$ in~$J$ are the same as in~$I$.
    Indeed, assuming the contrary, the covering facts of~$e$ in~$J$ would be a
    strict subset of the covering facts of~$e$ in~$I$, so the weight of~$e$
    in~$J$ would be strictly less than~$\Theta$. Thus, as $e$ is non-leaf, by
    Claim~\ref{clm:dissocweight}, we could dissociate~$e$ in~$J$ and obtain a
    model~$J'$ of~$Q$, contradicting the fact that it is tight.

    Last, we claim that the incident facts of~$e$ in~$J$ are the same as in~$I$.
    Indeed, assuming otherwise, the incident facts of~$e$ in~$J$ would be a strict subset
    of those of~$e$ in~$I$. By Lemma~\ref{lem:fctremove}, as~$e$ is clean, removing incident
    facts of~$e$ without removing covering facts of~$e$ must reduce the extra
    weight, or keep the extra weight unchanged and reduce the lexicographic
    weight. Thus, as $e$ is a non-leaf edge in~$J$, by Claim~\ref{clm:dissocextra} or by
    Claim~\ref{clm:dissoclex}, we could dissociate the edge without breaking the
    query, contradicting the fact that~$e$ is tight in~$J$.

    Thus, in particular the extra weight (resp., lexicographic weight) of~$e$
    in~$J$  is the same as the extra weight (resp., lexicographic weight) of~$e$
    in~$I$. Thus we have established that $I$ is a subinstance-minimal model
    of~$Q$ with a tight clean edge~$e$ having the right weight, extra weight,
    and lexicographic weight, which concludes the proof.
  \end{proof}

\section{Hardness with a Non-Iterable Critical Model}
\label{sec:noniter}
\begin{toappendix}
  \label{apx:noniter}
\end{toappendix}

\begin{figure}
  \begin{subfigure}[t]{.3\linewidth}
    \begin{tikzpicture}[inner sep=1,xscale=.9]
      \node (u) at (.5, 0) {$u$};
      \node (v) at (1.5, 0) {$v$};
      \draw[thick,orange] (u) -- (v);
      \node (x1) at (-.25, 1) {$x_1$};
      \node (x2) at (1, .5) {$x_2$};
      \node (x3) at (2.25, 1) {$x_3$};
      \draw[thick,blue] (u) -- (x1);
      \draw[thick,blue] (u) -- (x2);
      \draw[thick,olive] (v) -- (x2);
      \draw[thick,olive] (v) -- (x3);

      \draw[thick,black] (x1) -- (x2);
      \draw[thick,black] (x2) -- (x3);
      \draw[thick,black] (x1) -- (x3);

      \node (t1) at (-1, 0) {$t_{1}$};
      \node (t2) at (-1, -1) {$t_{2}$};
      \draw[purple] (u) -- (t1.east);
      \draw[purple] (u) -- (t2.east);
      \path[purple] (t1.south)  edge     (t2.north);

      \node (w1) at (3, 0) {$w_{1}$};
      \node (w2) at (3, -1) {$w_{2}$};
      \draw[purple] (v) -- (w1.west);
      \draw[purple] (v) -- (w2.west);
      \path[purple] (w1.south)  edge     (w2.north);
      \path[purple] (t2.east) edge [bend right = 15] (w2.west);
      \draw[purple] (x1) -- (t1.east);

      \draw[purple] (x3) -- (w1.west);
    \end{tikzpicture}

    \caption{Example critical model~$M$}
    \label{fig:critnonit}
  \end{subfigure}
  \hfill
  \begin{subfigure}[t]{.3\linewidth}

    \begin{tikzpicture}[inner sep=1,xscale=.9]
      \node (u1) at (.5, 0) {$u\phantom{'}$};
      \node (v1) at (1.5, 0) {$v'$};
      \node (u2) at (.5, -1) {$u'$};
      \node (v2) at (1.5, -1) {$v\phantom{'}$};
      \draw[thick,orange] (u1) -- (v1);
      \draw[thick,orange] (u2) -- (v1);
      \draw[thick,orange] (u2) -- (v2);
      \node (x1) at (-.25, 1) {$x_1$};
      \node (x2) at (1, .5) {$x_2$};
      \node (x3) at (2.25, 1) {$x_3$};
      \draw[thick,blue] (u1) -- (x1);
      \draw[thick,blue] (u1) -- (x2);
      \draw[thick,olive] (v1) -- (x2);
      \draw[thick,olive,dashed] (v1) -- (x3);
      \draw[thick,blue,dashed] (u2) -- (x1);
      \draw[thick,blue] (u2) -- (x2);
      \draw[thick,olive] (v2) -- (x2);
      \draw[thick,olive] (v2) -- (x3);

      \draw[thick,black] (x1) -- (x2);
      \draw[thick,black] (x2) -- (x3);
      \draw[thick,black] (x1) -- (x3);

      \node (t1) at (-1, 0) {$t_{1}$};
      \node (t2) at (-1, -1) {$t_{2}$};
      \draw[purple] (u1.west) -- (t1.east);
      \draw[purple] (u1.west) -- (t2.east);
      \draw[purple] (u2.west) -- (t1.east);
      \draw[purple] (u2.west) -- (t2.east);
      \path[purple] (t1.south)  edge     (t2.north);

      \node (w1) at (3, 0) {$w_{1}$};
      \node (w2) at (3, -1) {$w_{2}$};
      \draw[purple] (v1.east) -- (w1.west);
      \draw[purple] (v1.east) -- (w2.west);
      \draw[purple] (v2.east) -- (w1.west);
      \draw[purple] (v2.east) -- (w2.west);
      \path[purple] (w1.south)  edge     (w2.north);
      \path[purple] (t2.east) edge [bend right = 15] (w2.west);
      \draw[purple] (x1) -- (t1.east);

      \draw[purple] (x3) -- (w1.west);
    \end{tikzpicture}

    \caption{Iteration of~$M$}
    \label{fig:iternonit}
  \end{subfigure}
  \hfill
  \begin{subfigure}[t]{.36\linewidth}
    \begin{tikzpicture}[inner sep=1,yscale=.85]
      \node (u1) at (.5, 0) {$u_1$};
      \node (v1) at (1.5, 0) {$v_1$};
      \node (u2) at (.5, -1) {$u_2$};
      \node (v2) at (1.5, -1) {$v_2$};
      \draw[thick,orange] (u1) -- (v1);
      \draw[thick,orange] (u1) -- (v2);
      \draw[thick,orange] (u2) -- (v2);
      \node (x1) at (-.25, 1) {$x_1$};
      \node (x2) at (1, .5) {$x_2$};
      \node (x3) at (2.25, 1) {$x_3$};
      \draw[thick,blue] (u1) -- (x1);
      \draw[thick,blue] (u1) -- (x2);
      \draw[thick,olive] (v1) -- (x2);
      \draw[thick,olive] (v1) -- (x3);
      \draw[thick,blue] (u2) -- (x1);
      \draw[thick,blue] (u2) -- (x2);
      \draw[thick,olive] (v2) -- (x2);
      \draw[thick,olive] (v2) -- (x3);

      \draw[thick,black] (x1) -- (x2);
      \draw[thick,black] (x2) -- (x3);
      \draw[thick,black] (x1) -- (x3);

      \node (t11) at (-1, 0) {\footnotesize $t_{1,1}$};
      \node (t12) at (-1, -.2) {\footnotesize $t_{1,2}$};
      \node (t13) at (-1, -.4) {\footnotesize $t_{1,3}$};
      \node (t21) at (-1, -1) {\footnotesize $t_{2,1}$};
      \node (t22) at (-1, -1.2) {\footnotesize $t_{2,2}$};
      \node (t23) at (-1, -1.4) {\footnotesize $t_{2,3}$};
      \draw[purple] (u1.west) -- (t11.east);
      \draw[purple] (u1.west) -- (t12.east);
      \draw[purple] (u1.west) -- (t13.east);
      \draw[purple] (u1.west) -- (t21.east);
      \draw[purple] (u1.west) -- (t22.east);
      \draw[purple] (u1.west) -- (t23.east);
      \draw[purple] (u2.west) -- (t11.east);
      \draw[purple] (u2.west) -- (t12.east);
      \draw[purple] (u2.west) -- (t13.east);
      \draw[purple] (u2.west) -- (t21.east);
      \draw[purple] (u2.west) -- (t22.east);
      \draw[purple] (u2.west) -- (t23.east);
      \path[purple] (t11.west)  edge   [bend right=20]  (t21.west);
      \path[purple] (t12.west)  edge   [bend right=20]  (t22.west);
      \path[purple] (t13.west)  edge   [bend right=20]  (t23.west);

      \node (w11) at (3, 0) {\footnotesize $w_{1,1}$};
      \node (w12) at (3, -.2) {\footnotesize $w_{1,2}$};
      \node (w13) at (3, -.4) {\footnotesize $w_{1,3}$};
      \node (w21) at (3, -1) {\footnotesize $w_{2,1}$};
      \node (w22) at (3, -1.2) {\footnotesize $w_{2,2}$};
      \node (w23) at (3, -1.4) {\footnotesize $w_{2,3}$};
      \draw[purple] (v1.east) -- (w11.west);
      \draw[purple] (v1.east) -- (w12.west);
      \draw[purple] (v1.east) -- (w13.west);
      \draw[purple] (v1.east) -- (w21.west);
      \draw[purple] (v1.east) -- (w22.west);
      \draw[purple] (v1.east) -- (w23.west);
      \draw[purple] (v2.east) -- (w11.west);
      \draw[purple] (v2.east) -- (w12.west);
      \draw[purple] (v2.east) -- (w13.west);
      \draw[purple] (v2.east) -- (w21.west);
      \draw[purple] (v2.east) -- (w22.west);
      \draw[purple] (v2.east) -- (w23.west);
      \path[purple] (w11.east)  edge   [bend left=20]  (w21.east);
      \path[purple] (w12.east)  edge   [bend left=20]  (w22.east);
      \path[purple] (w13.east)  edge   [bend left=20]  (w23.east);
      \path[purple] (t21.east) edge [bend right = 15] (w21.west);
      \path[purple] (t22.east) edge [bend right = 15] (w22.west);
      \path[purple] (t23.east) edge [bend right = 15] (w23.west);
      \draw[purple] (x1) -- (t11.east);
      \draw[purple] (x1) -- (t12.east);
      \draw[purple] (x1) -- (t13.east);

      \draw[purple] (x3) -- (w11.west);
      \draw[purple] (x3) -- (w12.west);
      \draw[purple] (x3) -- (w13.west);
    \end{tikzpicture}

    \caption{$3$-saturated coding $I_{G,3}$ in~$M$ of $G = (\{1, 2\},
    \{(1,1), (1,2),(2,2)\})$}
    \label{fig:nonitercode}
  \end{subfigure}
  \caption{Examples of Section~\ref{sec:noniter} and illustration of the
  notation}
\end{figure}

Having defined critical models, we now start our hardness proof. As
in~\cite{amarilli2021dichotomy}, we will distinguish two cases, based on whether
we can break~$Q$ with an \emph{iteration} process on a critical
model.

\begin{definition}
  Let $M = (I, e, F_\LL, F_\RR)$ be a critical model, let $e = (u, v)$, and
  let $C$ be the covering facts of~$e$.
  Let $A$ and $B$ be the set of the left-incident and right-incident facts
  of~$e$ in~$I$, respectively.
  The \emph{iteration} of~$M$ is obtained by modifying~$I$ in the following way:
  \begin{itemize}
    \item Add fresh elements $u'$ and~$v'$, copy~$e$ 
      on $(u, v')$, $(u', v')$, $(u', v)$, and remove the facts
      of~$C$.
    \item Create a copy of the facts of~$A \setminus \{F_\LL\}$ where we
      replace~$u$ by~$u'$.
    \item Create a copy of the facts of~$B \setminus \{F_\RR\}$ where we replace~$v$ by~$v'$.
  \end{itemize}
\end{definition}

\begin{example}
  Consider the critical model in Figure~\ref{fig:critnonit}, with edge
  $(u,v)$ and where $F_\LL$
and $F_\RR$ are binary facts respectively using $u$ and $x_1$ and $v$ and~$x_3$.
  Its iteration is shown in
Figure~\ref{fig:iternonit}, with dashed edges representing edges
where~$F_\LL$ and~$F_\RR$ are missing.
\end{example}

A \emph{non-iterable} critical model~$M$ is one whose iteration no longer
satisfies the query; otherwise $M$ is \emph{iterable}.
In this section, we show hardness when there is a non-iterable critical model:

\begin{proposition}
  \label{prp:pp2dnf2}
  Assume that $Q$ has a non-iterable critical model.
  Then the uniform reliability problem for~$Q$ is \#P-hard.
\end{proposition}

We prove this result in the rest of this section.

\subparagraph*{Fixing notation.}
Fix the critical model $M = (I, e, F_\LL, F_\RR)$ and let $e = (u, v)$ be the tight
clean edge. 
We must introduce some notation to talk about the incident facts of~$e$ in~$I$,
which is summarized in Figure~\ref{fig:critnonit}. As $e$ is clean, we know that its
incident facts are either extra facts or copy facts --- there are no garbage
facts.

Let $C \subseteq I$ be the covering facts of~$e$ in~$I$ (in orange on the
picture), with $|C| = \Theta$.
Let~$X = \{x_1, \ldots, x_k\}$ be the elements different from~$u$ and~$v$
with which one of~$u$ or~$v$ has a (non-unary) extra
fact or has one of the two facts~$F_\LL$ and~$F_\RR$. Note that some of the
elements in~$X$
may have facts with both~$u$ and~$v$ (i.e., triangles), like~$x_2$ in the
picture. We may have $k=0$, specifically when $F_\LL$ and $F_\RR$ are unary
facts and any other extra facts are unary.

Further let $T = \{t_1, \ldots, t_\tau\}$ be the left copy elements of~$e$ not
in~$X$, and let $W = \{w_1, \ldots, w_\omega\}$ be the right copy elements of~$e$ not
in~$X$, with $T$ and~$W$ disjoint (because copy elements cannot form
triangles). We exclude elements of~$X$ because, if $F_\LL$ (resp., 
$F_\RR$) is a copy fact, then $X$ contains exactly one left copy element (resp.,
exactly one right copy element)\footnote{
Because of this, in general $(\tau,\omega)$ may be less than the critical
lexicographic weight $\Lambda$.}.
Also note that we may have $\tau = \omega = 0$, i.e., if there
are no copy facts except possibly those of the edges of~$F_\LL$ and
of~$F_\RR$.

To recapitulate, the incident facts of~$e$ in~$I$ only involve elements from~$X
\sqcup T \sqcup W$. Specifically, they are the unary facts on~$u$, the unary facts on~$v$,
the non-unary extra facts (which involve one of~$\{u,v\}$ and one element of~$X$), the facts~$F_\LL$
and~$F_\RR$ which respectively involve $u$ and~$v$ and (if they are non unary) one
element of~$X$, and the other left and right copy facts forming isomorphic
copies of~$e$ as edges $(u, t_j)$ with $1 \leq j \leq \tau$ and $(w_i, v)$ with
$1 \leq i \leq \omega$. Notice again how, if $F_\LL$ or $F_\RR$ are copy facts,
then these notations handle them as extra facts along with any other covering
facts of their edge.
Note that our description of the incident facts of~$e$ does not describe the 
facts that may exist between elements of $X \sqcup T \sqcup W$, and indeed these
may be arbitrary (some are pictured in Figure~\ref{fig:critnonit}).

\subparagraph*{Coding bipartite graphs.}
We will reduce from our variant of \#PP2DNF (Definition~\ref{def:pp2dnf}) by
using~$M$ to code 
a bipartite graph $G = (U \sqcup V, E)$.
Intuitively, we will create
one copy $u_i$ of~$u$ for each vertex $i$ of~$U$, one copy $v_j$ of~$v$ for each vertex
$j$ of~$V$,
and copy the edge~$e$ on $(u_i, v_j)$ for each edge $(i,j)$ of $E$.
The reason why we distinguish $X$ and $T$ and $W$ is because
we will handle them differently.
For the incident facts of~$e$ that are unary or involve elements of~$X$, we 
will create one single copy of them for each $u_i$ and each $v_j$. Indeed, we
 will show that edges $(u_i,v_j)$ that are missing one such incident fact can be 
dissociated (if an extra fact is
missing, using Claim~\ref{clm:dissocextra}) or mapped in a specific way in the
iteration (if one of $F_\LL$ or $F_\RR$ is a copy fact and we are missing one of
the covering facts of their edge). For the (copy) facts involving~$T \sqcup W$, we
will copy them (using the fact that they are binary)
by creating a large number~$q$ of copies of $T \sqcup W$. This
\emph{saturation} process will
in fact create a large number of copies of all facts involving some element
of~$T \sqcup W$, which we call
the \emph{saturated facts}.

Let us accordingly define the \emph{saturated coding} of a bipartite graph
in~$M$:

\begin{definition}
  \label{def:satcoding}

  Let $G = (U \sqcup V, E)$ be a non-empty bipartite graph,
  and assume without loss of
  generality that $U = \{1, \ldots, n\}$ and $V = \{1, \ldots, m\}$.

  Let $q > 0$ be some integer.
  The \emph{$q$-saturated coding} of~$G$ in~$M$, written~$I_{G,q}$, is the instance
  defined by modifying~$I$ in the following way:
  \begin{itemize}
    \item For all $1 \leq p \leq q$, create fresh elements $T_p = \{t_{1,p},
      \ldots, t_{\tau,p}\}$ and $W_p = \{w_{1,p},
      \ldots, w_{\omega,p}\}$.
      Identify~$t_j = 
      t_{j,1}$ for $1 \leq j \leq \tau$ and $w_i = w_{i,1}$ for $1 \leq i \leq
      \omega$.
    \item Letting $\Phi$ be the set of the saturated facts, 
      for each $1 \leq p \leq q$, create a copy of~$\Phi$ where each element
      $t_j$ is replaced by~$t_{j,p}$ and each element
      $w_i$ is replaced by~$w_{i,p}$.
    \item Create elements $u_1, \ldots, u_n$ and $v_1, \ldots, v_n$,
      where we identify $u=u_1$ and $v=v_1$.
    \item Create a copy of all incident facts of~$e$ for all~$u_i$ and
      $v_j$. Formally, let $A$ and $B$ be the set of the left-incident and
      right-incident facts of~$e$ \emph{in the current
      model} (i.e., involving the $t_{j,p}$ and $w_{i,p}$): note that $A$
      (resp., $B$) contains in particular $F_\LL$ (resp., $F_\RR$) and any unary facts on~$u$
      (resp., on~$v$).
      For each $1 \leq i
      \leq n$, create a copy of the facts of~$A$ replacing~$u$ by~$u_i$,
      and for each $1 \leq j \leq m$ create a copy of the facts of~$B$ replacing~$v$ by~$v_j$.
    \item Copy~$e$ (i.e., $C$) on $(u_i,v_j)$ for each
      $(i,j) \in E$, and remove the facts of~$C$ if $(u_1,v_1) \notin E$.
  \end{itemize}
\end{definition} 

The saturated coding process is illustrated in Figure~\ref{fig:nonitercode}. Note that
the process is in polynomial time if the value~$q$ is polynomial in the
size~$|G|$ of the input bipartite graph.

\subparagraph*{Understanding the coding.}
Letting $G = (U \sqcup V, E)$ be a non-empty bipartite graph and writing $U = \{1,
\ldots, n\}$ and $V = \{1, \ldots, m\}$,
we study the coding $I_{G,q}$ 
to relate subsets of~$I_{G,q}$ to subsets of~$U \times E \times V$. For this, we
partition the facts of~$I_{G,q}$ in five kinds (see
Figure~\ref{fig:nonitercode}):

\begin{itemize}
  \item The \emph{base facts} (pictured in black), which are the facts that do not involve any of
    the elements $u_1, \ldots, u_n$, $v_1, \ldots, v_m$ or any element of
    $\bigsqcup_{1\leq p \leq q} T_p \sqcup W_p$ (but they may involve elements
    of~$X$). These facts are precisely the facts of~$I$ that do not involve the elements~$u$
    or~$v$ or any element of~$T \sqcup W$, and they are unchanged in the coding.
  \item The \emph{saturated facts} (in purple), i.e., the facts involving some element
    of~$T_p \sqcup W_p$ for some~$1 \leq p \leq q$. These facts exist in~$q$ copies, %
    and some (corresponding to facts of~$I$ between $u$ or $v$ and an
    element of $T \sqcup W$) have been further copied $n$ times (if they
    involve~$u$)
    or $m$ times (if they involve~$v$).
  \item The \emph{non-saturated left-incident facts} (in blue) of each vertex $i
    \in U$, which are the facts which involve $u_i$ and do not involve the
    $T_p \sqcup W_p$, i.e., are unary or involve an element of~$X$. These facts 
    include in particular one copy of~$F_\LL$.
  \item The \emph{non-saturated right-incident facts} (in green)
    of each vertex $j \in V$,
    that involve~$v_j$ and not the $T_p \sqcup W_p$, i.e., are unary
    or involve an element of~$X$; they include one copy of~$F_\RR$.
  \item The \emph{copy of~$e$} (in orange) for each edge $(i,j) \in E$,
    which is on the edge~$(u_i,v_j)$ of~$I_{G,q}$.
\end{itemize}

The last three kinds are what we are interested in for the reduction,
but the first two kinds need to be dealt with. We will show that the base facts
must all be present to satisfy the query, and that each edge has some copy of
the saturated facts with high probability.

\subparagraph*{Base facts.}
We say that a subinstance of~$I_{G,q}$ is \emph{well-formed} if
all base facts are present, and \emph{ill-formed} if at least one is missing. 
The following is easy to see by subinstance-minimality of~$I$:

\begin{propositionrep}
  \label{prp:badill2}
  The ill-formed subinstances do not satisfy the query.
\end{propositionrep}

\begin{proof}
  Let $F$ be some missing base fact in the ill-formed subinstance~$I'$. Consider the function from~$I_{G,q}$ to~$I$
  that identifies all copies of~$u$, and all copies of each $t_j$ and $w_i$.
  This is a homomorphism, and the preimage of~$F$ is the single fact~$F$. Hence,
  $I'$, as a subinstance of~$I_{G,q} \setminus \{F\}$,
  has a homomorphism to~$I \setminus \{F\}$. As $I$ is subinstance-minimal, we
  conclude that $I \setminus \{F\}$ does not satisfy the query, hence~$I'$ also
  does not.
\end{proof}

Hence, the number of subinstances of~$G_{I,q}$ satisfying the query is the
number of well-formed subinstances that do. Thus, in the sequel, we 
only consider well-formed subinstances.

\subparagraph*{Saturated facts.}
For the saturated facts, we will intuitively define \emph{valid} subinstances where, 
for each 
ordered pair of vertices $(i,j) \in U \times V$, considering the copies $u_i$ and $v_j$
of~$u$ and~$v$, there is a complete copy of the saturated facts that are
``relevant'' to them.
More precisely, looking back at the original instance~$I$,
and considering the facts of~$I$ involving an
element of~$T \sqcup W$, there are of two types. The first type are the
facts that do not involve~$u$ or~$v$, i.e., they only
involve elements of~$T \sqcup W$ and possibly of $\dom(I) \setminus \{u, v\}$.
Each such fact has been copied
$q$ times in~$I_{G,q}$, and the copy numbered $1 \leq p \leq q$ uses one or two
elements of~$T_p \sqcup W_p$. The second type are the facts involving $u$ or~$v$
in~$I$ (they cannot involve both).
These facts have been copied $n\times q$ or $m\times q$ times in~$I_{G,q}$, each
copy using one element of~$T_p \sqcup W_p$ for some $1 \leq p \leq q$ and one 
$u_i$ for some $1 \leq i \leq n$ or one $v_j$ for some $1 \leq j \leq m$.
What we require of a valid subinstance $J \subseteq I_{G,q}$ is that, for each
pair of vertices $(i,j) \in U \times
V$, we have in $J$ some copy $1 \leq p \leq q$ containing all facts of the first
type and all facts of the second type involving $u_i$ and~$v_j$:

\begin{definition}
  We partition the saturated facts of~$I_{G,q}$ in~$q$ copies: formally, 
  the \emph{$p$-th saturated copy} for~$1 \leq p \leq q$ is the subset of the
  saturated facts of~$I_{G,q}$ that involve some element of~$T_p \sqcup W_p$.
  A \emph{saturation index} for~$I_{G,q}$ is a function $\iota\colon U \times V \to
  \{1, \ldots, q\}$.
 
  For $J \subseteq I_{G,q}$, we say that $J$ is
  \emph{valid} for~$\iota$ if, for each $(i,j) \in U\times V$, letting $p
  \colonequals \iota(i,j)$, considering the facts of the $p$-th
  saturated copy, $J$ contains all such facts that are:
  \begin{itemize}
    \item of the first type, i.e., $J$ contains all facts of~$I_{G,q}$ that
      involve some element
      of $T_p \sqcup W_p$ and do not
  involve any elements of $\{u_{i'} \mid 1 \leq i' \leq n\} \sqcup
  \{v_{j'} \mid 1 \leq j' \leq m\}$;
\item of the second type and involve $u_i$ or~$v_j$, i.e., $J$ contains all facts
  of~$I_{G,q}$ that involve some element of~$T_p \sqcup W_p$ and
  involve either~$u_i$ or~$v_j$.
  \end{itemize}

  We call $J$ \emph{valid} if there is
  a saturation index for which it is valid; otherwise $J$ is \emph{invalid}.
\end{definition}

Note that, for 
each choice of ordered pair $(i,j) \in U \times V$, the required facts can be found in a different
saturated copy $\iota(i,j)$, i.e., we do not require that there is a~$p$ such that $J$ contains all facts of the
$p$-th saturated copy. Indeed this stronger requirement
would be too hard to ensure: intuitively, the number of
facts required for each $(i,j)$ is constant (it only depends on~$I$),
but the number of facts in the $p$-th saturated copy depends on~$G$ (it is
linear in $|U|\times |V|$).

We now show that we can pick a number~$q$ of copies which is polynomial in the
input~$G$, but makes it very unlikely that a
random subinstance is invalid. Thanks to this, we do not need to know which ones
of the invalid subinstances satisfy~$Q$. Indeed, the proportion of subinstances
of~$I_{G,q}$ that satisfy~$Q$ will be the proportion of valid subinstances
that do, up to an error which is much less than the probability of any valid
subinstance and can be eliminated by rounding:

\begin{lemmarep}
  \label{lem:negligible}
  There is a polynomial $P_M$ depending on the critical model~$M$ such that, for 
  any non-empty bipartite graph~$G = (U \sqcup V, E)$,
  letting $\chi \colonequals |U| + |V| + |E|$ be the size of~$G$ and defining $q
  \colonequals P_M(\chi)$, the proportion of subinstances of~$I_{G,q}$ that are
  invalid is strictly less than $2^{-(\chi |I|+1)}$.
\end{lemmarep}

  \begin{proof}
    We define $n \colonequals |U|$ and $m \colonequals |V|$. The first step in our proof is to change
    slightly the notion of valid subinstances to the more stringent
    requirement of a \emph{routine} subinstance, where we partition the
    copies of saturated facts among the ordered pairs of vertices of~$U \times V$. The
    intuition is that each ordered pair of~$U \times V$ will look for witnessing facts
    in a different subset of the copies, which will simplify the analysis. 

    Formally, let us first define $q = n \times m \times q'$, with $q'$ to be defined
    later.
  From this definition of~$q$, we will label the saturated copies by talking
    about the $(i, j, p)$-th saturated copy with $1 \leq i \leq
  n$, $1 \leq j \leq m$, and $1 \leq p \leq q'$. 
    Further, the copies of the elements $t_1, \ldots, t_\tau$ of~$T$ 
    will be indexed as $t_{j', i, j,
    k}$ with $1 \leq j' \leq \tau$, $1 \leq i \leq n$, $1 \leq j \leq m$,
    and $1 \leq p \leq q'$, and likewise the copies of the 
    $w_1, \ldots, w_\omega$ of~$W$ 
    will be indexed as $w_{i', i, j, k}$ with $1 \leq i' \leq \omega$, $1 \leq i \leq n$, $1 \leq j \leq m$,
    and $1 \leq p \leq q'$.
    The intuition is that,
    when we consider one ordered pair $(u_i, v_j)$ with $1 \leq i \leq n$ and $1 \leq j
    \leq m$, we will only consider the $(i,j,p)$-th saturated copies for $1 \leq p
    \leq q'$. In other words, we will only consider saturation indexes
    $\iota\colon U \times V \to \{1, \ldots, q\}$ which, seeing them by abuse
    of notation as functions of type $\iota\colon U \times V \to U \times V \times \{1, \ldots,
    q'\}$, ensure that for all $(i,j) \in U \times V$ we have $\iota(i,j) =
    (i,j,p)$ for some $1 \leq p \leq q'$.
    We say that $J \subseteq I_{G,q}$ is \emph{routine} if there is a
    saturation index of this restricted form for which it is valid.

    Given that routine subinstances are in particular valid, to show our lower
    bound on the probability of valid subinstances, it suffices to show it on
    the routine subinstances. Thus, we only study routine subinstances in
    the rest of this proof. (The only
    reason why we define valid subinstances rather than routine subinstances
    in the main text is that the definition of valid is somewhat
    easier to present.)

    We will adopt the probabilistic perspective in which facts of~$I_{G,q}$ have
    probability~$q$, and will talk of the \emph{probability} that a subinstance
    is routine instead of the proportion of routine subinstances. The
    probability that~$J \subseteq I_{G,q}$ is routine is the probability that it
    admits a saturation index of the restricted form above for which it is
    valid, i.e., for each $(i,j) \in U\times V$ there is a choice of $1 \leq p
    \leq q'$ such
    the $(i,j,p)$-th saturated copy contains all facts that involve $u_i$ or
    involve~$v_j$ or involve no
    element of the form $u_i$ or~$v_j$.
    The definition of routine subinstances (unlike valid
    subinstances) ensures that these events across the $(i,j) \in U \times V$
    are independent, because they talk about of facts in disjoint
    subsets of the saturated copies, i.e., disjoint subsets of facts. Hence, the
    probability that a subinstance is routine is an independent conjunction
    requiring, for all $(i,j) \in U \times V$,
    that there is a choice of $p$ for this $i$ and~$j$.

    Further, for a choice of $1 \leq i \leq n$ and $1 \leq j \leq m$, the
    existence of a suitable~$1 \leq p \leq q'$ is a disjunction across $q'$ disjoint
    saturated copies, i.e., disjoint subsets of facts again. So, the probability
    that a suitable $1 \leq p \leq q'$ exists for some $(i,j) \in U \times V$
    is a disjunction of independent probabilistic events, each of which for $1
    \leq p \leq q'$ states that in the $(i,j,p)$-th saturated copy all necessary facts
    are present.

    Last, for a choice of $1 \leq i \leq n$ and $1 \leq j \leq n$ and $1 \leq p
    \leq q'$, the probabilistic event in question 
    is the conjunction of the presence of all the
    required facts, i.e., again a conjunction of independent probabilistic events.
    Each fact has probability $1/2$, and the number of such facts is at
    most~$|I|$, because the facts of a saturated copy that either involve $u_i$
    or involve~$v_j$ or involve no copy of~$u$ and~$v$ are in bijection
    with the facts of~$I$ involving some element of $T \sqcup W$, i.e., a subset
    of~$I$, thus having cardinality at most~$|I|$.

    Thus, the probability of getting a non-routine subinstance is the
    conjunction of $n \times m$ independent events, each of which is the
    disjunction of $q'$ independent events, each of which is the conjunction of
    at most $|I|$ events having probability~$1/2$. Thus, the probability of
    non-routine subinstances is at most:
  \[
    1 - (1 - (1 - 2^{-|I|})^{q'})^{|U| \times |V|}
  \]
  We can now use Lemma~\ref{lem:proba}, which is specifically intended for this
  purpose (see Appendix~\ref{apx:prelim}), with $\zeta \colonequals 2^{-|I|}$,
  and $q \colonequals q'$ and $\chi \colonequals |U| \times |V|$.
  The lemma shows that, for the above quantity to be less than
  $2^{-(\chi |I|+1)}$, or more stringently less than $\epsilon \colonequals \left(2^{-(|I|+1)}\right)^\chi$, we can
  take:
  \[ q' \colonequals 1 + \left\lfloor\frac{\ln(\chi) - \chi \ln\left(2^{-(|I|+1)}\right)}{-\ln(1 -
  2^{-|I|})}\right\rfloor.
  \]
  Remembering that~$|I|$ is a constant, this value is bounded by a polynomial in~$\chi$, thus the same is
  true of~$q$. So indeed we can define a polynomial~$P_M$ giving us a
  suitable~$q$ as a function of $\chi = |U|\times|V|$, and for this $q$ the
  probability of non-routine subinstances, hence of invalid subinstances, is
  less than~$\epsilon$, which concludes.
\end{proof}

Thanks to this, we focus on the well-formed
subinstances~$J$ where we keep some subset of the saturated facts making~$J$
valid.
We now fix $q$ to the value
of Lemma~\ref{lem:negligible}, and build $I_{G,q}$ in polynomial time in the
input bipartite graph~$G$ (with the critical model~$M$ being fixed).

\subparagraph*{Good and bad subinstances.}
Let us now study the status of the last
three kinds of facts:

\begin{definition}
  \label{def:complete}
  Let $J \subseteq I_{G,q}$.
  For $1 \leq i \leq n$ (resp., $1 \leq j \leq m$), the
  vertex~$i \in U$ (resp., $j \in V$) is
  \emph{complete} in~$J$ if all its non-saturated left-incident facts (resp.,
  non-saturated right-incident facts) are present in~$J$, and \emph{incomplete}
  otherwise.
  The edge $(i, j) \in E$ is \emph{complete} in~$J$ if all covering facts of
  $(u_i, v_j)$ in~$I_{G,q}$ are present in~$J$, and \emph{incomplete}
  otherwise.
  We call $J$ \emph{good} if there is an edge~$(i,j)\in E$ with $(i,j)$, $i$, and $j$ complete, and \emph{bad} otherwise.
\end{definition}

We now claim that, among the well-formed valid subinstances, the good ones
satisfy the query, and the bad ones do not. This is easy to see for good
subinstances:

\begin{propositionrep}
  \label{prp:good2}
  For any good valid well-formed subinstance $J\subseteq I_{G,q}$,
  there is a homomorphism from~$I$ to~$J$.
\end{propositionrep}
\begin{proofsketch}
  As $J$ is well-formed all base facts are present, and $J$ is valid for some
  saturation index~$\iota$.
  Let $(i,j) \in E$ be an edge witnessing that~$J$ is good.
  The homomorphism maps $T \sqcup W$ to $T_{\iota(i,j)} \sqcup W_{\iota(i,j)}$,
  maps~$u$ to~$u_i$ and~$v$ to~$v_j$, and is the identity otherwise.
\end{proofsketch}

\begin{proof}
  Let $(i, j) \in E$ be an edge witnessing that $J$ is good, i.e., $i$ and $j$ are
  complete and $(i, j)$ is complete, let $\iota$ be a saturation index for
  which~$J$ is valid, and let $p \colonequals \iota(i,j)$.

  Define a homomorphism $h$ to map $u$ and~$v$ to~$u_i$ and~$v_j$, to map $T
  \sqcup W$ to $T_p \sqcup W_p$, and to be the identity otherwise. Let us
  consider all facts of~$I$:
  \begin{itemize}
    \item The facts that do not involve $u$, $v$, or $T \sqcup W$ are unchanged
      by the homomorphism and are mapped to base facts, which are all present
      because $J$ is well-formed.
    \item For the facts that involve $T \sqcup W$, we know by definition of a
      saturation index that all facts of the $p$-th saturated copy that either
      involve~$u_i$ or involve~$v_j$ or involve no copy of~$u$ and~$v$ 
      are present in~$J$, and these are in one-to-one correspondence with the
      facts of~$T \sqcup W$ involving~$u$ and~$v$ in~$I$, i.e., they define
      suitable images for these facts.
    \item For the facts that involve $u$ and not~$v$ and possibly elements
      of~$X$, they are mapped by~$h$ to the non-saturated left-incident facts
      of~$u_i$, which are present because $i$ is complete. The same applies to
      the facts involving~$v$ and not~$u$ because~$j$ is complete.
    \item For the facts of~$I$ that involve~$u$ and~$v$, i.e., the covering
      facts of~$u$ and~$v$, they are mapped by~$h$ to the copy of~$e$
      on~$(u_i,v_j)$, which contains all these facts because $(i,j)$ is
      complete.
  \end{itemize}
\end{proof}

For bad subinstances, we show with much more effort that they do not satisfy the
query:

\begin{propositionrep}
  \label{prp:bad2}
  Any bad subinstance~$J \subseteq I_{G,q}$ does not satisfy the query.
\end{propositionrep}

\begin{proofsketch}
  It suffices to study the case with no saturation, i.e., $q = 1$.
  We dissociate incomplete edges with Claim~\ref{clm:dissocweight}, and dissociate
  complete edges missing at least one incident extra fact with
  Claim~\ref{clm:dissocextra}, which does not break~$Q$. Then we show how
  to map this homomorphically to the iteration $I'$ of~$M$, by mapping complete
  vertices to~$u$ and~$v$ in the dissociation, and mapping the vertices which are missing facts
  of the edges of~$F_\LL$ or~$F_\RR$
  to~$u'$ and~$v'$ respectively (after dissociating these edges if~$F_\LL$
  or~$F_\RR$ are copy facts). This contradicts the assumption that $M$ was
  non-iterable, i.e., that $I'$ violates~$Q$.
\end{proofsketch}

\begin{proof}
  The definition of the $q$-saturated coding clearly ensures that for any~$q >
  0$ the instance $I_{G,q}$
  has a homomorphism $h$ to~$I_{G,1}$, obtained by mapping the elements
  of~$T_p \sqcup W_p$ to~$T_1 \sqcup W_1$. Note now that if a subinstance
  of~$I_{G,q}$ is bad then its image by~$h$ is also a bad subinstance
  of~$I_{G,1}$, because the homomorphism only merges saturated facts. So it
  suffices to consider the case~$q=1$, i.e.,
  show that any bad subinstance $J \subseteq I_{G,1}$ does not satisfy the
  query, which we do in the sequel. For convenience we identify $T_1 \sqcup W_1$ with $T \sqcup
  W$ in the notation.

  \emph{[While not necessary for the proof, it may help the reader to assume that all
  the saturated facts are present in the subinstances that we consider in the
  proof.
  The intuitive reason why the presence or absence of the saturated facts does not matter is
  that they are copy facts of the edges $(u_i, v_j)$, so they only contribute to the
  lexicographic weight (or, if they are incomplete, consist of garbage facts
  that do not contribute to the weight at all); but the notion of the lexicographic
  weight will not intervene in this proof at all.]}

  We adopt in this proof the convention explained in Appendix~\ref{apx:basic}.
  Assume by contradiction that~$J$ satisfies the query.
  Our goal in the proof will be to modify~$J$ without breaking the query,
  intuitively by dissociating edges; and then map the result to the iteration,
  which was assumed to violate the query.

  \bigskip

  Recall the notion of vertices of~$U$ and~$V$ and edges of~$E$ being
  \emph{complete}.

  In the case where $F_\LL$ is a copy fact of~$e$ in~$I$ and not an extra fact
  of~$e$ in~$I$, we must distinguish two kinds of incomplete vertices of~$U$:
  the \emph{extra incomplete} and the \emph{$F_\LL$-incomplete}. In this case,
  letting $x_\LL$ be the other element than~$u$ used in~$F_\LL$, we call $i \in U$
  \emph{extra incomplete} if it is missing a non-saturated left-incident fact
  corresponding to a left extra fact of~$e$ in~$I$, i.e., a unary fact on~$u$ or
  a binary fact between~$u$ and some element of~$X \setminus \{x_\LL\}$. We call $i
  \in U$ \emph{$F_\LL$-incomplete} if it is only missing non-saturated
  left-incident facts that are binary facts between~$u$ and~$x_\LL$. When~$F_\LL$ is
  an extra fact of~$e$ in~$I$, we consider that all incomplete $i \in U$ are
  extra-incomplete.

  Likewise, we distinguish the incomplete $j \in V$ between the \emph{extra
  incomplete} and \emph{$F_\RR$-incomplete}: if $F_\RR$ is a copy fact of~$e$
  in~$I$, letting $x_\RR$ be the other element than~$v$ that it uses, then the
  \emph{$F_\RR$-incomplete} $j \in V$ are the incomplete $j \in V$ missing only
  facts between~$x_\RR$ and~$v$, and the \emph{extra-incomplete} $j \in V$ are the
  other ones; if $F_\RR$ is an extra fact of~$e$ in~$I$ then all incomplete $j \in
  V$ are extra-incomplete.

  \emph{[It may help the reader to understand that there are really four cases in the
  proof, depending on whether $F_\LL$ is an extra fact or a copy fact, and
  whether $F_\RR$ is an extra fact or a copy fact. The presentation of the
  definitions and of the proof is designed so that, in the interest of brevity, all four cases are handled at
  once. Intuitively, the easiest case of the proof is when both $F_\LL$ and
  $F_\RR$ are extra facts, in which case all incomplete vertices are
  extra-incomplete, and all edges $(u_i,v_j)$ where one of the
  vertices $i$ and $j$ is incomplete can be dissociated because they have extra
  weight $< \Xi$. This makes it easy to
  define the homomorphism to the iteration. By
  contrast, if $F_\LL$ and $F_\RR$ are copy facts, in particular when $\Xi = 0$
  so that all incident facts of~$e$ in~$I$ are copy facts, then the incomplete
  left vertices $i$ and right vertices $j$ must be missing some facts involving the same
  elements as~$F_\LL$ and $F_\RR$ respectively. Then the argument is that their
  edge with~$x_\LL$ and $x_\RR$ has weight $< \Theta$ and can be dissociated, and
  this dissociation allows us to map the elements $u_i$ and $v_j$ corresponding
  to incomplete left vertices $i$ and right vertices $j$ to the elements $u'$
  and $v'$ of the iteration.]}

  \bigskip

  The first step is to dissociate incomplete copies of~$e$ on ordered pairs
  $(u_i,v_j)$ for $(i,j) \in E$ because they have weight $< \Theta$.
  For any incomplete edge~$(i,j) \in E$, consider the ordered pair $(u_i,v_j)$. It is
  either a non-edge or leaf edge and can be dissociated without breaking the
  query, or it has weight $< \Theta$ and can be dissociated by
  Claim~\ref{clm:dissocweight} without breaking the query. Let $J_1$ be the
  result of performing these dissociations on~$J$: it still satisfies the query.
  In $J_1$, each ordered pair $(u_i,v_j)$ is either an edge of
  weight~$\Theta$ (if $(i,j)$ is complete) or a non-edge (otherwise), i.e.,
  relative to~$J$, the
  covering facts of the ordered pairs $(u_i,v_j)$ for incomplete edges $(i,j)$ have
  been removed. In exchange, we have added leaf edges involving some of the
  $u_i$ and~$v_j$ together with a fresh vertex, whose covering facts are (up to
  renaming) a
  (strict) subset of~$C$ (i.e., they will intuitively be garbage facts).
  We call these \emph{dangling edges}, and we say that a
  dangling edge is \emph{attached} to the element of $\{u_i \mid i \in U\}
  \sqcup \{v_j \mid j \in V\}$ that occurs in its covering facts.

  The second step is to get dissociate complete copies of~$e$ relating vertices $u_i$
  and~$v_j$ such that one of~$i$ and~$j$ is extra-incomplete, because they have  
  extra weight $< \Xi$. For any complete edge~$(i,j) \in E$ where one of~$i$
  and~$j$ is extra-incomplete, consider the edge $(u_i,v_j)$. It has weight~$\Theta$.
  Its incident facts are the following:
  \begin{itemize}
    \item Covering facts of other edges of the
  form~$(u_i,v_{j'})$ or~$(u_{i'},v_j)$, whose covering facts are an isomorphic
      copy of some (not necessarily strict) subset of~$C$,
      so they are accounted for in the extra weight or
      not at all. Note that by construction they do not achieve triangles, so they are not reflected in the extra weight of~$e$.
    \item Dangling edges created in the first step, but these consist of a
      strict subset of~$C$ (up to renaming), so they are garbage facts and are not
      reflected in the weight at all.
    \item Covering facts of edges of the form $(u_i,t_{j'})$ or $(w_{i'},v_j)$
      with elements of $T_1 \sqcup W_1$, which are a subset of the covering
      facts of the same edges in~$I$, i.e., copies of the edge~$e$ that do not
      achieve triangles. Hence, these facts are either accounted in the
      lexicographic weight (if all facts are present) or not at all (if they are
      garbage facts), and in all cases they are not reflected in the extra
      weight of~$e$.
    \item The non-saturated left-incident facts of~$u_i$ and the non-saturated
      right-incident facts of~$v_j$; as one of $i$ and $j$ is extra incomplete, one of them is missing which
      corresponds to an extra fact of~$e$ in~$I$. So these are the only facts
      that define the extra
      weight of $(u_i,v_j)$, and we can therefore see that the extra weight is strictly less
      than~$\Xi$.
  \end{itemize}
  Thus, the edge $(u_i,v_j)$ is either non-leaf and can vacuously be
  dissociated without breaking the query, or it has weight~$\Theta$ and extra weight~$< \Xi$ and can be
  dissociated without breaking the query by Claim~\ref{clm:dissocextra}. Let $J_2$ be the result of performing
  these dissociations: $J_2$ satisfies the query, and
  in~$J_2$ compared to~$J_1$ all edges where one of the
  endpoints is extra-incomplete have been dissociated.

  The third step is only necessary in the case where~$F_\LL$ is a copy fact of~$e$
  in~$I$, and letting $x_\LL$ be the other element than~$u$ used in~$F_\LL$, 
  it consists of dissociating the edges $(u_i,x_\LL)$ where $i$ is
  $F_\LL$-incomplete, because they have weight $< \Theta$.
  Formally, let us consider all $F_\LL$-incomplete $i \in U$, and consider the
  ordered pair $(u_i,x_\LL)$. Either it is
  a non-edge or leaf edge and can be vacuously dissociated without breaking the
  query, or as $i \in U$ is $F_\LL$-incomplete its weight is $<\Theta$, so by
  Claim~\ref{clm:dissocweight} again it can be dissociated without breaking the
  query. Let $J_3$ be the result of performing these dissociations: in~$J_3$
  compared to~$J_2$ for all $F_\LL$-incomplete $i \in U$ there is no edge
  $(u_i,x_\LL)$. If $F_\LL$ is an extra fact of~$e$ in~$I$, we simply let $J_3 =
  J_2$ and the latter requirement is vacuously true as there are no $F_\LL$-incomplete $i
  \in U$ at all. In both cases, $J_3$ satisfies the query.

  The fourth step is the symmetric of the third step: if $F_\RR$ is a copy fact
  of~$e$ in~$I$, letting $x_\RR$ be the other element than~$v$ that it uses, 
  for each $F_\RR$-incomplete $j \in V$, we dissociate the edge $(x_\RR,v_j)$
  without breaking the query because it has weight $<\Theta$, and let $J_4$ be
  the result; otherwise if $F_\RR$
  is an extra fact
  of~$e$ in~$I$ we let $J_4 \colonequals J_3$.
  In both cases, $J_4$ satisfies the query.

  \bigskip

  We are now ready to define a homomorphism~$h$ from~$J_4$ to the iteration~$I'$
  of the critical model~$M$, which by assumption does not satisfy the query,
  giving us the desired contradiction. Define $h$ in the following way:
  \begin{itemize}
    \item We define $h$ to be the identity on the elements not in $T \sqcup W
      \sqcup \{u_i \mid i \in I\} \sqcup \{v_j \mid j \in I\}$ that are
      in~$\dom(I)$, i.e., are not leaf elements of dangling edges created in one
      of the four steps. The facts in the induced subinstance on these elements
      is exactly the induced subinstance of~$I$ on the same elements, i.e.,
      $\dom(I) \setminus (T \sqcup W \sqcup \{u, v\})$, so this correctly maps
      the facts (intuitively these correspond to the base facts).
    \item We define $h$ to be the identity on~$T \sqcup W$, as we know that the
      facts of~$J_4$ involving $T \sqcup W$ but not the $u_i$ and $v_j$ are a subset of the facts of~$I$
      involving $T \sqcup W$ but not $u$ and $v$ (intuitively these correspond
      to the saturated facts of the first type, i.e., those that do not use the $u_i$ and $v_j$)
    \item We map the $u_i$ and $v_j$ in the following way:
      \begin{itemize}
        \item If $i$ is complete or extra-incomplete, we map $u_i$ to~$u$.
        \item If $i$ is $F_\LL$-incomplete, we map~$u_i$ to~$u'$.
        \item If $j$ is complete or extra-incomplete, we map $v_j$ to~$v$.
        \item If $j$ is $F_\RR$-incomplete, we map $v_j$ to~$v'$.
      \end{itemize}
    \item We map the dangling edges in the following way:
      \begin{itemize}
        \item For dangling edges created in the first step, their covering edges
          are a (strict) subset of~$C$ up to renaming, so if they are attached
          to an element~$v_j$ then if $h(v_j) = v$ we map the leaf element
          to~$u'$ otherwise $h(v_j) = v'$ and we map it to~$u$; and if they are attached to an element~$u_i$
          then if $h(u_i) = u$ we map the leaf element to~$v'$ otherwise $h(u_i)
          = u'$ and we map it to~$v$.
        \item For dangling edges created in the second step, their covering
          edges are an isomorphic copy of~$C$ up to renaming, and they are again
          attached to the $u_i$ and $v_j$ so we proceed in the same way as the
          previous bullet point.
        \item For dangling edges created in the third step, their covering edges
          are a (strict) subset of~$C$ again. The copy attached to an $u_i$ is
          dealt with as above, and for the copy attached to~$x_\LL$ we map the
          leaf element to~$u$.
        \item For the fourth step, we proceed in the same way, mapping the leaf
          element of the copy attached to~$x_\RR$ to~$v$.
      \end{itemize}
  \end{itemize}

  We must show that~$h$ is indeed a homomorphism. Clearly the only important point
  is to show that the copies of~$e$ and their incident facts are correctly
  mapped. It is easy that the dangling edges are correctly mapped, so we do not
  consider them in what follows.

  Consider an ordered pair $(u_i,v_j)$. For the left-incident facts which are not edges
  of the form $(u_i,v_{j'})$, if $h(u_i) = u$ then correctness is easy to see: the
  left-incident facts in~$J_4$ are with the elements of~$T\sqcup W$ and the
  elements of~$X$ which are reflected in~$I'$ on~$u$. Otherwise, we know that $i$ was $F_\LL$-incomplete meaning that
  in the third step we ensured that there was no edge $(u_i,x_\LL)$ in~$J_4$ so
  that we can indeed map all incident facts: relative to~$I$, the edge $(u',v')$
  in~$I'$ is only missing the left-incident fact corresponding to~$F_\LL$ up to
  renaming, i.e., between~$u'$ and~$x_\LL$ but in fact we have no fact to map to
  this edge. The reasoning for right-incident facts is symmetric.

  Now, consider the copies of~$e$ themselves. From the definition of~$I'$, the
  only point to verify is that we have no copy $(u_i,v_j)$ of~$e$ that remains in~$J_4$
  where $h(u_i) = u$ and $h(v_j) = v$, i.e., each of $i$ and $j$ is complete or
  extra-incomplete. Remember that the copies of~$e$ on ordered pairs~$(u_i,v_j)$
  that remain after
  the first step only correspond to edges $(i,j)$ which are complete, and that
  those that remain after the second step only correspond to vertices $i$ and $j$ for
  which none are extra-incomplete. Hence, if $(u_i,v_j)$ is a copy of~$e$ in~$J$
  then $i$ and $j$ are both complete and so is the edge $(i,j)$.
  Together with $(i,j)$ being complete, this would witness the fact that $J$ was
  a good subinstance of~$I_{G,1}$, which is impossible because $J$ was defined to be bad.

  Thus, we have shown a homomorphism from~$J_4$, which satisfies the query, to~$I'$,
  which does not, so we have reached a contradiction and the proof is finished.
\end{proof}

This establishes that the status of~$Q$ on a valid well-formed subinstance~$J$
depends on whether $J$ is good or bad, i.e., depends on which of the last three kinds
of facts were kept in~$J$. Now, the subsets of these facts are
clearly in correspondence with the subsets of $U \times E \times V$ for
the $\lambda,\mu,\nu$-\#PP2DNF problem
(see Definition~\ref{def:pp2dnf}),
for some choice of constant probabilities $\lambda,\mu,\nu$. Further, a subset of $U
\times E \times V$ is counted in $\lambda,\mu,\nu$-\#PP2DNF if and only if the
corresponding subset of the last three kinds of facts yields a good subinstance.
As the ill-formed subinstances are easy to count, and
the invalid ones are negligible, we can conclude the reduction and establish
Proposition~\ref{prp:pp2dnf2}. The full proof is in the appendix.

\begin{toappendix}
\begin{proof}[Proof of Proposition~\ref{prp:pp2dnf2}]
  We reduce from the problem $\lambda,\mu,\nu$-\#PP2DNF, with constant probabilities
  $\lambda$ and $\mu$ and~$\nu$ to be defined later.

  We are given as input
  a bipartite graph $G = (U \sqcup V, E)$, and let $\chi \colonequals |U| + |V|
  + |E|$. We assume without loss of generality
  that~$G$ is non-empty, otherwise the answer on this input instance is trivial. We let $n \colonequals
  |U| + |V| + |E|$, which
  is polynomial in the input. Let $q$ be the value given in
  Lemma~\ref{lem:negligible}, which is again polynomial in the input. Construct the
  coding~$I_{G,q}$ of~$G$ in~$I$ using the fixed critical model $M = (I, e, F_\LL, F_\RR)$,
  writing $e = (u,v)$.
  We want to show that the value computed by our oracle, i.e., the number of
  subinstances of~$I_{G,q}$ satisfying~$Q$, reveals what we wanted
  to compute about~$G$, namely, the total probability of subsets $U' \times E'
  \times V'$ of $U \times E \times V$ where $U' \times V' \cap E'$ is nonempty.
  Note that what our oracle returns is, up to renormalization, the
  probability of getting a subinstance of~$I_{G,q}$ that satisfies~$Q$.

  Remember that in the main text we partitioned the facts of~$I_{G,q}$ in five kinds.
  Let $\Psi_\LL$ (resp, $\Psi_\RR$) be the number of non-saturated left-incident
  facts (resp., non-saturated right-incident facts)
  of a vertex of~$U$ (resp., of a vertex of~$V$) in~$I_{G,q}$;
  equivalently, these are respectively the
  number of left-incident and right-incident facts of~$e$ in~$I$ not involving
  $T \sqcup W$.
  Further remember that, by Proposition~\ref{prp:badill2}, the subinstances where
  some base fact is missing, i.e., the ill-formed subinstances, do not satisfy
  the query; so we can assume that the base facts are present and restrict our
  attention to these subinstances, i.e., the well-formed subinstances.

  For the saturated facts, we have partitioned the subinstances between
  those that are valid and those that are invalid. So, restricting our attention
  to the well-formed subinstances, the probability
  returned by the oracle is:
  \begin{equation}
    \label{eqn:defo}
    O = \epsilon X + (1-\epsilon) Y
  \end{equation}
  where $\epsilon$ is the probability that a subinstance is invalid
  conditioned on the fact of being well-formed,
  $X$ is the probability that a subinstance satisfies the query conditioned on the fact of being well-formed and invalid, and 
  $Y$ is the probability that a subinstance satisfies the query conditioned
  on the fact of being well-formed and valid.
  Note that, because the base facts and saturated facts are disjoint, the
  probabilistic events ``the subinstance is valid'' and ``the subinstance is
  well-formed'' are independent, so in fact $\epsilon$ is also the probability
  that a subinstance is invalid without conditioning. Thus, we know by
  Lemma~\ref{lem:negligible} that:
  \begin{equation}
    \label{eqn:epsbound2}
    \epsilon < 2^{-(\chi |I|+1)}
  \end{equation}
  Now, we have $\Theta+\Psi_\LL+\Psi_\RR \leq |I|$ because,
  considering~$I$, the facts of~$\Theta$ are the covering facts of~$e$, the facts
  of~$\Psi_\LL$ are a subset of the left-incident facts of~$e$,
  and the facts of~$\Psi_\RR$ are a
  subset of the right-incident facts of~$e$. So we have:
  \begin{equation}
    \label{eqn:epsbound}
    \epsilon < 2^{-(\chi (\Theta+\Psi_\LL+\Psi_\RR)+1)}
  \end{equation}

  Now, for the value~$Y$, by Propositions~\ref{prp:good2} and~\ref{prp:bad2}, 
  we know that the probability that a valid and well-formed subinstance
  satisfies the query is precisely the probability that it is good, i.e., the
  probability that it has a complete edge connecting two complete vertices.
  The probability that $i\in U$ is complete is 
  $\lambda \colonequals 2^{-\Psi_\LL}$, the probability that $(i,j) \in E$ is complete is 
  $\mu \colonequals 2^{-\Theta}$,
  and the probability that $j\in V$ is complete is $\nu \colonequals
  2^{-\Psi_\RR}$, and there is a clear probability-preserving correspondence
  between the choices of subsets $(U', E', V') \subseteq U \times E \times V$
  and the choice of which left vertices, edges, and right vertices will be
  complete in a valid well-formed subinstance of~$I_{G,q}$.
  So the probability that a valid and well-formed subinstance satisfies~$Q$ is
  exactly
  the probability of obtaining a good triple $(U', E', V') \subseteq U \times E \times
  V$
  in the problem $\lambda,\mu,\nu$-\#PP2DNF, i.e., the answer that we wish to
  compute to conclude the reduction.
  Now, there are 
  $Y' \colonequals 2^{(|U| + |V| + |E|) \times (\Theta+\Psi_\LL + \Psi_\RR)}$
  possible triples overall in that problem,
  so the value~$Y$ is of the form $\frac{Y''}{Y'}$ with $0 \leq
  Y'' \leq Y'$, with~$Y''$ being the answer to the problem on~$G$, i.e., the value
  that we wish to recover to conclude the reduction.

  We now claim that we can recover~$Y$ from the oracle answer~$O$.
  To see why, note that we can
  rewrite Equation~\ref{eqn:defo} to:
  \[
    O = Y + \epsilon (X - Y)
  \]
  where $Y$ and~$X$ are conditional probabilities, so:
  \begin{equation}
    \label{eqn:xy}
    -1 \leq X - Y \leq 1
  \end{equation}
  We can multiply by~$Y'$ and get:
  \[
    Y' O = Y'' + \epsilon Y' (X - Y)
  \]
  We know that~$Y''$ is a number of subinstances, i.e., an integer. Now,
  remembering that $0 \leq \epsilon \leq 1$ because it is a probability, we
  have by Equation~\ref{eqn:epsbound} that as $Y' =
  2^{\chi (\Theta+\Psi_\LL+\Psi_\RR)}$
  we have $0 \leq \epsilon Y' < 1/2$ and by Equation~\ref{eqn:epsbound}
  we have
  $-1/2 < \epsilon Y' (X - Y) < 1/2$. This means that we can recover from~$Y' O$ the
  value~$Y''$ by rounding, hence we can recover~$Y$ from~$O$, which concludes the proof.
\end{proof}
\end{toappendix}

\section{Hardness when all Critical Models are Iterable}
\label{sec:iter}
\begin{toappendix}
  \label{apx:iter}
\end{toappendix}

In this last section, we show hardness in the case where all critical models are
iterable:

\begin{propositionrep}
  \label{prp:ustcon}
  Assume that $Q$ has a critical model and that all critical models of~$Q$ are
  iterable. Then the uniform reliability problem for~$Q$ is \#P-hard.
\end{propositionrep}

\begin{figure}
  \begin{subfigure}[b]{.29\linewidth}
    \centering

    \begin{tikzpicture}[xscale=1,yscale=.4,inner sep=1.5]
    \node (u) at (0, 0) {$u$};
    \node (v) at (1, 0) {$v$};
    \node (t) at (-1, 1) {$t$};
    \node (t1) at (-1, 0) {$t_1$};
    \node (t2) at (-1, -1) {$t_2$};
    \node (w1) at (2, 1) {$w_1$};
    \node (w2) at (2, 0) {$w_2$};
    \node (w) at (2, -1) {$w$};
    \draw[orange,->,thick] (u) -- (v);
    \draw[orange,->,thick] (u) -- (t);
    \draw[orange,->,thick] (u) -- (t1);
    \draw[orange,->,thick] (u) -- (t2);
    \draw[orange,->,thick] (w) -- (v);
    \draw[orange,->,thick] (w1) -- (v);
    \draw[orange,->,thick] (w2) -- (v);
    \end{tikzpicture}
    \medskip

    \begin{tikzpicture}[xscale=1,yscale=.75,inner sep=1.5]
    \node (u1) at (0, 1) {$u_1$};
    \node (v1) at (1, 1) {$v_1$};
    \node (u2) at (0, .33) {$u_2$};
    \node (v2) at (1, .33) {$v_2$};
    \node (u3) at (0, -.33) {$u_3$};
    \node (v3) at (1, -.33) {$v_3$};
    \node (u4) at (0, -1) {$u_4$};
    \node (v4) at (1, -1) {$v_4$};
    \node (t) at (-1, 1) {$t$};
    \node (t1) at (-1, 0) {$t_1$};
    \node (t2) at (-1, -1) {$t_2$};
    \node (w1) at (2, 1) {$w_1$};
    \node (w2) at (2, 0) {$w_2$};
    \node (w) at (2, -1) {$w$};
    \draw[orange,->,thick] (u1) -- (v1);
    \draw[orange,->,thick] (u2) -- (v1);
    \draw[orange,->,thick] (u2) -- (v2);
    \draw[orange,->,thick] (u3) -- (v2);
    \draw[orange,->,thick] (u3) -- (v3);
    \draw[orange,->,thick] (u4) -- (v3);
    \draw[orange,->,thick] (u4) -- (v4);
    \draw[orange,->,thick] (u1) -- (t);
    \draw[orange,->,thick] (u1) -- (t1);
    \draw[orange,->,thick] (u1) -- (t2);
    \draw[orange,->,thick] (u2) -- (t1);
    \draw[orange,->,thick] (u2) -- (t2);
    \draw[orange,->,thick] (u3) -- (t1);
    \draw[orange,->,thick] (u3) -- (t2);
    \draw[orange,->,thick] (u4) -- (t1);
    \draw[orange,->,thick] (u4) -- (t2);
    \draw[orange,->,thick] (w1) -- (v1);
    \draw[orange,->,thick] (w2) -- (v1);
    \draw[orange,->,thick] (w1) -- (v2);
    \draw[orange,->,thick] (w2) -- (v2);
    \draw[orange,->,thick] (w1) -- (v3);
    \draw[orange,->,thick] (w2) -- (v3);
    \draw[orange,->,thick] (w1) -- (v4);
    \draw[orange,->,thick] (w2) -- (v4);
    \draw[orange,->,thick] (w) -- (v4);
    \end{tikzpicture}
    \caption{Example critical model~$M$ (top), 4-step iteration (bottom)}
    \label{fig:iter}
  \end{subfigure}
\hfill
  \begin{subfigure}[b]{.4\linewidth}
    \begin{tikzpicture}[xscale=1,yscale=1,inner sep=1.5]
      \node (ur) at (0, 1.5) {$u_r$};
      \node (uar) at (0, .9) {$u_{ar}$};
      \node (urs) at (0, .3) {$u_{rs}$};
      \node (ua) at (0, -.3) {$u_a$};
      \node (uas) at (0, -.9) {$u_{as}$};
      \node (us) at (0, -1.5) {$u_s$};
      \node (varr) at (2, 1.5) {$v_{ar,r}$};
      \node (vrsr) at (2, 1) {$v_{rs,r}$};
      \node (vara) at (2, .5) {$v_{ar,a}$};
      \node (vrss) at (2, 0) {$v_{rs,s}$};
      \node (vasa) at (2, -.5) {$v_{as,a}$};
      \node (vass) at (2, -1) {$v_{as,s}$};
      \node (v) at (2.2, -1.5) {$v$};
    \node (t) at (-1.5, 1) {$t$};
    \node (t1) at (-1.5, 0) {$t_1$};
    \node (t2) at (-1.5, -1) {$t_2$};
    \node (w1) at (3.5, 1) {$w_1$};
    \node (w2) at (3.5, 0) {$w_2$};
    \node (w) at (3.5, -1) {$w$};
    \draw[orange,->,thick] (ur.east) -- (varr.west);
    \draw[orange,->,thick] (ur.east) -- (vrsr.west);
    \draw[orange,->,thick] (uar.east) -- (varr.west);
    \draw[orange,->,thick] (uar.east) -- (vara.west);
    \draw[orange,->,thick] (uas.east) -- (vass.west);
    \draw[orange,->,thick] (uas.east) -- (vasa.west);
    \draw[orange,->,thick] (urs.east) -- (vrsr.west);
    \draw[orange,->,thick] (urs.east) -- (vrss.west);
    \draw[orange,->,thick] (ua.east) -- (vara.west);
    \draw[orange,->,thick] (ua.east) -- (vasa.west);
    \draw[orange,->,thick] (us.east) -- (vrss.west);
    \draw[orange,->,thick] (us.east) -- (vass.west);
    \draw[black,->,thick] (us.east) -- (v.west);

    \draw[black,->,thick] (ur.west) -- (t);

    \draw[black,->,thick] (ur.west) -- (t1);
    \draw[black,->,thick] (ur.west) -- (t2);
    \draw[purple,->,thick] (ua.west) -- (t1);
    \draw[purple,->,thick] (ua.west) -- (t2);
    \draw[purple,->,thick] (us.west) -- (t1);
    \draw[purple,->,thick] (us.west) -- (t2);
    \draw[orange,->,thick] (uas.west) -- (t1);
    \draw[orange,->,thick] (uas.west) -- (t2);
    \draw[orange,->,thick] (urs.west) -- (t1);
    \draw[orange,->,thick] (urs.west) -- (t2);
    \draw[orange,->,thick] (uar.west) -- (t1);
    \draw[orange,->,thick] (uar.west) -- (t2);

    \draw[orange,->,thick] (w1) -- (varr.east);
    \draw[orange,->,thick] (w2) -- (varr.east);
    \draw[orange,->,thick] (w1) -- (vara.east);
    \draw[orange,->,thick] (w2) -- (vara.east);
    \draw[orange,->,thick] (w1) -- (vasa.east);
    \draw[orange,->,thick] (w2) -- (vasa.east);
    \draw[orange,->,thick] (w1) -- (vass.east);
    \draw[orange,->,thick] (w2) -- (vass.east);
    \draw[orange,->,thick] (w1) -- (vrsr.east);
    \draw[orange,->,thick] (w2) -- (vrsr.east);
    \draw[orange,->,thick] (w1) -- (vrss.east);
    \draw[orange,->,thick] (w2) -- (vrss.east);

    \draw[black,->,thick] (w1) -- (v.east);
    \draw[black,->,thick] (w2) -- (v.east);
    \draw[black,->,thick] (w) -- (v.east);
    \end{tikzpicture}
    \caption{The coding $I_G$ of a graph~$G$ in~$M$:\\ $G = (\{a,r,s\}, \{\{r,s\}, \{a,r\},
    \{a,s\}\})$.}
    \label{fig:itercode}
  \end{subfigure}
\hfill
    \begin{subfigure}[b]{.28\linewidth}
    \begin{tikzpicture}[xscale=1,yscale=.45,inner sep=1.5]
    \node (u) at (0, .5) {$u$};
    \node (v) at (1, -.5) {$v$};
    \node (up) at (0, -.5) {$u'$};
    \node (vp) at (1, .5) {$v'$};

    \node (t) at (-1, 1) {$t$};
    \node (t1) at (-1, 0) {$t_1$};
    \node (t2) at (-1, -1) {$t_2$};
    \node (w1) at (2, 1) {$w_1$};
    \node (w2) at (2, 0) {$w_2$};
    \node (w) at (2, -1) {$w$};

    \draw[orange,->,thick] (u) -- (vp);
    \draw[orange,->,thick] (up) -- (v);

    \draw[orange,->,thick] (u) -- (t);
    \draw[orange,->,thick] (u) -- (t1);
    \draw[orange,->,thick] (u) -- (t2);
    \draw[orange,->,thick] (up) -- (t1);
    \draw[orange,->,thick] (up) -- (t2);
    \draw[orange,->,thick,dashed] (up) -- (t);

    \draw[orange,->,thick] (w) -- (v);
    \draw[orange,->,thick] (w1) -- (v);
    \draw[orange,->,thick] (w2) -- (v);
    \draw[orange,->,thick] (w1) -- (vp);
    \draw[orange,->,thick] (w2) -- (vp);
    \draw[orange,->,thick,dashed] (w) -- (vp);
      \end{tikzpicture}
      \medskip

    \begin{tikzpicture}[xscale=1,yscale=.45,inner sep=1.5]
      \centering
    \node (u) at (0, 1) {$u$};
    \node (vp) at (1, 1) {$v'$};
      \node (u1) at (0, 0) {$u_{1}$};
      \node (u2) at (0, -1) {$u_{2}$};
    \node (up) at (0, -2) {$u'$};
      \node (v) at (1, -2) {$v\phantom{'}$};

    \node (t) at (-1, 1) {$t$};
    \node (t1) at (-1, -.5) {$t_1$};
    \node (t2) at (-1, -2) {$t_2$};
    \node (w1) at (2, 1) {$w_1$};
    \node (w2) at (2, -.5) {$w_2$};
    \node (w) at (2, -2) {$w$};

    \draw[orange,->,thick]  (u.east) -- (vp);
    \draw[orange,->,thick] (u1.east) -- (vp);
    \draw[orange,->,thick] (u2.east) -- (vp);
    \draw[orange,->,thick] (u1.east) --  (v);
    \draw[orange,->,thick] (u2.east) --  (v);
    \draw[orange,->,thick] (up.east) --  (v);

    \draw[orange,->,thick] (u.west) -- (t);
    \draw[orange,->,thick] (u.west) -- (t1);
    \draw[orange,->,thick] (u1.west) -- (t1);
    \draw[orange,->,thick] (u.west) -- (t2);
    \draw[orange,->,thick] (u2.west) -- (t2);
    \draw[orange,->,thick] (up.west) -- (t1);
    \draw[orange,->,thick] (up.west) -- (t2);

    \draw[orange,->,thick] (w.west) -- (v);
    \draw[orange,->,thick] (w1.west) -- (v);
    \draw[orange,->,thick] (w2.west) -- (v);
    \draw[orange,->,thick] (w1.west) -- (vp);
    \draw[orange,->,thick] (w2.west) -- (vp);
      \end{tikzpicture}
\caption{Fine dissociation (top) and explosion (bottom) of~$M$}
\label{fig:other}
    \end{subfigure}
  \caption{Examples of Section~\ref{sec:iter} and illustration of the notation}
\end{figure}

\begin{toappendix}
In the proofs of this section, we will use the \emph{fine dissociation}:

\begin{definition}
  \label{def:finedissoc}
  Let $M = (I, e, F_\LL,
  F_\RR)$ be a critical model, and let $e = (u, v)$.
  The \emph{fine dissociation} of~$M$ is the instance obtained by modifying~$I$
  in the following way:
  \begin{itemize}
    \item Create fresh elements~$u'$ and~$v'$.
    \item For each left-incident fact $F_\LL'$ of~$e$ in~$I$ except~$F_\LL$, create
      the fact obtained from~$F_\LL'$ by replacing~$u$ by~$u'$.
    \item For each right-incident fact $F_\RR'$ of~$e$ in~$I$ except~$F_\RR$, create
      the fact obtained from~$F_\RR'$ by replacing~$v$ by~$v'$.
    \item Copy the edge~$e$ on~$(u, v')$ and~$(u',v)$, and remove the covering
      facts of~$e$.
  \end{itemize}
\end{definition}

  We gave an illustration of the fine dissociation in Figure~\ref{fig:other} in
  the main text, but specialized to the case where $\Xi = 0$ (i.e., all incident
  facts are copy facts) because this is the setting of most of this section; but
  the fine dissociation can also be defined without this assumption, and we will
  use it in this general sense here.

  Notice that the definition of the fine dissociation is closely related to the notion of the same name
  in~\cite{amarilli2021dichotomy} but with a slight difference,
  e.g., we do not
  create incomplete copies of~$e$ on the edges $(u, v)$ and $(u', v')$. This
  difference is inessential, because these edges would be garbage facts of the
  edges $(u, v')$ and $(u', v)$.

Equivalently, the fine dissociation is the iteration of~$M$ but without the copy
  of~$e$ on~$(u',v')$; or it is like the dissociation but the two copies of~$e$
  are created not with leaf elements but with elements respectively involved in all the
  left-incident facts of~$e$ except~$F_\LL$ and all the right-incident facts
  of~$e$ except~$F_\RR$.

We make the following claim, which relies on the notion of lexicographic weight (but,
  for now, only in the 
componentwise sense). The claim holds no matter whether the critical model
  is iterable or not. The proof method is the same as that of Lemma~7.6
  of~\cite{amarilli2021dichotomy} (see in particular Figure~8
  of~\cite{amarilli2021dichotomy}), but skipping the first two steps thanks to
  the omission of the incomplete edges $(u,v)$ and $(u',v')$, and reasoning
  about extra weight and (componentwise) lexicographic weight instead of side
  weight.
\begin{claim}
  \label{clm:finedissoc}
  The fine dissociation of a critical model does not satisfy the query.
\end{claim}

\begin{proof}
  Letting $M = (I, e, F_\LL, F_\RR)$ and writing $e = (u, v)$, 
  consider the copy of~$e$ on~$(u, v')$ in the fine dissociation~$I'$ of~$M$.
  If $(u,v')$ is a leaf edge in~$I'$, i.e., $F_\RR$ was the only right-incident
  fact of~$e$ in~$I$, then we can vacuously dissociate $(u,v')$, using the
  convention introduced in Appendix~\ref{apx:basic}. Otherwise, let us show that
  we can dissociate the non-leaf edge~$(u,v')$ without breaking the query.
  The incident facts of~$(u,v')$ are the same as
  those of~$e$ in~$I$, except it is missing the right-incident fact $F_\RR$
  on~$v'$. As $M$ is a critical model, $e$ is clean, so $F_\RR$ is either an extra
  fact or a copy fact of~$e$ in~$I$ (not a garbage fact). We know by similar reasoning to
  Lemma~\ref{lem:fctremove} that either the extra weight of~$(u,v')$ in~$I'$ is
  less than that of~$e$ in~$I$, i.e., is $< \Xi$ (if $F_\RR$ is an extra fact
  of~$e$ in~$I$); or that the extra weight of~$(u,v')$ in~$I'$ is~$\Xi$ but the
  lexicographic weight of~$(u,v')$ in~$I'$ is less than that of~$e$ in~$I$,
  i.e., it is $< \Lambda$. Thus, by applying Claim~\ref{clm:dissocextra} in the
  first case and Claim~\ref{clm:dissoclex} in the second case, we know that we
  can dissociate $(u,v')$ in~$I$ without breaking the query. Let $I'_1$ be the
  result of this dissociation.

  Considering $I'_1$ and the edge $(u',v)$, but noticing that the incident facts
  of~$(u',v)$ in~$I'_1$ are the same as in~$I'$, the 
  symmetric argument (but
  noticing the absence of~$F_\LL$ instead of~$F_\RR$), shows that we can dissociate
  $(u',v)$ in~$I'_1$ without breaking the query. We do so, obtaining~$I_2'$
  which satisfies the query.

  Now, we can homomorphically
  merge $u$ and $u'$, merge $v$ and $v'$, and merge the leaf copies of~$e$
  involving~$u$ and~$u'$ (respectively in steps 1 and 2) and $v$ and~$v'$
  (respectively in steps 2 and 1). This maps~$I_2'$ to the dissociation of~$e$
  in~$I$, which does not satisfy the query because $e$ is tight. We have reached
  a contradiction, so the proof is concluded.
\end{proof}

Thanks to this observation, in the case where the iteration cannot break the query, we will be able to
simplify critical models by noticing that the critical extra weight is~$0$.
Namely:
\end{toappendix}

A first observation is that, in this case, we have $\Xi=0$, by contraposition of
the following:

\begin{claimrep}
  \label{clm:pure}
  If the critical extra weight is $>0$, then $Q$ has a non-iterable critical model.
\end{claimrep}

\begin{proofsketch}
  Take a critical model $M =
  (I, e, F_\LL, F_\RR)$ with $e = (u,v)$ and one of~$F_\LL, F_\RR$ an extra fact.
  The edge $(u',v')$ in
  the iteration of~$M$ has weight~$\Theta$ 
  and extra weight $<\Xi$, so we can dissociate it without
  breaking~$Q$ and merge the two resulting copies. This yields the so-called \emph{fine dissociation} (see
  Figure~\ref{fig:other}, and Definition~\ref{def:finedissoc} in the appendix), which
  violates~$Q$.
\end{proofsketch}

\begin{proof}
  Let $(I, e, F_\LL', F_\RR')$ be a critical model of the query, and let $e = (u,
  v)$. As the critical extra
  weight is $>0$, we know that one of $u$ and $v$ has an extra fact.
  Let us replace the choice of incident facts $F_\LL'$ and $F_\RR'$ to use extra
  facts if possible, i.e., pick incident facts $F_\LL$ and $F_\RR$ where 
  at least one of them
  is an extra fact of~$e$. This yields another critical model $M = (I, e, F_\LL, F_\RR)$, with the same
  weight and extra weight and lexicographic weight. Let us show that $M$ is
  non-iterable.

  Consider the iteration $I'$ of~$M$. Consider first the copy of~$e$
  of~$C$ on $(u', v')$. This edge has weight~$\Theta$, and has extra weight
  $< \Xi$. Indeed, it is missing the copies of~$F_\LL$ and~$F_\RR$ and its other
  incident facts are copies of the other incident facts of~$e$ in~$I$ (i.e.,
  copy facts and extra facts) and the covering facts of the copies of~$e$
  on~$(u',v)$ and~$(u,v')$ but these are copy facts of~$(u',v')$ in~$I'$ (note
  in particular that neither $u$ nor~$v$ forms a triangle with $(u',v')$).
  Thus, indeed the edge $(u',v')$ in~$I'$ can be dissociated without breaking
  the query (by Claim~\ref{clm:dissocextra}).
  Next, we can homomorphically map the leaf element of the
  two leaf edges thus created into~$u$ and~$v$, merging the two leaf edges
  into~$(u,v)$ and~$(u',v')$, i.e., removing them
  without breaking the query.
  We have obtained as a result of this process the fine dissociation of~$M$, and
  shown that it satisfies the query, which contradicts the result of
  Claim~\ref{clm:finedissoc}, concluding the proof.
\end{proof}

Hence, in the rest of the section, we assume $\Xi=0$, and fix an iterable critical model $M = (I, e, F_\LL, F_\RR)$. 
All incident facts of~$e = (u,v)$ in~$I$ are copy
facts,
so we let $t, t_1, \ldots, t_{\tau-1}$ be
the left copy elements and $w, w_1, \ldots, w_{\omega-1}$ be the right copy
elements, where $t$ and $w$ are the elements that occur in~$F_\LL$
and~$F_\RR$ respectively (the choice of $F_\LL$ and $F_\RR$ from now on only
matters in that it distinguishes two copy elements~$t$ and~$w$).
The lexicographic weight of~$e$ in~$I$ is thus $\Lambda =
(\tau, \omega)$ with $\tau, \omega \geq 1$.
We let $C$ be the covering facts of~$e$ in~$I$. See
Figure~\ref{fig:iter}.

\subparagraph*{$\bm{n}$-step iteration.}
Let us now define the \emph{$n$-step iteration} of~$M$.
It is related to iteration
in~\cite{amarilli2021dichotomy}, but specialized to the case where $\Xi = 0$, i.e., all
incident facts are copy facts.

\begin{definition}
  For $n > 0$, the \emph{$n$-step iteration} of~$M$ is obtained by modifying~$I$:
  \begin{itemize}
    \item Create elements $u_1, \ldots, u_n$ and $v_1, \ldots, v_n$, where we
      identify~$u$ and~$u_1$ and $v_n$ and~$v$.
    \item For all $1 \leq i, j \leq n$,
      copy~$e$ on $(u_i, t_{j'})$ and $(w_{i'},v_j)$
      for all $1 \leq j' < \tau$ and $1 \leq i' < \omega$.
    \item For all $1 \leq i \leq n$, copy~$e$ on $(u_i, v_i)$ for all $1 \leq i \leq n$ and on
      $(u_{i+1}, v_i)$ for all~$1 \leq i < n$.
    \item Remove the facts of~$C$, except in the trivial case where $n = 1$.
  \end{itemize}
\end{definition}

The iteration is illustrated in Figure~\ref{fig:iter}. Note that the 1-step
iteration is exactly~$I$. Further, the 2-step iteration
resembles the iteration in Section~\ref{sec:noniter}, but omits some
incomplete copies of~$(u,t)$ and~$(w,v)$ (i.e., the dashed edges in
Figure~\ref{fig:iternonit}): as $t$ and $w$ are copy elements 
these facts would
be garbage facts so the difference is inessential.

We now show that, if the iteration process of Section~\ref{sec:noniter} cannot
break~$Q$ on any critical model, then 
$Q$ must also be satisfied in the $n$-step iteration of any critical
model~$M$ for any $n > 0$. This proposition summarizes how we use the hypothesis that all
critical models are iterable:

\begin{propositionrep}
  \label{prp:iterability}
  Let $Q$ be a query that has a critical model. Assume that all critical models
  for~$Q$ are
  iterable. Then $\Xi = 0$ and, for any critical model $M$ of~$Q$, for any~$n >
  0$, the $n$-iteration of~$M$ satisfies~$Q$; further it is a subinstance-minimal model of~$Q$.
\end{propositionrep}

\begin{proofsketch}
  Intuitively, the $n$-step iteration can be achieved by repeatedly performing the
  iteration from Section~\ref{sec:noniter}. A tedious point in the proof is to show that
  subinstance-minimality is preserved throughout this process.
\end{proofsketch}

\begin{toappendix}
\begin{figure}
  \hfill
    \begin{tikzpicture}[xscale=1.3,yscale=.75,inner sep=1.5]
    \node (u1) at (0, 1) {$u_1$};
    \node (v1) at (1, 1) {$v_1$};
    \node (u2) at (0, -1) {$u_2$};
    \node (v2) at (1, -1) {$v_2$};
    \node (t) at (-1, 1) {$t$};
    \node (t1) at (-1, 0) {$t_1$};
    \node (t2) at (-1, -1) {$t_2$};
    \node (w1) at (2, 1) {$w_1$};
    \node (w2) at (2, 0) {$w_2$};
    \node (w) at (2, -1) {$w$};
    \draw[orange,->,thick] (u1) -- (v1);
    \draw[orange,->,thick] (u2) -- (v1);
    \draw[orange,->,thick] (u2) -- (v2);
    \draw[orange,->,thick] (u1) -- (t);
    \draw[orange,->,thick] (u1) -- (t1);
    \draw[orange,->,thick] (u1) -- (t2);
    \draw[orange,->,thick] (u2) -- (t1);
    \draw[orange,->,thick] (u2) -- (t2);
    \draw[orange,->,thick] (w1) -- (v1);
    \draw[orange,->,thick] (w2) -- (v1);
    \draw[orange,->,thick] (w1) -- (v2);
    \draw[orange,->,thick] (w2) -- (v2);
    \draw[orange,->,thick] (w) -- (v2);
    \end{tikzpicture}
    \hfill
    \hfill
    \begin{tikzpicture}[xscale=1.3,yscale=.75,inner sep=1.5]
    \node (u1) at (0, 1) {$u_1$};
    \node (v1) at (1, 1) {$v_1$};
    \node (u2) at (0, -1) {$u_2$};
    \node (v2) at (1, -1) {$v_2$};
    \node (t) at (-1, 1) {$t$};
    \node (t1) at (-1, 0) {$t_1$};
    \node (t2) at (-1, -1) {$t_2$};
    \node (w1) at (2, 1) {$w_1$};
    \node (w2) at (2, 0) {$w_2$};
    \node (w) at (2, -1) {$w$};
    \draw[orange,->,thick] (u1) -- (v1);
    \draw[orange,->,thick] (u2) -- (v1);
    \draw[orange,->,thick] (u2) -- (v2);
    \draw[orange,->,thick] (u1) -- (t);
    \draw[orange,->,thick,dashed] (u2) -- (t);
    \draw[orange,->,thick] (u1) -- (t1);
    \draw[orange,->,thick] (u1) -- (t2);
    \draw[orange,->,thick] (u2) -- (t1);
    \draw[orange,->,thick] (u2) -- (t2);
    \draw[orange,->,thick] (w1) -- (v1);
    \draw[orange,->,thick] (w2) -- (v1);
    \draw[orange,->,thick] (w1) -- (v2);
    \draw[orange,->,thick] (w2) -- (v2);
    \draw[orange,->,thick] (w) -- (v2);
    \draw[orange,->,thick,dashed] (w) -- (v1);
    \end{tikzpicture}
    \hfill\null

    \caption{Illustration of the proof of Proposition~\ref{prp:iterability}. The
    left picture represents the 2-step iteration of the model~$M$ of
    Figure~\ref{fig:iter} (top), and the right picture represents its iteration
    in the sense of Section~\ref{sec:noniter}. The only difference are the two
    dashed edges which are copies of the edge $(u,t)$ and $(w,v)$ respectively, missing the facts
    $F_\LL$ and~$F_\RR$ respectively. They have weight $< \Theta$, so can be
    dissociated, and merged in $(u,v')$ and~$(u',v)$, so indeed the difference
    between the two processes is inessential. (However, the 2-step iteration,
    unlike the iteration, is a subinstance-minimal model
    of~$Q$, as we show.)}
    \label{fig:compareiter}
\end{figure}
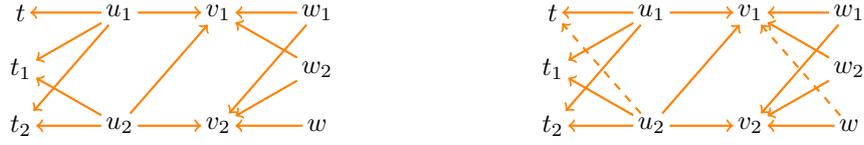
\begin{figure}
    \hfill
    \begin{tikzpicture}[xscale=1.3,yscale=.75,inner sep=1.5]
    \node (u1) at (0, 1) {$u_1$};
    \node (v1) at (1, 1) {$v_1$};
    \node (u2) at (0, 0) {$u_2$};
    \node (v2) at (1, 0) {$v_2$};
    \node (u3) at (0, -2) {$u_3$};
    \node (v3) at (1, -2) {$v_3$};
    \node (t) at (-1, 1) {$t$};
    \node (t1) at (-1, -.5) {$t_1$};
    \node (t2) at (-1, -2) {$t_2$};
    \node (w1) at (2, 1) {$w_1$};
    \node (w2) at (2, -.5) {$w_2$};
    \node (w) at (2, -2) {$w$};
    \draw[orange,->,thick] (u1) -- (v1);
    \draw[orange,->,thick] (u2) -- (v1);
    \draw[orange,->,thick] (u2) -- (v2);
    \draw[orange,->,thick] (u3) -- (v2);
    \draw[orange,->,thick] (u3) -- (v3);
    \draw[orange,->,thick] (u1) -- (t);
    \draw[orange,->,thick] (u1) -- (t1);
    \draw[orange,->,thick] (u1) -- (t2);
    \draw[orange,->,thick] (u2) -- (t1);
    \draw[orange,->,thick] (u2) -- (t2);
    \draw[orange,->,thick] (u3) -- (t1);
    \draw[orange,->,thick] (u3) -- (t2);
    \draw[orange,->,thick] (w1) -- (v1);
    \draw[orange,->,thick] (w2) -- (v1);
    \draw[orange,->,thick] (w1) -- (v2);
    \draw[orange,->,thick] (w2) -- (v2);
    \draw[orange,->,thick] (w1) -- (v3);
    \draw[orange,->,thick] (w2) -- (v3);
    \draw[orange,->,thick] (w) -- (v3);
    \end{tikzpicture}
    \hfill
    \hfill
    \begin{tikzpicture}[xscale=1.3,yscale=.75,inner sep=1.5]
    \node (u1) at (0, 1) {$u_1$};
    \node (v1) at (1, 1) {$v_1$};
    \node (u2) at (0, 0) {$u_2$};
    \node (v2) at (1, 0) {$v_2$};
    \node (v2p) at (1, -.7) {$v_2'$};
    \node (u3) at (0, -2) {$u_3$};
    \node (u3p) at (0, -1.3) {$u_3'$};
    \node (v3) at (1, -2) {$v_3$};
    \node (t) at (-1, 1) {$t$};
    \node (t1) at (-1, -.5) {$t_1$};
    \node (t2) at (-1, -2) {$t_2$};
    \node (w1) at (2, 1) {$w_1$};
    \node (w2) at (2, -.5) {$w_2$};
    \node (w) at (2, -2) {$w$};
    \draw[orange,->,thick] (u1) -- (v1);
    \draw[orange,->,thick] (u2) -- (v1);
    \draw[orange,->,thick] (u2) -- (v2);
    \draw[orange,->,thick] (u3p) -- (v2);
    \draw[orange,->,thick] (u3p) -- (v2p);
    \draw[orange,->,thick] (u3) -- (v2p);
    \draw[orange,->,thick] (u3) -- (v3);
    \draw[orange,->,thick] (u1) -- (t);
    \draw[orange,->,thick] (u1) -- (t1);
    \draw[orange,->,thick] (u1) -- (t2);
    \draw[orange,->,thick] (u2) -- (t1);
    \draw[orange,->,thick] (u2) -- (t2);
    \draw[orange,->,thick] (u3) -- (t1);
    \draw[orange,->,thick] (u3) -- (t2);
    \draw[orange,->,thick] (u3p) -- (t1);
    \draw[orange,->,thick] (u3p) -- (t2);
    \draw[orange,->,thick] (w1) -- (v1);
    \draw[orange,->,thick] (w2) -- (v1);
    \draw[orange,->,thick] (w1) -- (v2);
    \draw[orange,->,thick] (w2) -- (v2);
    \draw[orange,->,thick] (w1) -- (v2p);
    \draw[orange,->,thick] (w2) -- (v2p);
    \draw[orange,->,thick] (w1) -- (v3);
    \draw[orange,->,thick] (w2) -- (v3);
    \draw[orange,->,thick] (w) -- (v3);
    \end{tikzpicture}
    \hfill\null
    \caption{Illustration of the proof of Proposition~\ref{prp:iterability}. The
    left picture
    illustrates the 3-step iteration of the model~$M$ of Figure~\ref{fig:iter}
    (top). Here we have $\Lambda = (3,3)$. Observe how removing a fact in any of
    the edges $(u_i,v_i)$ or $(u_{i+1}, v_i)$ allows us to dissociate it
    (because its weight is now~$<\Theta$) and
    merge to the fine dissociation (Figure~\ref{fig:other}, top). Further,
    removing a fact in any other orange edge (involving~$u_i$ or~$v_i$) allows us to dissociate it (because
    its weight is now~$<\Theta$), 
    merge the dangling copy on~$u_i$ or~$v_i$ on the edge of the form
    $(u_i,v_i)$, and then the edge $(u_i,v_i)$ now has lexicographic weight
    $(3,2)$ or~$(2,3)$ and can also be dissociated. The right picture
    illustrates the result of performing the 2-step iteration
    of $(u_3,v_2)$ in the left picture with
    some choice of left-incident fact in the edge $(u_2,v_2)$ and a choice of
    right-incident fact in the edge~$(u_3,v_3)$: one can see that the result is
    isomorphic to the 4-step iteration of~$M$ (Figure~\ref{fig:iter}, bottom).}
    \label{fig:process}
\end{figure}
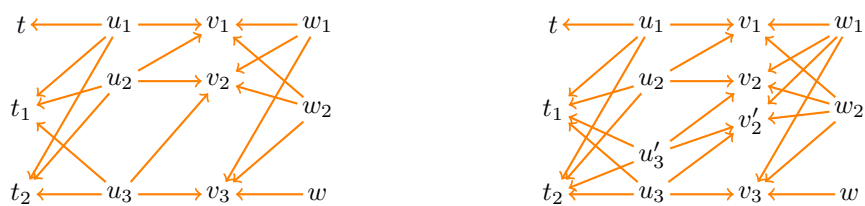
\end{toappendix}

\begin{proof}
  The fact that $\Xi = 0$ is directly given by the contrapositive of Claim~\ref{clm:pure}.

  Let us show
  that the $n$-iteration of any critical model~$M$
  is a subinstance-minimal model of~$Q$,
  by induction on~$n$. 

  The base case of $n=1$ is vacuous because the $1$-iteration of any
  critical model~$M$ is $M$ itself, which is by definition a subinstance-minimal
  model of~$Q$.

  For the case $n=2$, we first show that the $2$-step iteration $I_2$ satisfies the
  query, and then that it is subinstance-minimal. 
  We assume that the iteration of~$M$ satisfied the query,
  now the $2$-step iteration $I_2$ of~$M$ is identical to the iteration of~$M$
  up to a minor difference, illustrated in Figure~\ref{fig:compareiter}.
  Specifically, letting $t$ and $w$ be the other elements of~$F_\LL$ and $F_\RR$
  respectively (which are copy facts as $\Xi = 0$), if $\Theta > 1$ then $I_2$
  has no edge $(u_2,t)$ and $(w,v_1)$ whereas the iteration has
  edges $(u',t)$ and $(w,v')$ containing a copy of the covering facts of~$e$
  except $F_\LL$ and $F_\RR$ respectively. But these facts are garbage facts; given
  the iteration, we can dissociate these two edges (they have weight $<\Theta$,
  so we can use Claim~\ref{clm:dissocweight} if they are non-leaf edges) without
  breaking the query, homomorphically merge
  the incident leaf edge on~$u'$ and~$v'$ with $(u',v)$ and $(u,v')$, merge the
  other incident leaf edges to~$(u,t)$ and~$(w,v)$, and this establishes that
  $I_2$ satisfies the query. 

    \bigskip

    Now, let us show subinstance-minimality. Consider a strict subset $J$
    of~$I_2$.
    By monotonicity, it suffices to consider the case of a single missing fact
    $F$, i.e., $J = I_2 \setminus \{F\}$. The argument is sketched on an
    example for the 3-step iteration on Figure~\ref{fig:process} (left).

    If $F$ is one of the facts of~$I_2$ that does not involve $u_1, u_2, v_1,
    v_2$, then considering the homomorphism~$h$ mapping~$J$ to~$I$
    by mapping $u_1$ and $u_2$ to~$u$ and~$v_1$ and~$v_2$ to~$v$  and being the
    identity otherwise, then the only preimage of~$F$ by~$h$ would be~$F$, so
    $h$ is in fact a homomorphism from~$J$ to $I \setminus \{F\}$. As $I$ is a
    subinstance-minimal model of~$Q$, we know that $I \setminus \{F\}$ does not
    satisfy~$Q$, and neither does~$J$.

    If $F$ is a fact of the edge $(u_1, t_{j'})$ for some $1 \leq j' < \tau$
    (first case), or
    a fact of the edge $(u_1, t)$, i.e., $F_\LL$ or another fact on the same
    elements (second case), then it~$J$ that edge has weight $< \Theta$ and can be
    dissociated without breaking the query. 
    Consider
    now the edge $(u_1, v_1)$: its weight is~$\Theta$, its extra weight is $\Xi
    = 0$, its left copy elements are $t$ and the $t_{j''}$ for $j'' \in \{1,
    \ldots \tau-1\} \setminus \{i'\}$ (in the first case) or the $t_{j''}$ for
    $1 \leq j'' < \tau$ (second case), i.e., $\tau-1$ left copy elements, so its
    lexicographic weight is less than $\Lambda = (\tau, \omega)$ and we can
    dissociate it without breaking the query by Claim~\ref{clm:dissoclex}. We
    can now map the result homomorphically to the fine dissociation of~$M$ (Definition~\ref{def:finedissoc}) 
    by mapping $u_1$ to~$u$, mapping $u_2$ to~$u'$, mapping $v_1$ and $v_2$
    to~$v$, mapping the dangling copies of~$e$ on $u_1$ to~$(u,v')$ and on~$v_1$
    to~$(u',v)$, and the dangling copy of~$e$ on~$t_{j''}$ to $(u,t_{j''})$ (in the
    first case)  or the dangling copy of~$e$ on~$t$ to~$(u,t)$ (in the second
    case).
We know by Claim~\ref{clm:finedissoc} that the fine dissociation does not
satisfy the query, which concludes.

    If $F$ is a fact of the edge $(w_{i'}, v_2)$ for some $1 \leq i' < \omega$ (first
    case), or a fact of the edge $(w, v_2)$, i.e., $F_\RR$ or a fact on the same
    elements (second case), then we reason in the same way: we dissociate this
    edge, and then the edge $(u_2, v_2)$ has left copy elements $v_1$ and the
    $t_{j''}$ for $1 \leq j'' < \tau$, i.e., $\tau$ left copy elements, and it
    has $\omega-1$ right copy elements, so it has strictly smaller lexicographic
    weight, and we can dissociate it and conclude in the symmetric way as above.

    If $F$ is a fact of the edge $(u_2, t_{j'})$ for some $1 \leq j' < \tau$,
    then in~$J$ that edge has weight~$< \Theta$ and can be dissociated. We now
    consider the edge $(u_2, v_1)$: its weight is~$\Theta$, its extra weight is
    $\Xi = 0$, its left copy elements are $v_2$ and the $t_{j''}$ for $j'' \in
    \{1, \ldots \tau-1\} \setminus \{i'\}$ (note that the dangling edge
    created on the previous dissociation consists of garbage facts), 
    and its right copy elements are $u_1$ and the $w_{i'}$
    for $1 \leq i' < \omega$. So there are $\tau-1$ left-incident elements, so the lexicographic
    weight of the edge is less than $\Lambda$. So we can dissociate the edge
    without breaking the query by Claim~\ref{clm:dissoclex}. We can now map the
    result to the fine dissociation as
    the identity except that the dangling edges on~$u_2$ and~$v_1$ and $t_{j'}$ are mapped
    to~$(u_1, u_2)$ and $(v_1, v_2)$ and $(u_1, t_{j'})$ respectively. 
We conclude again by 
    Claim~\ref{clm:finedissoc}.

    If $F$ is a fact of the edge $(w_{i'}, v_1)$ for some $1 \leq i' < \omega$,
    the reasoning is symmetric: we dissociate the edge $(w_{i'}, v_1)$, we
    notice that the edge $(u_2, v_1)$ now has $\tau$ left-incident elements and
    $\omega-1$ right-incident elements, so we dissociate again by
    Claim~\ref{clm:dissoclex}, and map to the fine dissociation and conclude by
    Claim~\ref{clm:finedissoc}.

    Last, if $F$ is a fact of a copy of~$e$, we distinguish the three cases. If
    $F$ is a fact of $(u_1, v_1)$, then we dissociate this edge in~$J$ as it has
    weight $< \Theta$, and we can map to the fine dissociation of~$M$ 
    mapping $u_1$ to~$u$ and $u_2$ to~$u'$ and $v_1$ and $v_2$ to~$v$ and
    mapping the dangling edges on~$u_1$ and~$v_1$ to $(u, v')$ and $(u', v')$.
    If $F$ is a fact of $(u_2, v_2)$, the reasoning is similar. If $F$ is a fact
    of $(u_2, v_1)$, we dissociate it, and we map to the fine dissociation in
    the same way as in the previous paragraph.

    Thus, no matter the missing fact~$F$, we know that $I_2 \setminus \{F\}$
    does not satisfy the query, which shows that it is subinstance minimal.

    \bigskip

  We now take care of the induction step of the induction. Now, let $n>2$,
  let $M$ be a critical model
  assume that the $(n-1)$ iteration $I_{n-1}$ of~$M$ 
  is a subinstance-minimal model of the
  query, and let us show that the same is true for
  the $n$-iteration $I_n$. We first show that it satisfies the query, and then that it
  is subinstance-minimal.

  Consider the copy of~$e$ on~$(u_n, v_{n-1})$. We claim that this edge is tight
  in~$I_{n-1}$. Indeed, it is a non-leaf edge, as witnessed by the covering facts of
  the edge $(u_n, v_n)$ and $(u_{n-1}, v_{n-1})$. Further, it is tight because
  if we dissociate it then the result has a homomorphism to the fine
  dissociation which by Claim~\ref{clm:finedissoc} does not satisfy the
  query. Specifically, the homomorphism maps $u_1, \ldots, u_{n-1}$ to~$u$, maps
  $u_n$ to~$u'$, maps $v_1, \ldots, v_{n-1}$ to~$v'$, maps $v_n$ to~$v$, and maps
  the dangling edge on~$u_{n-1}$ and on $v_{n-1}$ to $(u,v')$ and $(u',v)$
  respectively.

  Let us now consider this tight edge $e' = (u_n, v_{n-1})$ and pick arbitrary
  facts $F_\LL'$ and $F_\RR'$ respectively in the copies $(u_{n-1}, v_{n-1})$ and
  $(u_n, v_n)$ of~$e$.
  We claim that $M' = (I_{n-1}, e, F_\LL, F_\RR)$ is a critical instance.
  For this, we first notice that $e$ has weight~$\Theta$. As $\Xi = 0$, there is
  nothing to check. As for the lexicographic weight, we know that $(u_2, v_1)$
  is missing the copy of the edge with~$t$ and with~$w$, but has in exchange
  the edge~$(u_1, v_1)$ incident to~$v_1$ and the edge~$(u_2, v_2)$ incident
  to~$u_2$, so it is indeed unchanged.

  Now, consider the 2-iteration of $M'$: as $M'$ is a critical instance, it is
  iterable, and by applying the base case of $n=2$ we know that its 2-iteration 
  satisfies the query. One can see that the result (see
  Figure~\ref{fig:process}, right) that the
  result is
  isomorphic to~$I_n$ (shown in Figure~\ref{fig:iter}, bottom).
  Thus, the $n$-iteration $I_n$ satisfies the query. We name the
  elements accordingly in what follows.

  \bigskip

  We last show that $I_n$ is subinstance-minimal. The argument is similar to the
  subinstance-minimality of~$I_2$: see again Figure~\ref{fig:process} (left) and
  the explanation in the caption. If a fact $F$ is missing that does not involve
  any of the $u_i$ and $v_i$ for $1 \leq i \leq n$, we can map back to
  $I\setminus \{F\}$ by mapping all $u_i$ to~$u$ and all~$v_i$ to~$v$ and we
  conclude by subinstance-minimality of~$I$. If a fact is missing in one of the
  copies $(u_i,v_i)$ of~$e$, then we can dissociate it because it has
  weight~$<\Theta$ and map homomorphically
  to the fine dissociation 
  by mapping $u_1, \ldots, u_i$ to~$u$, mapping $u_{i+1}, \ldots, u_{n}$
  to~$u'$, mapping $v_1, \ldots, v_{i-1}$ to~$v'$, mapping $v_i, \ldots,
  v_{n}$ to~$v$, and mapping
  the dangling leaf edge on~$u_i$ to~$(u,v')$ and the dangling leaf edge
  on~$v_i$ to~$(u',v)$.
  If a fact is missing on one of the copies $(u_{i+1},v_i)$ of~$e$, then we map $u_1, \ldots, u_i$ to~$u$, map $u_{i+1}, \ldots, u_{n}$ to
  $u'$, map $v_1, \ldots, v_i$ to~$v'$, map $v_{i+1}, \ldots, v_{n}$ to~$v$,
  and map the dangling leaf edge on~$u_{i+1}$ to~$(u',v')$ and the dangling leaf
  edge on~$v_i$ to $(u,v')$. Last, if one of the other facts is missing, it is a
  fact involving some $u_i$ or involving some $v_j$, i.e., an incident fact to
  some edge $(u_i, v_i)$, and we can dissociate the edge of the fact because it
  has weight~$<\Theta$ then argue
  as before that the edge $(u_i,v_i)$ now has lexicographic weight~$<\Lambda$
  and can be dissociated, and then we map to the fine dissociation as explained
  in what precedes. This concludes the proof.
\end{proof}

\subparagraph*{Coding.}
We explain how to code an undirected graph to reduce from 
$\phi,\eta$-U-ST-CON for some $0 < \phi \leq
1$ and $0 < \eta < 1$ (see Definition~\ref{def:ustcon}): this time no saturation
is needed. 
Proposition~\ref{prp:iterability} will then intuitively show that some paths in the
coding make~$Q$ true.

\begin{definition}
  Let $G = (V, E)$ be an undirected graph with source~$r$ and sink~$s$, 
  with~$r \neq s$.
  The \emph{coding} $I_G$ of~$G$ in~$M$ is the instance defined by modifying $I$
  in the following way:
  \begin{itemize}
    \item For all $a \in V$, create a fresh element $u_a$, and copy
      $(u,t_{j'})$ on $(u_a,t_{j'})$ for all $1 \leq j' < \tau$.
    \item We identify $u$ to~$u_r$, so $u_r$ %
      also occurs in another copy of~$e$, namely the edge $(u_r,t)$.
    \item For each edge $\pi = \{a, b\} \in E$, create fresh elements $u_\pi,
      v_{\pi,a}, v_{\pi,b}$, copy $(u,t_{j'})$ on~$(u_\pi,t_{j'})$ for all $1
      \leq j' < \tau$, copy $(w_{i'},v)$ on~$(w_{i'},v_{\pi,\beta})$ for all
      $1 \leq i' < \omega$ and $\beta \in \{a,b\}$, and copy~$(u,v)$ on $(u_a,
      v_{\pi,a})$, $(u_\pi, v_{\pi,a})$, $(u_\pi,
      v_{\pi,b})$, and $(u_b, v_{\pi,b})$.
    \item Copy $(u,v)$ on~$(u_s,v)$, and then remove the facts of~$C$.
  \end{itemize}
\end{definition}

An example is given in Figure~\ref{fig:itercode}, shortening the vertex names
for readability.
The coding $I_G$ can clearly be built in polynomial time in~$G$.
We partition the facts of~$I_G$ in four kinds:

\begin{itemize}
  \item The \emph{base facts} (not pictured), i.e., the facts involving no element of~$\{u_a \mid a
    \in V\} \cup \{v_{\pi,\beta} \mid \pi \in E, \beta \in \pi\} \cup \{v\}$.
  \item The \emph{supplementary base facts} (in black), i.e., the covering facts of
    $(u_r,t)$ and $(u_r,t_{j'})$ for $1 \leq j' < \tau$, and the covering facts
    of $(u_s,v)$ and $(w,v)$ and
    $(w_{i'},v)$ for $1 \leq i' < \omega$.
  \item The \emph{vertex facts} (in purple) of each vertex~$a \in V \setminus \{r\}$, i.e., the covering
    facts of $(u_a,t_{j'})$ for $1 \leq j' < \tau$.
  \item The \emph{edge facts} (in orange) of each edge $\pi = \{a, b\}$ of~$E$, i.e., all covering
    facts and incident facts of $(u_\pi,v_{\pi,a})$ and
    $(u_\pi,v_{\pi,b})$, including the covering facts of
    $(u_a,v_{\pi,a})$ and $(u_b,v_{\pi,b})$.
\end{itemize}

Similarly to Section~\ref{sec:noniter}, the base facts of~$I_G$ are precisely
the facts of~$I$ that do not involve~$u$ or~$v$. A subinstance $J \subseteq I_G$
is \emph{well-formed} if it contains all base facts and
supplementary base facts, and \emph{ill-formed} otherwise. We can then use
subinstance-minimality to show:

\begin{claimrep}
  \label{clm:ill3}
  The ill-formed subinstances do not satisfy the query.
\end{claimrep}

  \begin{proof}
  By monotonicity, it suffices to consider the case of a single missing fact $F$.

  We first focus on the case of a missing base fact~$F$.
  In this case, we can map~$J$ homomorphically to~$I \setminus
  \{F\}$, and conclude by subinstance-minimality.

  Second, let us study the case where a supplementary base fact~$F$
  involving~$u = u_r$ is missing.
  We define a homomorphism~$h$ to a strict subset of the 2-step iteration by
  mapping all $u_x$ with $x \neq r$ to~$u'$, mapping $u_r$ to~$u$, mapping all
  $v_{\pi,r}$ adjacent to~$u_r$ to~$v'$, and mapping all other~$v_{\pi,\beta}$
  to~$v$. To check that $h$ is indeed a homomorphism, the important points are
  that there are no edges $(u_x,t)$ with $x \neq s$, there are no edges
  $(w,v_{\pi,r})$, and there are no edges $(u_r, v_{\pi,\beta})$ with~$\beta
  \neq r$.
  Further, the
  left-incident fact $F'$ of~$e$ in~$I$ corresponding to the missing fact~$F$ has no
  image because the only
  vertex mapped to~$u$ by~$h$ is~$u_r$ which is missing~$F$. Thus, $h$ is a 
  homomorphism to the strict subset of the 2-step iteration where we
    removed~$F'$, and we conclude by Proposition~\ref{prp:iterability}.

  Last, we consider the case of a missing supplementary base fact~$F$ that involves~$v$. We again define a
  homomorphism $h$ to a strict subset of the 2-step iteration. We map all $u_x$ with $x
  \neq t$ to~$u$, map $u_s$ to~$u'$, map $v$ to~$v$ and map all other copies
  of~$v$ to~$v'$. Again, to see that~$h$ is a homomorphism, the important points
  is that $(u_s,t)$ is not an edge of~$I_G$, that there are no edges of~$I_G$ of
  the form $(w,v_{\pi,\beta})$, and that there are no edges of the form
  $(u_x,v)$ with $v \neq t$. Further, again the fact corresponding to~$F$ in the
  2-step iteration has no image, so we have a homomorphism to a strict subset of
  the 2-step iteration and conclude by
  Proposition~\ref{prp:iterability}.
\end{proof}

Now, consider a well-formed subinstance $J \subseteq I_G$.
A vertex $a \in V$ is \emph{complete} in~$J$ if all vertex facts of~$a$
are present, and \emph{incomplete} otherwise;
and an edge $\pi \in E$ is \emph{complete} in~$J$ if all its edge facts of~$\pi$
are
present, and \emph{incomplete} otherwise. A \emph{complete path} in~$J$ is a path
connecting~$r$ and~$s$ in~$G$ such that all traversed edges and vertices are
complete in~$J$ (except~$r$, for which completeness was not defined). We say that $J$
is \emph{good} if it has a complete path, and \emph{bad} otherwise.
We can easily see that  good subinstances satisfy the query, because they contain an
iterate of~$M$ and we can use Proposition~\ref{prp:iterability}:

\begin{claimrep}
  \label{clm:good3}
  For any good well-formed subinstance $J \subseteq I_G$, there is a
  homomorphism from the $(2n+1)$-step iteration of~$M$ to~$J$, where $n$ is the
  length of a complete path in~$J$.
\end{claimrep}

\begin{proof}
  The argument may be easier to follow graphically on an example by considering
  the coding shown in Figure~\ref{fig:itercode} and the iteration shown in
  Figure~\ref{fig:iter} (bottom).

  Consider the subinstance~$J$ and the witnessing complete path in~$G$, which we
  assume to be simple: $r = a_0, a_1, \ldots, a_n = s$.
  This means that the following edges have all of their covering facts in~$J$ and all
  their incident facts in~$J$ except possibly for the left-incident facts of the first edge: 
  $(u_{a_0}, v_{\pi_0,a_0})$, 
  $(u_{\pi_0}, v_{\pi_0,a_0})$,
  $(u_{\pi_0}, v_{\pi_0,a_1})$,
  $(u_{a_1}, v_{\pi_0,a_1})$,
  $\ldots$
  $(u_{a_{n-1}}, v_{\pi_{n-1},a_{n-1}})$, 
  $(u_{\pi_{n-1}}, v_{\pi_{n-1},a_{n-1}})$,
  $(u_{\pi_{n-1}}, v_{\pi_{n-1},a_n})$,
  $(u_{a_n}, v_{\pi_{n-1},a_n})$,
  where we let $\pi_0 = \{a_0, a_1\}$, $\ldots$, $\pi_{n-1} = \{a_{n-1},
  a_n\}$.

  Further, as $J$ is well-formed, all supplementary base facts are present.
  Thus, we know that the edge
  $(u_s, v_{\pi_0, u_s})$ in fact has all its left-incident facts, hence all its
  incident facts by considering the right-incident facts of the second edge in
  the sequence above, and that we
  can extend the sequence above with the edge $(u_s, v)$ which has all its
  covering facts and all its right-incident facts, hence all its incident facts
  by considering the left-incident facts of the last edge in the sequence above.

  The facts mentioned so far correspond to the facts involving the elements
  $u_i$ and $v_i$ in 
  the $(2n+1)$-iteration $I'$
  of~$M$, which satisfies the query because~$M$ is iterable and by
  Proposition~\ref{prp:iterability}. As for the other facts of $I'$, they
  correspond to base facts of~$I_G$, which are all present because $J$ is
  well-formed. Thus, indeed $J$ contains a subinstance which is isomorphic
  to~$I'$, concluding the proof.
\end{proof}

It is again far more challenging to show the other claim:

\begin{claimrep}
  \label{clm:bad3}
  Any bad subinstance $J \subseteq I_G$ does not satisfy the query.
\end{claimrep}

\begin{proofsketch}
  We dissociate all copies of~$e$ that are missing a fact or are of the form
  $(u_\beta, v_{\pi,\beta})$ and are missing an incident fact with some element
  $w_{i'}$. Then, we map the result by a homomorphism~$h$ to a
  structure called the \emph{explosion} (pictured in Figure~\ref{fig:other}, see
  Definition~\ref{def:explosion} in the appendix), which intuitively reflects
  all maximal strict subsets of the $\{t_1, \ldots, t_{\tau-1}\}$, and violates~$Q$ (by considering the lexicographic weight of its edges). We
  define~$h$
  along the cut of~$G$ defined by considering
  the vertices reachable from~$r$ via a complete path.
\end{proofsketch}

\begin{toappendix}
To show this claim, we will need to consider a new structure defined from the
  critical model~$M$, which we call the \emph{explosion}. Note that
  this structure is not symmetric, intuitively to account for the asymmetry in
  the definition of the lexicographic weight.

\begin{definition}
  \label{def:explosion}
  The \emph{explosion} of~$M$ is
  defined in the following way:
  \begin{itemize}
    \item Create a copy~$v'$ of~$v$ and $u'$ of~$u$
    \item For each $1 \leq j' < \tau$, copy the edges $(u, t_{j'})$
      on~$(u', t_{j'})$, and for each $1 \leq i' < \omega$, copy the edges $(w_{i'}, v)$
      on~$(w_{i'}, v')$.
    \item For every strict subset $\Sigma$ of~$\{t_1, \ldots, t_\tau\}$ (there may be
      none, in case $\tau = 0$), create a copy~$u_\Sigma$ of~$u$, and for each
      $t_j \in \Sigma$ copy the edge $(u, t_j)$ on~$(u_\Sigma, t_j)$.
    \item Copy~$e$ on~$(u, v')$ and $(u', v)$ and on $(u_\Sigma, v)$ and
      $(u_\Sigma, v')$ for each subset above $\Sigma$ (i.e., none if $\tau =
      0$).
    \item Remove the covering facts~$C$ of~$(u,v)$.
  \end{itemize}
\end{definition}

  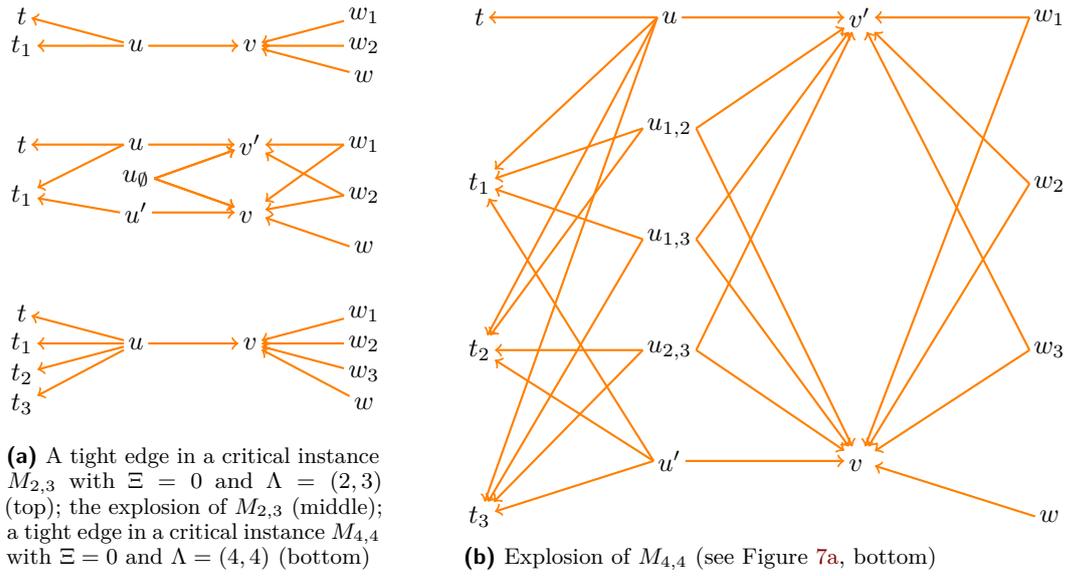
\begin{figure}
    \begin{subfigure}[b]{.35\linewidth}
    \begin{tikzpicture}[xscale=1.5,yscale=.4,inner sep=1.5]
    \node (u) at (0, 0) {$u$};
    \node (v) at (1, 0) {$v$};
    \node (t) at (-1, 1) {$t$};
    \node (t1) at (-1, 0) {$t_1$};
    \node (w1) at (2, 1) {$w_1$};
    \node (w2) at (2, 0) {$w_2$};
    \node (w) at (2, -1) {$w$};
    \draw[orange,->,thick] (u) -- (v);
    \draw[orange,->,thick] (u) -- (t);
    \draw[orange,->,thick] (u) -- (t1);
    \draw[orange,->,thick] (w) -- (v);
    \draw[orange,->,thick] (w1) -- (v);
    \draw[orange,->,thick] (w2) -- (v);
    \end{tikzpicture}
      \smallskip

    \begin{tikzpicture}[xscale=1.5,yscale=.45,inner sep=1.5]
    \node (u) at (0, 1) {$u$};
    \node (vp) at (1, 1) {$v'$};
      \node (u1) at (0, 0) {$u_\emptyset$};
    \node (up) at (0, -1) {$u'$};
      \node (v) at (1, -1) {$v\phantom{'}$};

    \node (t) at (-1, 1) {$t$};
    \node (t1) at (-1, -.5) {$t_1$};
    \node (w1) at (2, 1) {$w_1$};
    \node (w2) at (2, -.5) {$w_2$};
    \node (w) at (2, -2) {$w$};

    \draw[orange,->,thick]  (u.east) -- (vp);
    \draw[orange,->,thick] (u1.east) -- (vp);
    \draw[orange,->,thick] (u2.east) -- (vp);
    \draw[orange,->,thick] (u1.east) --  (v);
    \draw[orange,->,thick] (u2.east) --  (v);
    \draw[orange,->,thick] (up.east) --  (v);

    \draw[orange,->,thick] (u.west) -- (t);
    \draw[orange,->,thick] (u.west) -- (t1);
    \draw[orange,->,thick] (up.west) -- (t1);

    \draw[orange,->,thick] (w.west) -- (v);
    \draw[orange,->,thick] (w1.west) -- (v);
    \draw[orange,->,thick] (w2.west) -- (v);
    \draw[orange,->,thick] (w1.west) -- (vp);
    \draw[orange,->,thick] (w2.west) -- (vp);
      \end{tikzpicture}
      \smallskip

    \begin{tikzpicture}[xscale=1.5,yscale=.4,inner sep=1.5]
    \node (u) at (0, 0) {$u$};
    \node (v) at (1, 0) {$v$};
    \node (t) at (-1, 1) {$t$};
    \node (t1) at (-1, 0) {$t_1$};
    \node (t2) at (-1, -1) {$t_2$};
    \node (t3) at (-1, -2) {$t_3$};
    \node (w1) at (2, 1) {$w_1$};
    \node (w2) at (2, 0) {$w_2$};
    \node (w3) at (2, -1) {$w_3$};
    \node (w) at (2, -2) {$w$};
    \draw[orange,->,thick] (u) -- (v);
    \draw[orange,->,thick] (u) -- (t);
    \draw[orange,->,thick] (u) -- (t1);
    \draw[orange,->,thick] (u) -- (t2);
    \draw[orange,->,thick] (u) -- (t3);
    \draw[orange,->,thick] (w) -- (v);
    \draw[orange,->,thick] (w1) -- (v);
    \draw[orange,->,thick] (w2) -- (v);
    \draw[orange,->,thick] (w3) -- (v);
    \end{tikzpicture}
      \vspace{-.3cm}
      \caption{A tight edge in a critical instance $M_{2,3}$ with $\Xi = 0$ and
      $\Lambda = (2,3)$ (top); the explosion of~$M_{2,3}$ (middle); a tight edge
      in a critical instance $M_{4,4}$ with $\Xi = 0$ and $\Lambda = (4,4)$
      (bottom)}
      \label{fig:explosion1}
  \end{subfigure}
  \hfill
    \begin{subfigure}[b]{.57\linewidth}
    \begin{tikzpicture}[xscale=2.5,yscale=1.47,inner sep=1.5]
      \centering
    \node (u) at (0, 1) {$u$};
    \node (vp) at (1, 1) {$v'$};
      \node (u12) at (0, 0) {$u_{1,2}$};
      \node (u13) at (0, -1) {$u_{1,3}$};
      \node (u23) at (0, -2) {$u_{2,3}$};
    \node (up) at (0, -3) {$u'$};
      \node (v) at (1, -3) {$v\phantom{'}$};

    \node (t) at (-1, 1) {$t$};
    \node (t1) at (-1, -.5) {$t_1$};
    \node (t2) at (-1, -2) {$t_2$};
    \node (t3) at (-1, -3.5) {$t_3$};
    \node (w1) at (2, 1) {$w_1$};
    \node (w2) at (2, -.5) {$w_2$};
    \node (w3) at (2, -2) {$w_3$};
    \node (w) at (2, -3.5) {$w$};

    \draw[orange,->,thick]  (u.east) -- (vp);
    \draw[orange,->,thick] (u12.east) -- (vp);
    \draw[orange,->,thick] (u13.east) -- (vp);
    \draw[orange,->,thick] (u23.east) -- (vp);
    \draw[orange,->,thick] (u12.east) --  (v);
    \draw[orange,->,thick] (u13.east) --  (v);
    \draw[orange,->,thick] (u23.east) --  (v);
    \draw[orange,->,thick] (up.east) --  (v);

    \draw[orange,->,thick] (u.west) -- (t);
    \draw[orange,->,thick] (u.west) -- (t1);
    \draw[orange,->,thick] (u12.west) -- (t1);
    \draw[orange,->,thick] (u12.west) -- (t2);
    \draw[orange,->,thick] (u.west) -- (t2);
    \draw[orange,->,thick] (u13.west) -- (t1);
    \draw[orange,->,thick] (u13.west) -- (t3);
    \draw[orange,->,thick] (u.west) -- (t3);
    \draw[orange,->,thick] (u23.west) -- (t2);
    \draw[orange,->,thick] (u23.west) -- (t3);
    \draw[orange,->,thick] (up.west) -- (t1);
    \draw[orange,->,thick] (up.west) -- (t2);
    \draw[orange,->,thick] (up.west) -- (t3);

    \draw[orange,->,thick] (w.west) -- (v);
    \draw[orange,->,thick] (w1.west) -- (v);
    \draw[orange,->,thick] (w2.west) -- (v);
    \draw[orange,->,thick] (w3.west) -- (v);
    \draw[orange,->,thick] (w1.west) -- (vp);
    \draw[orange,->,thick] (w2.west) -- (vp);
    \draw[orange,->,thick] (w3.west) -- (vp);
      \end{tikzpicture}
      \caption{Explosion of $M_{4,4}$ (see Figure~\ref{fig:explosion1}, bottom)}
      \label{fig:expl44}
      \end{subfigure}
      \caption{Examples for the explosion (see also Figure~\ref{fig:other},
      bottom)}
      \label{fig:exploexa}
      \end{figure}

  See Figure~\ref{fig:other} (bottom) for an illustration.

  Note that when $\tau = 1$, i.e., $t$ is the only left copy element
  of~$e$ in~$I$, then the explosion is almost the same as the fine
  dissociation, except that the incomplete copies of the edges containing the
  facts $F_\LL$ and $F_\RR$ are not created -- but intuitively these do not matter
  because they are garbage facts.  (The reason why the fine dissociation create
  these edges, like iteration, but unlike $n$-step iteration and the explosion,
  is because the fine dissociation is used to show that $\Xi = 0$, i.e., before
  we know that $\Xi = 0$, and in the case where $\Xi>0$ these edges could make a
  difference.)

  When $\tau = 2$, an example is shown on Figure~\ref{fig:explosion1} (top and
  middle). When $\tau = 3$, an example is shown on Figure~\ref{fig:iter} (top)
  and Figure~\ref{fig:other} (bottom) in the main text. When $\tau = 4$, an
  example is shown on Figure~\ref{fig:iter} (bottom) and
  Figure~\ref{fig:expl44}. Remember that $\tau \leq \omega$ by definition of the
  critical lexicographic weight. Note how, as $\tau$ increases, the edges of the
  form $(u_\Sigma, v)$ and $(u_\Sigma, v')$ have a arbitrarily high number of
  right copy elements, but always $\tau-1$ left copy elements, so that they can
  be dissociated thanks to our definition of the lexicographic weight.

We claim the following, which is the only place where we use the full power of
  lexicographic minimality:

\begin{lemmarep}
  \label{lem:lopsided}
  The explosion of a critical model does not satisfy the query.
\end{lemmarep}

\begin{proof}
  Let us assume that it does and show a contradiction. If $\tau=0$ as we argued
  the explosion is a subset of the fine dissociation which does not satisfy the
  query by Claim~\ref{clm:finedissoc}, so we assume $\tau > 0$.

  Consider each edge of the form $(u_\Sigma, v)$. The edge
  has weight $\Theta$, the extra weight $\Xi$ is clearly zero, so let us compute
  the lexicographic weight by considering the left copy elements. (The reader
  may want to refer to examples, e.g., Figure~\ref{fig:exploexa}.)
  There are $\tau-2$ copy elements in the strict subset $\Sigma$, plus the
  element~$v'$. Hence, the number of left copy elements is~$\tau-1$, but
  $\Lambda = (\tau, \omega)$, so the lexicographic weight is strictly smaller.
  (Note that the number of right copy elements may be greater than $\omega$.)

  Hence, by
  Claim~\ref{clm:dissoclex}, we can dissociate these edges without breaking the
  query. The same argument shows that we can dissociate all edges $(u_\Sigma,
  v)$.

  Now, we can homomorphically map the result to the fine dissociation, by
  mapping the $u_\Sigma$ to $u$, mapping $u$ to~$u$ and $v$ to~$v$ and~$u'$
  to~$u'$ and~$v'$ to~$v'$, and mapping the dangling edges on~$v'$ to~$(u,v')$
  and on~$v$ to~$(u',v)$ and on the $u_\Sigma$ to~$(u,v)$. We know 
  by Claim~\ref{clm:finedissoc} that the fine dissociation does not satisfy the
  query, so we have reached a contradiction, which concludes the proof.
\end{proof}

We can now show:

\begin{proof}[Proof of Claim~\ref{clm:bad3}]
  Consider a bad subinstance $J$. Let us reason by contradiction, assume
  that $J$ satisfies the query, and rewrite it to an instance satisfying the
  query and having a homomorphism to the explosion, which is a
  contradiction by Lemma~\ref{lem:lopsided}.

  We will use the conventions of Appendix~\ref{apx:basic}. The proof is in
  several steps.  It may be helpful to informally comment on why the definition
  introduces the intermediate vertices $u_\pi$, which were not present in the
  coding of Section~7 of~\cite{amarilli2021dichotomy}. The reason is that these
  intermediate vertices have precisely two incident copies of~$e$ of the form
  $(u_\pi, v_{\pi,\beta})$ with $\beta \in \pi$, so their number of left copy
  elements is precisely~$\tau$ and we can dissociate them if their number of
  right copy elements become $< \omega$, i.e., if some right-incident fact is
  missing. By contrast, the vertices $u_a$ have more left copy elements, e.g.,
  $\tau+1$ if all vertices of the input undirected graph have degree~3.

  \bigskip

  The first step is to notice that copies (up to renaming) of the set $C$ of
  facts which are incomplete can be dissociated, because either they are empty
  or leaf and this is vacuous, or they are non-leaf and
  their weight is $< \Theta$ so we can use Claim~\ref{clm:dissocweight}. We call
  \emph{dangling edges} the edges created after these dissociations, and say
  they are \emph{attached} to the element which is non-leaf in them, if there is
  one. (The reason why there may be none is that some of the incomplete copies
  considered in this step may be such that both their elements are leaves, i.e., a copy up to renaming of~$C$ where both
  elements involved do not occur in any other fact. Then the ``dissociation'' of
  this edge (following the conventions in Appendix~\ref{apx:basic}) is
  isomorphic to it and is not attached to any element. However, these edges are not connected to
  anything else in the instance and can be mapped homomorphically to any copy
  of~$e$ in the explosion at the end of the process; so we simply ignore these copies
  in what follows.) Thus, letting~$J_2$ be the result of this process, in~$J_2$ relative to~$J$ all incomplete copies of~$C$ have
  been removed and replaced by dangling edges attached to some of the endpoints.

  The second step is to notice that whenever an element of the
  form~$v_{\pi,\beta}$
  is missing some right copy element in~$J_1$ which is not~$u_\beta$ (i.e., an edge of the form $(w_{i'},
  v_{\pi, \beta})$, which has been dissociated in the first step),
  then the edge $(u_\pi,v_{\pi,\beta})$ can be dissociated. 
  This is vacuous if it is not an edge or a leaf edge, and otherwise this is
  because its weight is~$\Theta$, its extra weight is $0$ (note that there are
  no triangles in the coding except possibly in the base facts),
  and its lexicographic weight can be accounted as follows:
  the element $u_\pi$ has at most $\tau-1$ left copy
  elements (namely, the $t_1, \ldots, t_{j'}$) and
  there is another left copy element namely~$v_{\pi,\beta'}$ for $\beta'$ the element
  such that $\pi = \{\beta, \beta'\}$,
  so $\tau$ left copy elements in total;
  and the element $v_{\pi,\beta}$ is missing some right copy element so is
  connected to at most $\omega-2$ elements, plus~$u_\beta$ because of the edge
  $(u_\beta, v_{\pi,\beta})$, so it is connected to at most $\omega-1$ elements.
  Thus, as $\Lambda = (\tau, \omega)$, we see that the lexicographic weight
  of~$(u_\pi,v_{\pi,\beta})$ is less than $\Lambda$, and by
  Claim~\ref{clm:dissoclex} we can indeed dissociate the edge $(u_\pi,
  v_{\pi,\beta})$. We let $J_3$ be
  the result of dissociating all edges that can be dissociated in this fashion.
  In $J_3$ relative to~$J_2$, in any sequence $u_a,
  v_{\pi,a}, u_\pi, v_{\pi,b}, u_b$ with $\pi = \{a, b\}$, if 
  one of $v_{\pi,a}$ or $v_{\pi,b}$ is missing a right-incident
  fact with one of the $w_{i'}$, then the corresponding edge $(u_\pi,v_{\pi,a})$, 
  or $(u_\pi,v_{\pi,b})$,
  has been dissociated.

  \bigskip

  Now we can define the homomorphism $h$ from~$J_3$ to the explosion. We
  initialize it to be the identity on~$u_r = u$ and on~$v$, and on all elements
  except the elements $u_a$ with $a \in V$ or $u_\pi$ with $\pi \in E$ or
  $v_{\pi,\beta}$ with $\pi \in E$ and $\beta \in \pi$.

  Next, we let $R$ be the set of nodes $v$ of the graph~$G$ such that there is a
  path in~$G$ from~$r$ to~$v$ which is complete in~$J$. For each $v \in R$, we map $u_v$ to~$u$,
  in particular we map $u_r$ to~$u$ as we said before.
  Note that, as $J$ is bad, then $s \notin
  R$, as this would otherwise witness the existence of a complete path from~$r$
  to~$s$ in~$J$, i.e., that~$J$ is good.
  Further, for each edge $\pi \in E$ between vertices of~$R$, i.e., $\pi
  \subseteq R$, we map $u_\pi$ to~$u$ and~$v_{\pi,\beta}$ to~$v'$ for each
  $\beta \in \pi$. All facts between these elements are clearly correctly
  mapped, because we have mapped copies of~$u$ to~$u$ and copies of~$v$
  different from~$v$ to~$v'$, and in the explosion there is a copy $(u,v')$
  of~$e$ and all needed incident facts.

  Now, we consider the set $R'$ of nodes of~$G$ that are adjacent to a node
  in~$R$ but which are not in~$R$. Note that $u_r \notin R'$, but
  potentially~$u_s \in R'$.
  There are several kinds of such nodes:
  \begin{itemize}
    \item The incomplete nodes, i.e., the vertices $a$ with some vertex facts missing. For
      such a node~$a$, remembering that in the first step we dissociate all edges that
      were incomplete copies of~$e$, we let $\Sigma$ be some maximal strict subset of $\{t_1,
      \ldots, t_{\tau-1}\}$ containing all the elements of this form which are
      left copy elements of $u_a$. We map $u_a$ to $u_\Sigma$ in the explosion, and for every edge
      $\pi = \{b, a\}$ with $b$ in~$R$ we map $u_{\pi}$ to~$u$ and $v_{\pi,a}$
      and $v_{\pi,b}$ to~$v'$.
    \item The complete nodes $a \in V$. We map these to~$u'$. For these, as they were
      not added to~$R$, we know that all edges $\pi = \{b,a\}$ with $b \in R$ were
      incomplete. There are four subcases for each such edge:
      \begin{itemize}
        \item $u_\pi$ is missing some left copy elements with one of the $t_1, \ldots,
          t_{\tau-1}$. In this case, remembering that incomplete edges were
          dissociated in step 1, letting $\Sigma$ be a maximal strict subset
          of $\{t_1, \ldots, t_{\tau-1}\}$ containing all the elements of this
          from that are
          left copy elements of~$u_\pi$ in~$J_3$, we map $u_\pi$ to~$u_\Sigma$ in the
          explosion, and map $v_{\pi,b}$ to~$v'$ and $v_{\pi,a}$ to~$v$.
        \item $v_{\pi,a}$ is missing some right copy elements with one of the
          $w_1, \ldots, w_{\omega-1}$. In this case, in the second step we
          dissociated the edge $(u_\pi, v_{\pi, a})$, so we can map $u_\pi$
          to~$u$ and $v_{\pi,b}$ to~$v'$ and $v_{\pi,a}$ to~$v$. Note that there
          is no edge $(u, v')$ in the explosion, so we really use the fact that the
          second step dissociated the edge.
        \item $v_{\pi,b}$ is missing some right copy elements. In this case,
          similarly to the previous case, we
          map $u_\pi$ to $u'$ and~$v_{\pi,b}$ to~$v'$ and $v_{\pi,a}$ to~$v$.
        \item Some fact in the edges $(u_b, v_{\pi,b})$ or $(u_\pi,v_{\pi,a})$
          or $(u_\pi,v_{\pi,a})$ or $(u_a, v_{\pi,a})$ is missing, so these
          edges were dissociated in the first step. In the case of an edge
          incident to~$u_\pi$, we conclude like in the two previous bullet
          points. If it is the edge $(u_b, v_{\pi,b})$, we map $v_{\pi,b}$ and
          $v_{\pi,a}$ to~$v$ and $u_\pi$ to~$u'$. If it is the edge $(u_a,
          v_{\pi,a})$, we map $v_{\pi,b}$ and $v_{\pi,a}$ to~$v'$ and $u_\pi$
          to~$u$.
      \end{itemize}
  \end{itemize}

  The homomorphism that we define correctly maps all facts of~$J_3$
  that correspond in the coding to edges $\{b, a\}$ with $b \in R$ and $a \in R'$.
  Further, for edges $\pi = \{a, a'\}$ with $a, a' \in R'$, we know that $u_a$ and
  $u_{a'}$ were mapped either to $u'$ or to $u_\Sigma$, so we can map $u_{\pi}$
  to~$u'$ and $v_{\pi,a}$ and $v_{\pi,b}$ to~$v$.

  Then, for the vertices $a$ in $V \setminus (R \cup R')$ and the edges $\pi$ involving
  them, we simply map $u_a$ and $u_\pi$ to~$u'$ and $v_{\pi,\beta}$ for $\beta
  \in \pi$ to~$v$ in this case, which is correct because these edges do not
  involve any vertex~$b$ whose element~$u_b$ was mapped to~$u$ by~$h$.

  Last, the dangling edges are mapped without difficulty, as they are copies of
  subsets of the covering facts of~$e$ (up to renaming), and the vertices to
  which they are attached are mapped to vertices having such an incident copy
  of~$e$ oriented in the right way.

  The important points to check for the correctness of the homomorphism are the
  following:
  \begin{itemize}
    \item The source $u_r$ was mapped to~$u$.
    \item The sink $u_s$ was mapped to $u'$ or some $u_\Sigma$, so the edge $(u_s, v)$
  is correctly mapped.
    \item Whenever a vertex is mapped to some~$u_\Sigma$ then its left-adjacent
      copy facts are contained in the subset~$\Sigma$
    \item No two elements $u_a$ and $u_{\{a,b\}}$ are mapped one to~$u$ and the
      other one to~$u'$ unless one of the edges $(u_a, v_{\{a,b\},a})$ and $(u_{\{a,b\}},
      v_{\{a,b\},b})$ has been dissociated.
  \end{itemize}

  So we have defined a homomorphism to the explosion and shown that the
  explosion satisfies the query, which is a contradiction by
  Lemma~\ref{lem:lopsided}. This concludes.
\end{proof}

This allows us to complete our reduction.

  \begin{proof}[Proof of Proposition~\ref{prp:ustcon}]
    We have fixed the critical model~$M$.
  We reduce from the problem $\phi,\eta$-U-ST-CON with $\phi =
  2^{-\Theta\times (\tau-1)}$,
   and $\eta = 2^{-\Theta(4 + (\tau-1) + 2\omega)}$,
    with $\Theta$ the critical weight and $\Lambda = (\tau, \omega)$ the
    critical side weight.
    Note that we can have $\phi=1$ if $\tau=1$, but that we always have $0 \leq \phi \leq 1$ and
    $0 \leq \eta < 1$.
  We assume without loss of generality that the
  source vertex $r$ is kept.

    Given $G$, we code it to the instance~$I_G$, which is in linear time in~$G$ (remember
    that~$I$ is fixed).
  We know by Claim~\ref{clm:ill3} that the
  subinstances missing a base fact or supplementary base fact do not satisfy the
    query, so
    we can simply study the well-formed subinstances.

    Now, the probability that a vertex $u_a$ with $a \neq r$ is complete
    is~$\phi$, the probability that an edge is complete is~$\eta$, and these
    events are independent and in correspondence with the vertices of~$V
    \setminus \{r\}$ and edges of~$E$ by a probability-preserving bijection,
    such that a subset $(V',E')$
    containing~$r$ features a path of kept vertices and edges
    connecting~$r$ and~$s$ in~$G$ iff the
    corresponding subinstance is good. We know by Claim~\ref{clm:good3}
    and~\ref{clm:bad3} that, on the well-formed subinstances, the query holds
    precisely on the good ones, so the result of uniform reliability for~$Q$
    on~$I_G$ is precisely the answer to $\phi,\eta$-U-ST-CON. This establishes
    the correctness of the reduction and concludes.
\end{proof}
\end{toappendix}

We then show hardness by reducing from $\phi,\eta$-U-ST-CON
for well-chosen constant probabilities~$\phi$ and~$\eta$
(up to assuming that the
source vertex~$r$ is always kept)
and thus conclude the reduction, establishing
Proposition~\ref{prp:ustcon}. Together with Proposition~\ref{prp:pp2dnf2}, as $Q$
has a critical model by Proposition~\ref{prp:critical} and
Theorem~\ref{thm:tight}, we have shown our main result (Theorem~\ref{thm:main}).

\section{Conclusion}
\label{sec:conc}
We have proved the intractability of uniform reliability for unbounded
homomorphism-closed queries on arity-two signatures.
We have not investigated the related problem of \emph{weighted uniform
reliability}~\cite{amarilli2022uniform}, which is the restricted case of
probabilistic query evaluation where we impose that
all facts of the input TID must have some fixed probability  different from~$1/2$.
We expect that our hardness
result should extend to this problem when the fixed probability is the same
across all relations (and is different from~$0$ and~$1$).
It seems more challenging to understand the setting where the fixed probability can
depend on the relation, in particular if we can require some relations to be 
be deterministic, i.e., only have tuples with probability~$1$.
In this setting, some unbounded homomorphism-closed queries would
become tractable (e.g., Datalog queries that involve only the deterministic
relations), and it is not clear what one can hope to show.

Coming back to the problem of (non-weighted) uniform reliability, 
an ambitious direction for future
work would be to extend our intractability result towards Conjecture~\ref{con:goal}.
The two remaining obstacles are the case of unbounded queries on arbitrary
signatures, which we intend to study in future work; and the case of bounded queries,
i.e., UCQs, where the general case is left open by Kenig and
Suciu~\cite{kenig2020dichotomy}.

Other natural extensions include the study of
queries satisfying weaker requirements than closure under homomorphisms; or other
notions of possible worlds, e.g., induced subinstances; or other notions
of intractability, e.g., the inexistence of lineages in tractable circuit
classes from knowledge compilation. Another broad question is whether the
techniques developed here have any connection to other areas of research, e.g.,
constraint satisfaction problems (CSPs).

\bibliography{main}

\begin{thebibliography}{10}

\bibitem{amarilli2023uniform}
Antoine Amarilli.
\newblock Uniform reliability for unbounded homomorphism-closed graph queries.
\newblock Full version with proofs, 2023.
\newblock URL: \url{https://arxiv.org/abs/2209.11177}.

\bibitem{amarilli2021dichotomy}
Antoine Amarilli and İsmail~İlkan Ceylan.
\newblock \href{https://arxiv.org/abs/1910.02048} {The dichotomy of evaluating
  homomorphism-closed queries on probabilistic graphs}.
\newblock {\em \mbox{\href{https://lmcs.episciences.org/}{LMCS}}}, 2021.
\newblock \href {https://doi.org/10.46298/lmcs-18(1:2)2022}
  {\path{doi:10.46298/lmcs-18(1:2)2022}}.

\bibitem{amarilli2022uniform}
Antoine Amarilli and Benny Kimelfeld.
\newblock \href{https://arxiv.org/abs/1908.07093} {Uniform reliability of
  self-join-free conjunctive queries}.
\newblock {\em LMCS}, 18(4), 2022.
\newblock \href {https://doi.org/10.46298/lmcs-18(4:3)2022}
  {\path{doi:10.46298/lmcs-18(4:3)2022}}.

\bibitem{Cook71}
Stephen~A. Cook.
\newblock \href{https://www.cs.toronto.edu/~sacook/homepage/1971.pdf}{The
  complexity of theorem-proving procedures}.
\newblock In {\em Proc.\ {STOC}}, 1971.

\bibitem{dalvi2012dichotomy}
Nilesh~N. Dalvi and Dan Suciu.
\newblock \href{https://homes.cs.washington.edu/~suciu/jacm-dichotomy.pdf}{The
  dichotomy of probabilistic inference for unions of conjunctive queries}.
\newblock {\em Journal of the ACM}, 59(6):30, 2012.

\bibitem{FiOl16}
Robert Fink and Dan Olteanu.
\newblock
  \href{http://www.cs.ox.ac.uk/people/Dan.Olteanu/papers/fo-tods16.pdf}{Dichotomies
  for queries with negation in probabilistic databases}.
\newblock {\em {TODS}}, 41(1), 2016.

\bibitem{graedel1998complexity}
Erich Gr{\"{a}}del, Yuri Gurevich, and Colin Hirsch.
\newblock
  \href{https://www.researchgate.net/profile/Yuri_Gurevich2/publication/2900852_The_Complexity_of_Query_Reliability/links/0c96053321102376cd000000/The-Complexity-of-Query-Reliability.pdf}{The
  complexity of query reliability}.
\newblock In {\em Proc.\ {PODS}}, 1998.

\bibitem{kenig2020dichotomy}
Batya Kenig and Dan Suciu.
\newblock \href{https://arxiv.org/abs/2008.00896}{A dichotomy for the
  generalized model counting problem for unions of conjunctive queries}.
\newblock In {\em Proc.\ {PODS}}, 2021.

\bibitem{livshits2020shapley}
Ester Livshits, Leopoldo~E. Bertossi, Benny Kimelfeld, and Moshe Sebag.
\newblock \href{https://drops.dagstuhl.de/opus/volltexte/2020/11944/}{The
  {S}hapley value of tuples in query answering}.
\newblock In {\em Proc.\ {ICDT}}, volume 155, 2020.

\bibitem{OlHu08}
Dan Olteanu and Jiewen Huang.
\newblock
  \href{https://www.cs.ox.ac.uk/people/dan.olteanu/papers/oh-sum08.pdf}{Using
  {OBDD}s for efficient query evaluation on probabilistic databases}.
\newblock In {\em Proc.\ {SUM}}, volume 5291, 2008.

\bibitem{OlHu09}
Dan Olteanu and Jiewen Huang.
\newblock
  \href{https://www.cs.ox.ac.uk/people/dan.olteanu/papers/oh-sigmod09.pdf}{Secondary-storage
  confidence computation for conjunctive queries with inequalities}.
\newblock In {\em Proc.\ {SIGMOD}}, 2009.

\bibitem{provan1983complexity}
J.~Scott Provan and Michael~O. Ball.
\newblock The complexity of counting cuts and of computing the probability that
  a graph is connected.
\newblock {\em SIAM Journal on Computing}, 12(4), 1983.

\bibitem{salimi2016quantifying}
Babak Salimi, Leopoldo~E. Bertossi, Dan Suciu, and Guy~Van den Broeck.
\newblock \href{https://arxiv.org/abs/1603.02705}{Quantifying causal effects on
  query answering in databases}.
\newblock In {\em TAPP}, 2016.

\bibitem{suciu2011probabilistic}
Dan Suciu, Dan Olteanu, Christopher R{\'{e}}, and Christoph Koch.
\newblock {\em Probabilistic databases}.
\newblock Synthesis Lectures on Data Management. Morgan {\&} Claypool
  Publishers, 2011.

\bibitem{valiant1979complexity}
Leslie~Gabriel Valiant.
\newblock
  \href{https://www.sciencedirect.com/science/article/pii/0304397579900446}{The
  complexity of computing the permanent}.
\newblock {\em TCS}, 8(2):189--201, 1979.

\end{thebibliography}

\end{document}